\documentclass[pra,aps,10pt,twocolumn,showkeys]{revtex4-2}
\usepackage{amsfonts}
\usepackage{amsmath}
\usepackage{amssymb}
\usepackage{graphicx}
\usepackage{mathtools}
\usepackage[colorlinks=true,linkcolor=blue,citecolor=red,plainpages=false,pdfpagelabels]
{hyperref}
\usepackage{algorithm}
\usepackage{algorithmic}
\usepackage{xfrac}
\usepackage{quantikz}
\usepackage{xcolor}

\providecommand{\U}[1]{\protect\rule{.1in}{.1in}}

\newtheorem{theorem}{Theorem}

\newenvironment{proof}[1][Proof]{\noindent\textbf{#1.} }{\ \rule{0.5em}{0.5em}}

\allowdisplaybreaks

\begin{document}
\preprint{ }
\title[Natural gradient and parameter estimation for quantum Boltzmann machines]{Natural gradient and parameter estimation for quantum Boltzmann machines}
\author{Dhrumil Patel}
\affiliation{Department of Computer Science, Cornell University, Ithaca, New York 14850, USA}
\author{Mark M. Wilde}
\affiliation{School of Electrical and Computer Engineering, Cornell University, Ithaca, New
York 14850, USA}
\keywords{natural gradient, quantum Boltzmann machines, parameterized thermal states,
Fisher--Bures information matrix, Kubo--Mori information matrix, Hamiltonian parameter estimation}

\begin{abstract}
Thermal states play a fundamental role in various areas of physics, and they are becoming increasingly important in quantum information science, with applications related to semi-definite programming, quantum Boltzmann machine learning, Hamiltonian learning, and the related task of estimating the parameters of a Hamiltonian. Here we establish formulas underlying the basic geometry of parameterized thermal states, and we delineate quantum algorithms for estimating the values of these formulas. More specifically, we establish formulas for the Fisher--Bures and Kubo--Mori information matrices of parameterized thermal states, and our quantum algorithms for estimating their matrix elements involve a
combination of classical sampling, Hamiltonian
simulation, and the Hadamard
test. These results have applications in developing a natural gradient descent algorithm for quantum Boltzmann machine learning, which takes into account the geometry of thermal states, and in establishing fundamental limitations on the ability to estimate the parameters of a Hamiltonian, when given access to thermal-state samples. For the latter task, and for the special case of estimating a single parameter, we sketch an algorithm that realizes a measurement that is asymptotically optimal for the estimation task. We finally stress that the natural gradient descent algorithm developed here can be used for any machine learning problem that employs the quantum Boltzmann machine ansatz.

\end{abstract}
\date{\today}
\startpage{1}
\endpage{10}
\maketitle
\tableofcontents

\section{Introduction}

\subsection{Background}

Quantum states at thermal equilibrium, known as thermal states, play an
essential role in many areas of physics, including condensed
matter~\cite{Continentino2021}, high energy, and quantum
chemistry~\cite{DSL15}. They arise from a typical physical process by which a
system of interest is weakly coupled to an external environment at a finite
temperature~\cite[Appendix~A-1]{Alhambra2023}. Due to their ubiquitious
occurrence in nature, they can be used to predict various properties of
physical systems. They are also uniquely characterized as being states of
maximal entropy subject to an energy constraint~\cite{Jaynes1957,Jaynes1957a},
the states that minimize the free energy~\cite[Eqs.~(2)--(3)]{Giudice2021},
and completely passive states, whereby no energy can be extracted from them,
even when taking multiple copies~\cite[Section~1.2.2]{Binder2018}.

In more recent years, thermal states have intersected quantum information
science in various intriguing ways. First, quantum algorithms for semidefinite
programming make use of them in a fundamental way~\cite{Brandao2017} (see
also~\cite{QASDP19,AGGW17}). As a quantum generalization of classical
Boltzmann machines, the model of quantum Boltzmann machines has emerged as a
method of performing machine learning tasks on a quantum
computer~\cite{Amin2018,Benedetti2017,Kieferova2017}, with applications to
generative modeling~\cite{Coopmans2024,tuysuz2024learninggeneratehighdimensionaldistributions} and estimating ground-state
energies~\cite{patel2024quantumboltzmann}. Finally, there has been interest in
performing measurements on thermal states in order to learn the underlying
Hamiltonian~\cite{Anshu2021,bakshi2023learning,Rouze2024} or, related to this
goal, to estimate the parameters that uniquely identify a
Hamiltonian~\cite{GarciaPintos2024,abiuso2024fundamentallimitsmetrologythermal}. All of these applications, in addition to
their fundamental role in physics, have motivated the development of quantum
algorithms for preparing thermal
states~\cite{chen2023quantumthermalstatepreparation,chen2023efficient,bergamaschi2024,chen2024randomizedmethodsimulatinglindblad,rajakumar2024gibbssamplinggivesquantum,rouze2024efficient,bakshi2024hightemperature,ding2024polynomial}.

\subsection{Main result}

In this paper, our fundamental contribution is to establish simple formulas
underlying the basic geometry of parameterized thermal states and to construct
efficient quantum algorithms for estimating these formulas. We then apply
these findings in two different ways: 1) we establish a method for quantum
Boltzmann machine learning called natural gradient descent and 2) we prove
fundamental limitations on how well one can estimate a Hamiltonian given
access to thermal-state samples, in the sense of the Cramer--Rao bound (see \eqref{eq:Cramer--Rao-multiple} and \eqref{eq:scalar-Cramer--Rao-from-matrix-one} in particular). Here, we consider general parameterized
Hamiltonians of the form
\begin{equation}
G(\theta)\coloneqq\sum_{j=1}^{J}\theta_{j}G_{j}, \label{eq:param-Ham}
\end{equation}
where $\theta \coloneqq (\theta_1, \ldots, \theta_J)$ is a parameter vector, $\theta_{j}\in\mathbb{R}$ and $G_{j}$ is a fixed Hamiltonian, for all
$j\in\{1,\ldots,J\}$. Such Hamiltonians lead to thermal states of the
following form:
\begin{equation}
\rho(\theta)\coloneqq\frac{e^{-G(\theta)}}{Z(\theta)},
\label{eq:thermal-state}
\end{equation}
where the partition function $Z(\theta)$ is defined as
\begin{equation}
Z(\theta
)\coloneqq\operatorname{Tr}[e^{-G(\theta)}].
\end{equation}
One of our primary concerns is related
to the geometry of such states, as induced by the Fisher--Bures and Kubo--Mori
metrics~\cite{Jarzyna2020}, and when doing so, one is naturally led to the
notions of the Fisher--Bures and Kubo--Mori information matrices. As some of
our main technical contributions, we establish elegant formulas for the
elements of the Fisher--Bures and Kubo--Mori information matrices. Of
practical relevance, we prove that their matrix elements can be efficiently
estimated on a quantum computer, by means of quantum algorithms that we call
the \textit{quantum Boltzmann--Fisher--Bures estimator} and the
\textit{quantum Boltzmann--Kubo--Mori estimator}, respectively. Let us note here the
striking similarity of these algorithms to the recently established quantum
Boltzmann gradient estimator~\cite{patel2024quantumboltzmann}.

\subsection{Applications}

\subsubsection{Natural gradient descent for quantum Boltzmann machine
learning}

Our first application is to develop a variant of the standard
gradient descent algorithm called natural gradient descent, which specifically
applies to quantum Boltzmann machine learning. Importantly, the gradient
update step in our approach can be performed efficiently on a quantum
computer. We should note that, as in~\cite{Coopmans2024,patel2024quantumboltzmann}, our algorithms for natural gradient descent assume the ability to sample from the thermal state in \eqref{eq:thermal-state} for every value of the parameter vector $\theta$. So the performance of these algorithms depends on that of thermal-state preparation algorithms~\cite{chen2023quantumthermalstatepreparation,chen2023efficient,bergamaschi2024,chen2024randomizedmethodsimulatinglindblad,rajakumar2024gibbssamplinggivesquantum,rouze2024efficient,bakshi2024hightemperature,ding2024polynomial}.

Natural gradient was introduced to classical machine
learning~\cite{amari1998natural} in order to account for the geometry induced
by parameterized probability distributions used in generative modeling. The
idea behind the method is that, by incorporating the natural geometry of the
problem into the search for a minimum value of the loss function, one can
navigate the optimization landscape more judiciously. In contrast, the
standard gradient descent algorithm assumes a Euclidean geometry, which is
poorly adapted to the weight space of neural networks due to their intrinsic
parameter redundancy~\cite{neyshabur2015path}.

Inspired by the usage of
natural gradient for classical machine learning, quantum scientists introduced
this method to quantum machine learning implemented with parameterized quantum
circuits (PQCs)~\cite{Stokes2020quantumnatural}. This was certainly an
interesting theoretical development, and evidence was given that the method
works well on particular problem instances. However, there is also evidence
that the barren-plateau problem, originally observed for PQCs using standard
gradient
descent~\cite{McClean_2018,marrero2021entanglement,arrasmith2022equivalence,holmes2022connecting,Fontana2024,Ragone2024}, persists even when using PQCs supplemented by natural
gradient~\cite{Haug2021}. In contrast, due to the fact that there is evidence
that quantum Boltzmann machines do not suffer from the barren-plateau
problem in some contexts~\cite{Coopmans2024}, there is the hope that quantum Boltzmann machine
learning supplemented by the natural gradient method developed here will offer
a compelling and general approach to machine learning using quantum computers.
Let us also note here that natural gradient has been developed for quantum
Hamiltonian-based models (distinct from quantum Boltzmann machines), with
numerical evidence given regarding its performance for various machine learning
tasks~\cite{sbahi2022provablyefficientvariationalgenerative}. The appendices of~\cite{sbahi2022provablyefficientvariationalgenerative} also serve as a notable and comprehensive review of the geometry of quantum states.

\subsubsection{Estimation of Hamiltonian parameters from thermal~states}

Our second application is to the problem of estimating a Hamiltonian when
given access to thermal-state samples. In this context, the following
multiparameter Cramer--Rao bound holds for an arbitrary unbiased estimator and
for a general parameterized family $(\sigma(\theta))_{\theta\in\mathbb{R}^{J}
}$ of states:
\begin{equation}
\operatorname{Cov}^{(n)}(\hat{\theta},\theta)\geq\frac{1}{n}\left[  I^{\operatorname{FB}
}(\theta)\right]  ^{-1},\label{eq:Cramer--Rao-multiple}
\end{equation}
where $n\in\mathbb{N}$ is the number of copies of the state $\sigma(\theta)$
available, $\hat{\theta}$ is an estimate of the parameter vector~$\theta$, the matrix $I^{\operatorname{FB}}(\theta)$ denotes the Fisher--Bures
information matrix (defined later in~\eqref{eq:qfim-explicit}--\eqref{eq:qfim-explicit-basis-ind}), and the covariance matrix
$\operatorname{Cov}^{(n)}(\hat{\theta},\theta)$ measures errors in estimation and is
defined in terms of its matrix elements as
\begin{multline}
\lbrack\operatorname{Cov}^{(n)}(\hat{\theta},\theta)]_{k,\ell}\coloneqq\\
\sum_{m}\operatorname{Tr}[M_{m}^{(n)}\sigma(\theta)^{\otimes n}](\hat{\theta
}_{k}(m)-\theta_{k})(\hat{\theta}_{\ell}(m)-\theta_{\ell}).
\label{eq:cov-mat-def}
\end{multline}
In the above, $(M_{m}^{(n)})_{m}$ is an arbitrary positive operator-value measure used for estimation (it
satisfies $M_{m}^{(n)}\geq0$ for all $m$ and $\sum_{m}M_{m}^{(n)}=I^{\otimes n}$). Note
that this measurement acts, in general, collectively on all $n$ copies of the
state $\sigma(\theta)^{\otimes n}$. Additionally,
\begin{equation}
\hat{\theta}(m)\coloneqq(\hat{\theta}_{1}(m),\hat{\theta}_{2}(m),\ldots
,\hat{\theta}_{J}(m))
\label{eq:parameter-estimate-function}
\end{equation}
is a function that maps the measurement outcome $m$ to an estimate
$\hat{\theta}(m)$ of the parameter vector $\theta$. See~\cite{Bengtsson2006,Liu2019,Sidhu2020,Jarzyna2020,Meyer2021fisherinformationin,sbahi2022provablyefficientvariationalgenerative,scandi2024quantumfisherinformationdynamical} for various reviews of quantum Fisher information, geometry of quantum states, and multiparameter estimation.

We stress
that~\eqref{eq:Cramer--Rao-multiple} is an operator inequality, meaning that it
is equivalent to all of the eigenvalues of the operator $\operatorname{Cov}
^{(n)}(\hat{\theta},\theta)-\frac{1}{n}\left[  I^{\operatorname{FB}
}(\theta)\right]^{-1}$ being non-negative. The inequality in
\eqref{eq:Cramer--Rao-multiple} was derived for the case $n=1$
in~\cite{Helstrom1968}, and the case of general $n\in\mathbb{N}$ follows from
the additivity of the Fisher--Bures information matrix with respect to
tensor-product states. The latter claim follows from a calculation
generalizing that in~\cite[Appendix~D]{Katariya2021}, and we provide it for convenience in
Appendix~\ref{app:add-FB-info-mat}. Let us finally note that the following operator inequalities
hold~\cite[Eq.~(60)]{Jarzyna2020}:
\begin{equation}
I^{\operatorname{KM}}(\theta
)\geq I^{\operatorname{FB}}(\theta) \geq 0,\label{eq:FB-KM-op-ineq}
\end{equation}
where $I^{\operatorname{KM}}(\theta)$ denotes the Kubo--Mori information
matrix (defined later in~\eqref{eq:KM-inf-matrix-explicit}--\eqref{eq:KM-inf-basis-ind}). This inequality implies that $\left[  I^{\operatorname{FB}}(\theta)\right]
^{-1}\geq\left[  I^{\operatorname{KM}}(\theta)\right]  ^{-1}$, which in turn
implies that $I^{\operatorname{FB}}(\theta)$ gives a better bound on $\operatorname{Cov}
^{(n)}(\hat{\theta},\theta)$ than does
$I^{\operatorname{KM}}(\theta)$. As such, we do not consider the matrix
$I^{\operatorname{KM}}(\theta)$ in the quantum multiparameter estimation application.

One can formulate a scalar performance metric
from~\eqref{eq:Cramer--Rao-multiple} by setting $W$ to be a positive
semi-definite \textquotedblleft weight\textquotedblright\ matrix and taking
the trace of both sides of~\eqref{eq:Cramer--Rao-multiple} with respect to~$W$,
which leads to
\begin{equation}
\operatorname{Tr}[W\operatorname{Cov}^{(n)}(\hat{\theta},\theta)]\geq\frac{1}
{n}\operatorname{Tr}[W\ I^{\operatorname{FB}}(\theta)^{-1}].
\label{eq:scalar-Cramer--Rao-from-matrix-one}
\end{equation}
The left-hand side of~\eqref{eq:scalar-Cramer--Rao-from-matrix-one} should
indeed be understood as a performance metric, which weights errors in
different ways, so that the inequality
in~\eqref{eq:scalar-Cramer--Rao-from-matrix-one} provides a fundamental
limitation on the performance of any possible quantum multiparameter
estimation scheme. The latter conclusion holds because the right-hand side of
\eqref{eq:scalar-Cramer--Rao-from-matrix-one} has no dependence on the
particular scheme being used for estimation and instead only depends on the
parameterized family $(\sigma(\theta))_{\theta\in\mathbb{R}^{J}}$.

The implication of our technical contribution to Hamiltonian parameter estimation is as follows: Our exact
formulas for the Fisher--Bures information matrix of the thermal-state
family in~\eqref{eq:thermal-state} provide fundamental limitations on the
performance of any possible unbiased scheme for estimating the parameter vector
$\theta$ from performing measurements on the thermal state $\rho(\theta)$. Additionally, for the special case of estimating a single parameter, we sketch an algorithm that realizes a measurement that is asymptotically optimal for the estimation task.

\subsection{Paper organization}

In the rest of the paper, we provide more background on the geometry of
quantum states induced by the Fisher--Bures and Kubo--Mori metrics, as well as
the related Fisher--Bures and Kubo--Mori information matrices (Section~\ref{sec:monotone-metrics-review}). After
that, we provide our formulas for the elements of the Fisher--Bures
and Kubo--Mori information matrices of parameterized thermal states of the
form in~\eqref{eq:thermal-state} (Section~\ref{sec:formulas-FB-KM-info-mats}), only sketching the idea behind their proofs
in the main text while providing detailed proofs in the appendix. We then
introduce our quantum algorithms, the quantum Boltzmann--Fisher--Bures and
Boltzmann--Kubo--Mori estimators, for estimating the elements of the 
Fisher--Bures and Kubo--Mori information matrices (Section~\ref{sec:q-alg-NG-QBM}). We finally go into more
detail in our two applications: a natural gradient method for quantum
Boltzmann machine learning (Section~\ref{sec:NG-QBM}), and estimating the parameters of a Hamiltonian
from copies of a thermal state (Section~\ref{sec:est-ham-param}). We conclude in Section~\ref{sec:conclusion} with a summary and some directions for future work.

\section{Review of monotone metrics}

\label{sec:monotone-metrics-review}

\subsection{Fisher--Bures metric}

When formulating a distance measure between two quantum state vectors
$|\psi\rangle$ and $|\phi\rangle$, one might initially guess that the
Euclidean distance $\left\Vert |\psi\rangle-|\phi\rangle\right\Vert _{2}
$\ would be a reasonable choice. However, in quantum mechanics, this reasoning
is misguided because two state vectors $|\psi\rangle$ and $e^{i\varphi}
|\psi\rangle$ that differ only by a global phase are physically
indistinguishable. In spite of these states being physically
indistinguishable, the distance 
\begin{equation}
\left\Vert |\psi\rangle-e^{i\varphi}
|\psi\rangle\right\Vert _{2}=\sqrt{2\left(  1-\cos(\varphi)\right)  }
\end{equation}
is
generally non-zero. As a remedy to this problem, the Bures distance
incorporates a minimization over all possible global phases:
\begin{align}
d_{B}(\psi,\phi)  &  \coloneqq\min_{\varphi\in\left[  0,2\pi\right]
}\left\Vert |\psi\rangle-e^{i\varphi}|\phi\rangle\right\Vert _{2}\\
&  =\sqrt{2\left(  1-\left\vert \langle\phi|\psi\rangle\right\vert \right)  },
\end{align}
where we have employed the abbreviations $\psi\equiv|\psi\rangle\!\langle\psi|$
and $\phi\equiv|\phi\rangle\!\langle\phi|$ and the second equality follows from
basic reasoning. The overlap quantity $\left\vert \langle\phi|\psi
\rangle\right\vert $ is known as the root fidelity of the state vectors $|\psi\rangle$ and $|\phi\rangle$.

One can follow the above line of reasoning when generalizing this distance to
general states $\rho$ and $\sigma$ (described by density operators). In doing
so, one appeals to the purification principle of quantum mechanics, which
states that every density operator $\rho_{S}$ for a system $S$ can be thought
of as the marginal of a pure state $\psi_{RS}^{\rho}$ of a bipartite system
consisting of $R$ and $S$, so that $\operatorname{Tr}_{R}[\psi_{RS}^{\rho
}]=\rho_{S}$. Let $\psi_{RS}^{\sigma}$ be a purification of $\sigma_{S}$. It
is well known that two states $\rho_{S}$ and $\sigma_{S}$ are equal if and
only if there exists a unitary $U_{R}$ acting on the reference system $R$ such
that $|\psi^{\rho}\rangle_{RS}=\left(  U_{R}\otimes I_{S}\right)
|\psi^{\sigma}\rangle_{RS}$, where $|\psi^{\rho}\rangle_{RS}$ and
$|\psi^{\sigma}\rangle_{RS}$ are the state vectors corresponding to $\psi
_{RS}^{\rho}$ and $\psi_{RS}^{\sigma}$, respectively. In this way, the unitary
$U_{R}$ generalizes the global phase discussed in the previous paragraph.
Furthermore, this line of reasoning then leads to the Bures distance of two
general states $\rho$ and $\sigma$~\cite{bures1969extension}:
\begin{align}
d_{B}(\rho,\sigma)  &  \coloneqq\min_{U_{R}}\left\Vert |\psi^{\rho}
\rangle_{RS}-\left(  U_{R}\otimes I_{S}\right)  |\psi^{\sigma}\rangle
_{RS}\right\Vert _{2}\\
&  =\sqrt{2\left(  1-\sqrt{F}(\rho,\sigma)\right)  },
\end{align}
where the fidelity $F(\rho,\sigma)$ is defined as
\begin{equation}
F(\rho,\sigma
)\coloneqq\left\Vert \sqrt{\rho}\sqrt{\sigma}\right\Vert _{1}^{2}.
\end{equation}
The
equality above is a consequence of Uhlmann's theorem~\cite{Uhlmann1976}, which
states that
\begin{equation}
F(\rho,\sigma)=\max_{U_{R}}\left\vert \langle\psi^{\sigma}
|_{RS}\left(  U_{R}\otimes I_{S}\right)  |\psi^{\rho}\rangle_{RS}\right\vert
^{2}.
\end{equation}
The above line of reasoning, as well as its deep connections to quantum
estimation theory, underpins why the Bures distance is a natural metric to
consider on the space of quantum states.

Now consider a parameterized family $\left(  \sigma(\theta)\right)
_{\theta\in\mathbb{R}^{J}}$ of quantum states, and suppose for simplicity that
each state $\sigma(\theta)$ is positive definite (indeed, this is the case for
\eqref{eq:thermal-state}, which is the main example that we will consider later
on). Then it is well known that the infinitesimal squared line element between
two members of the parameterized family is~\cite[Eq.~(77)]{Liu2019}
\begin{equation}
d_{B}^{2}(\sigma(\theta),\sigma(\theta+d\theta))=\frac{1}{4}\sum_{i,j=1}
^{J}I_{ij}^{\operatorname{FB}}(\theta)\ d\theta_{i}\ d\theta_{j},
\label{eq:Bures-Fisher-metric}
\end{equation}
where $I_{ij}^{\operatorname{FB}}(\theta)$ is a matrix element of a Riemannian
metric tensor called the Fisher--Bures information matrix, given
explicitly by the following formulas~\cite[Eqs.~(126) and (142)]{Sidhu2020}:
\begin{align}
I_{ij}^{\operatorname{FB}}(\theta)  &  \coloneqq \sum_{k,\ell}\frac{2}{\lambda_{k}+\lambda_{\ell}} \langle
k|\partial_{i}\sigma(\theta)|\ell\rangle\!\langle\ell|\partial_{j}\sigma
(\theta)|k\rangle\label{eq:qfim-explicit}\\
&  =2\int_{0}^{\infty}ds\ \operatorname{Tr}[\left(  \partial_{i}\sigma
(\theta)\right)  e^{-s\sigma(\theta)}\left(  \partial_{j}\sigma(\theta
)\right)  e^{-s\sigma(\theta)}]. \label{eq:qfim-explicit-basis-ind}
\end{align}
In the above, we use the shorthand
\begin{equation}
    \partial_{i}\equiv\partial_{\theta_{i}}
\end{equation} and a
spectral decomposition $\rho(\theta)=\sum_{k}\lambda_{k}|k\rangle\!\langle k|$,
where in the latter notation we have suppressed the dependence of the
eigenvalue $\lambda_{k}$ and the eigenvector $|k\rangle$ on the parameter
vector $\theta$. The second formula in~\eqref{eq:qfim-explicit-basis-ind} is
useful as it is basis independent.

\subsection{Kubo--Mori metric}

An alternative metric that one can consider is that induced by the quantum
relative entropy, which is defined for states $\omega$ and $\tau$ as
\cite{umegaki1962ConditionalExpectationOperator}
\begin{equation}
D(\omega\Vert\tau)\coloneqq\operatorname{Tr}[\omega(\ln\omega-\ln\tau)]
\label{eq:rel-ent-def}
\end{equation}
if the support of $\omega$ is contained in the support of $\tau$ and
$D(\omega\Vert\tau)\coloneqq+\infty$ otherwise. The quantum relative entropy
is a fundamental distinguishability measure in quantum information theory, due
to it obeying the data-processing inequality~\cite{Lindblad1975}\ and  having an operational meaning in the context of quantum hypothesis testing
\cite{hiai1991ProperFormulaRelative,nagaoka2000StrongConverseSteins}.

Due to the support condition and its related singularity, the metric induced by
the quantum relative entropy is not suitable for parameterized families of
pure states. However, for our case of interest here, the parameterized thermal
states in~\eqref{eq:thermal-state}, this is not problematic because all
thermal states in~\eqref{eq:thermal-state} are positive definite.

Following reasoning similar to that leading to~\eqref{eq:Bures-Fisher-metric},
the infinitesimal line element between two members of the parameterized family
$\left(  \sigma(\theta)\right)  _{\theta\in\mathbb{R}^{J}}$ is
\begin{equation}
D(\sigma(\theta)\Vert\sigma(\theta+d\theta))=\frac{1}{2}\sum_{i,j=1}^{J}
I_{ij}^{\operatorname{KM}}(\theta)\ d\theta_{i}\ d\theta_{j},
\label{eq:KM-metric-def}
\end{equation}
where $I_{ij}^{\operatorname{KM}}(\theta)$ is a matrix element of a
Riemannian metric tensor called the  Kubo--Mori information matrix,
given explicitly by the following formulas \cite[Eqs.~(B9), (B12), (B22), (B23)]{sbahi2022provablyefficientvariationalgenerative}:
\begin{align}
I_{ij}^{\operatorname{KM}}(\theta) &  \coloneqq\sum_{k,\ell}
c_{\operatorname{KM}}(\lambda_{k},\lambda_{\ell})\langle k|\partial_{i}
\sigma(\theta)|\ell\rangle\!\langle\ell|\partial_{j}\sigma(\theta)|k\rangle
\label{eq:KM-inf-matrix-explicit}\\
&  =\int_{0}^{\infty}ds\ \operatorname{Tr}[\left(  \partial_{i}\sigma\right)
\left(  \sigma+sI\right)  ^{-1}\left(  \partial_{j}\sigma\right)  \left(
\sigma+sI\right)  ^{-1}],\label{eq:KM-inf-basis-ind}
\end{align}
where
\begin{equation}
c_{\operatorname{KM}}(x,y)\coloneqq\left\{
\begin{array}
[c]{cc}
\frac{1}{x} & \text{if }x=y\\
\frac{\ln x-\ln y}{x-y} & \text{if }x\neq y
\end{array}
\right.  .
\end{equation}
In~\eqref{eq:KM-inf-basis-ind}, we have used the shorthand $\sigma\equiv
\sigma(\theta)$ for brevity. The formula in~\eqref{eq:KM-inf-basis-ind} can be
useful, as it is basis independent.

Interestingly, observe that the function $\frac{2}{\lambda_k + \lambda_\ell}$ is the inverse of the arithmetic mean of the eigenvalues $\lambda_k$ and $\lambda_\ell$, while the function $c_{\operatorname{KM}}(\lambda_{k},\lambda_{\ell})$ is the inverse of their logarithmic mean.

\section{Fisher--Bures and Kubo--Mori information matrices of thermal
states}

\label{sec:formulas-FB-KM-info-mats}

With this background in place, we move on to state our main technical results
(Theorems~\ref{thm:main} and \ref{thm:main-Kubo-Mori} below). Before doing so,
let us define $\Phi_{\theta}$ to be the following quantum channel:
\begin{equation}
\Phi_{\theta}(X)   \coloneqq\int_{\mathbb{R}}dt\ p(t)\ e^{-iG(\theta
)t}Xe^{iG(\theta)t}\ ,\label{eq:q-channel-phi-theta}
\end{equation}
where 
\begin{equation}
p(t)   \coloneqq\frac{2}{\pi}\ln\left\vert \coth(\pi
t/2)\right\vert \label{eq-mt:prob-dens}
\end{equation}
is a probability density function on $t\in\mathbb{R}$ known as the ``high-peak tent'' probability density~\cite{patel2024quantumboltzmann}, due to its shape when plotted. Prior
work~\cite{Hastings2007,Kim2012,Anshu2021,Coopmans2024} refers to the map
$\Phi_{\theta}$ as the quantum belief propagation superoperator, and it was
observed in~\cite{patel2024quantumboltzmann} that it is a quantum
channel (a completely positive, trace-preserving map) because $e^{-iG(\theta
)t}$ is unitary and, as stated above, $p(t)$ is a probability density function  (see also \cite[Footnote~32]{Kim2012}). 

\begin{theorem}
\label{thm:main}For the parameterized family of thermal states in
\eqref{eq:thermal-state}, the Fisher--Bures information matrix
elements are as follows:
\begin{equation}
I_{ij}^{\operatorname{FB}}(\theta)=\frac{1}{2}\left\langle\left\{
\Phi_{\theta}(G_{i}),\Phi_{\theta}(G_{j})\right\}\right\rangle_{\rho(\theta)}  -\left\langle
G_{i}\right\rangle _{\rho(\theta)}\left\langle G_{j}\right\rangle
_{\rho(\theta)},\label{eq:Fisher-matrix-entries}
\end{equation}
where $\Phi_{\theta}$ is the quantum channel defined in~\eqref{eq:q-channel-phi-theta} and 
$\langle C \rangle_{\sigma} \equiv \operatorname{Tr}[C \sigma]$.
\end{theorem}

Let us note the following alternative forms of $I_{ij}^{\operatorname{FB}
}(\theta)$ for the family in~\eqref{eq:thermal-state}:
\begin{align}
& I_{ij}^{\operatorname{FB}}(\theta) \notag \\
&  =\operatorname{Re}\left[
\operatorname{Tr}[\Phi_{\theta}(G_{i})\rho(\theta)\Phi_{\theta}(G_{j})\right]
]-\left\langle G_{i}\right\rangle _{\rho(\theta)}\left\langle G_{j}
\right\rangle _{\rho(\theta)}\label{eq:q-alg-FI-exp}\\
&  =\frac{1}{2}\operatorname{Tr}\left[  \left\{  \Phi_{\theta}(G_{i})-\mu
_{i},\Phi_{\theta}(G_{j})-\mu_{j}\right\}  \rho(\theta)\right]
,\label{eq:q-cov-matrix}
\end{align}
where
\begin{equation}
\mu_{i}\coloneqq\left\langle \Phi_{\theta}(G_{i})\right\rangle
_{\rho(\theta)}=\left\langle G_{i}\right\rangle _{\rho(\theta)}.
\label{eq:mean-mu-def}
\end{equation}
 The first
expression in~\eqref{eq:q-alg-FI-exp} is helpful for arriving at the
conclusion that each $I_{ij}^{\operatorname{FB}}(\theta)$ can be efficiently
estimated on a quantum computer (under the assumption that each $G_{i}$ is a
local operator, acting on a constant number of qubits). The second expression
in~\eqref{eq:q-cov-matrix} clarifies that $I^{\operatorname{FB}}(\theta)$
itself is a quantum covariance matrix, implying that $I^{\operatorname{FB}}(\theta)\geq 0$ and as already indicated more generally in~\eqref{eq:FB-KM-op-ineq}.

The presence of the quantum channel $\Phi_{\theta}$ in
\eqref{eq:Fisher-matrix-entries}--\eqref{eq:q-cov-matrix} is a hallmark of
non-commutativity of the set $\{G_{j}\}_{j=1}^{J}$ in~\eqref{eq:param-Ham}. If
the set $\{G_{j}\}_{j=1}^{J}$ is commuting, as is the case for thermal states
in classical mechanics, then there is no need for the channel $\Phi_{\theta}$
to be present in~\eqref{eq:Fisher-matrix-entries}--\eqref{eq:q-cov-matrix},
and the multiparameter quantum Cramer--Rao bound in~\eqref{eq:Cramer--Rao-multiple} is
saturated, with the optimal single-copy measurement scheme being to measure
each $G_{i}$ on each copy of $\rho(\theta)$.

The derivation of~\eqref{eq:Fisher-matrix-entries} is rather straightforward
when using the following formula for the partial derivative of~$\rho(\theta
)$:
\begin{equation}
\partial_{j}\rho(\theta)=-\frac{1}{2}\left\{  \Phi_{\theta}(G_{j}),\rho
(\theta)\right\}  +\rho(\theta)\left\langle G_{j}\right\rangle_{\rho(\theta)} .
\label{eq:derivative-thermal-state}
\end{equation}
This formula is a direct consequence of the developments in~\cite[Eq.~(9)]
{Hastings2007},~\cite[Proposition~20]{Anshu2020arXiv}, and~\cite[Lemma~5]
{Coopmans2024} and its proof was reviewed recently in
\cite{patel2024quantumboltzmann}. Plugging~\eqref{eq:derivative-thermal-state}
into the expression $\langle k|\partial_{i}\rho(\theta)|\ell\rangle$ in
\eqref{eq:qfim-explicit} leads to
\begin{equation}
\langle\ell|\partial_{j}\rho(\theta)|k\rangle=-\frac{1}{2}\langle\ell
|\Phi_{\theta}(G_{j})|k\rangle\left(  \lambda_{k}+\lambda_{\ell}\right)
+\delta_{k\ell}\lambda_{\ell}\left\langle G_{j}\right\rangle _{\rho(\theta)}.
\label{eq:overlap-QFI-term-thermal}
\end{equation}
Then plugging~\eqref{eq:overlap-QFI-term-thermal} into
\eqref{eq:qfim-explicit} and performing some basic algebraic manipulations
leads to~\eqref{eq:Fisher-matrix-entries}. See Appendix~\ref{app:FB-info-mat} for details.

We now state our formula for the Kubo--Mori information matrix elements of parameterized thermal states:

\begin{theorem}
\label{thm:main-Kubo-Mori}For the parameterized family of thermal states in
\eqref{eq:thermal-state}, the Kubo--Mori information matrix elements
are as follows:
\begin{equation}
I_{ij}^{\operatorname{KM}}(\theta)=\frac{1}{2}\left \langle\left\{
G_{i},\Phi_{\theta}(G_{j})\right\}\right \rangle_{\rho(\theta)}  -\left\langle G_{i}
\right\rangle _{\rho(\theta)}\left\langle G_{j}\right\rangle _{\rho(\theta
)},\label{eq:KM-info-matrix-entries}
\end{equation}
where $\Phi_{\theta}$ is the quantum channel defined in~\eqref{eq:q-channel-phi-theta} and $\langle C \rangle_{\sigma} \equiv \operatorname{Tr}[C \sigma]$.
\end{theorem}

Let us note the following alternative forms of $I_{ij}^{\operatorname{KM}
}(\theta)$ for the family in~\eqref{eq:thermal-state}:
\begin{align}
& I_{ij}^{\operatorname{KM}}(\theta) \notag \\
&  =\operatorname{Re}\left[
\operatorname{Tr}[G_{i}\rho(\theta)\Phi_{\theta}(G_{j})\right]
]-\left\langle G_{i}\right\rangle _{\rho(\theta)}\left\langle G_{j}
\right\rangle _{\rho(\theta)}\label{eq:q-alg-KMI-exp}\\
&  =\frac{1}{2}\operatorname{Tr}\left[  \left\{  G_{i}-\mu
_{i},\Phi_{\theta}(G_{j})-\mu_{j}\right\}  \rho(\theta)\right]
,\label{eq:q-cov-matrix-KM},
\end{align}
where we again made use of \eqref{eq:mean-mu-def}. 

Observe that, interestingly, the only difference between~\eqref{eq:Fisher-matrix-entries} and~\eqref{eq:KM-info-matrix-entries} is the second appearance  of the quantum channel $\Phi_\theta$ in the first term of~\eqref{eq:Fisher-matrix-entries}. 

The proof of Theorem~\ref{thm:main-Kubo-Mori} follows by plugging \eqref{eq:overlap-QFI-term-thermal} into~\eqref{eq:KM-inf-matrix-explicit} and
performing a number of algebraic manipulations. See Appendix~\ref{app:KM-info-mat-proof} for details. 

Based on the operator inequality in~\eqref{eq:FB-KM-op-ineq}, it follows that
\begin{equation}
v^{T}I^{\operatorname{KM}}(\theta)v\geq v^{T}I^{\operatorname{FB}}
(\theta)v\label{eq:ineq-FB-KM-PSD}
\end{equation}
for every vector $v=(v_{1},\ldots,v_{J})$. Defining $W\coloneqq \sum_{j=1}^{J}
v_{j}G_{j}$, it follows from~\eqref{eq:ineq-FB-KM-PSD}, the expressions in~\eqref{eq:Fisher-matrix-entries} and~\eqref{eq:KM-info-matrix-entries}, and the equality $\left\langle W\right\rangle _{\rho(\theta)} = \left\langle \Phi_{\theta}(W)\right\rangle _{\rho(\theta)}$  that
the following inequality holds:
\begin{multline}
\frac{1}{2}\operatorname{Tr}[\left\{  W,\Phi_{\theta}(W)\right\}  \rho
(\theta)]-\left[  \left\langle W\right\rangle _{\rho(\theta)}\right]  ^{2}\\
\geq\operatorname{Tr}[\left[  \Phi_{\theta}(W)\right]  ^{2}\rho(\theta
)]-\left[  \left\langle \Phi_{\theta}(W)\right\rangle _{\rho(\theta)}\right]
^{2}.\label{eq:anshu-et-al-ineq}
\end{multline}
We note that~\eqref{eq:anshu-et-al-ineq} was proved as~\cite[Lemma~30]{Anshu2020arXiv}. As such, here we have provided an alternative proof of~\eqref{eq:anshu-et-al-ineq},
which is based on the expressions stated in Theorems~\ref{thm:main} and~\ref{thm:main-Kubo-Mori} and the fundamental inequality in~\eqref{eq:FB-KM-op-ineq}, well known in quantum multiparameter estimation.

\section{Quantum algorithms for estimating the Fisher--Bures and
Kubo--Mori information matrices of thermal states}

\label{sec:q-alg-NG-QBM}

In this section, we develop our quantum algorithms for estimating the Fisher--Bures and Kubo--Mori information matrix elements given in~\eqref{eq:Fisher-matrix-entries} and~\eqref{eq:KM-info-matrix-entries}, respectively. We call them the quantum
Boltzmann--Fisher--Bures estimator and the 
quantum Boltzmann--Kubo--Mori estimator, respectively. Similar to the quantum Boltzmann gradient estimator of~\cite{patel2024quantumboltzmann}, they employ a combination of classical random sampling, Hamiltonian
simulation~\cite{lloyd1996universal,childs2018toward}, and the Hadamard
test~\cite{Cleve1998} (the last reviewed in Appendix~\ref{app:quant_primitive}).

\subsection{Quantum
Boltzmann--Fisher--Bures estimator}

\label{sec:QBFB-estimator}

We first develop the quantum
Boltzmann--Fisher--Bures estimator.
Suppose that each $G_{j}$ in~\eqref{eq:param-Ham} is
not only Hermitian but also unitary, as it is for the common case in which
each $G_{j}$ is a tensor product of Pauli operators. Then the second term
$\left\langle G_{i}\right\rangle _{\rho(\theta)}\left\langle G_{j}
\right\rangle _{\rho(\theta)}$ in~\eqref{eq:Fisher-matrix-entries}\ can be
estimated by means of a quantum algorithm. Since it can be written as
\begin{equation}
\left\langle G_{i}\right\rangle _{\rho(\theta)}\left\langle G_{j}\right\rangle
_{\rho(\theta)}=\operatorname{Tr}[\left(  G_{i}\otimes G_{j}\right)  \left(
\rho(\theta)\otimes\rho(\theta)\right)  ], \label{eq:simple-term-to-estimate}
\end{equation}
a procedure for estimating it is to generate the state $\rho(\theta
)\otimes\rho(\theta)$ and then measure the observable $G_{i}\otimes G_{j}$ on
these two copies. Through repetition, the estimate of $\left\langle
G_{i}\right\rangle _{\rho(\theta)}\left\langle G_{j}\right\rangle
_{\rho(\theta)}$ can be made as precise as desired. This procedure is
described in detail as~\cite[Algorithm~2]{patel2024quantumboltzmann}, the
result of which is that $O(\varepsilon^{-2}\ln\delta^{-1})$ samples of
$\rho(\theta)$ are required to have an accuracy of $\varepsilon>0$ with a
failure probability of $\delta\in\left(  0,1\right)  $.

We also delineate a quantum algorithm for estimating the first term
in~\eqref{eq:Fisher-matrix-entries}. It follows by direct substitution of~\eqref{eq:q-channel-phi-theta}, along with further algebraic manipulations,
that
\begin{multline}
\frac{1}{2}\left\langle\left\{  \Phi_{\theta}(G_{i}),\Phi_{\theta}
(G_{j})\right\}\right\rangle_{\rho(\theta)}   = \\
\int_{\mathbb{R}}\int_{\mathbb{R}}dt_{1}\ dt_{2}\ p(t_{1}
)\, p(t_{2})\ \operatorname{Re}[\operatorname{Tr}[U_{i,j}(\theta,t_{1}-t_{2}
)\rho(\theta)]],
\label{eq:QBFB-estimator-proof}
\end{multline}
where
\begin{equation}
U_{i,j}(\theta,t)\coloneqq G_{j}e^{-iG(\theta)t}G_{i}e^{iG(\theta)t}.
\label{eq:unitary-estimator}
\end{equation}
We can then estimate the first term in~\eqref{eq:Fisher-matrix-entries} by a
combination of classical random sampling, Hamiltonian
simulation~\cite{lloyd1996universal,childs2018toward}, and the Hadamard
test~\cite{Cleve1998}. This is the key insight behind the quantum
Boltzmann--Fisher--Bures estimator. Indeed, the basic idea is to sample
$t_{1}$ and $t_{2}$ independently from the probability density $p(t)$. Based
on these choices, we then execute the quantum circuit in
Figure~\ref{fig:QBFB-estimator}, which outputs a binary random
variable~$Y$, that has a realization $y=0$ occurring with probability
\begin{equation}
\frac{1}{2}\left(  1+\operatorname{Re}[\operatorname{Tr}[U_{i,j}(\theta
,t_{1}-t_{2})\rho(\theta)]]\right)
\end{equation}
 and $y=1$ occurring with probability
\begin{equation}
\frac{1}{2}\left(  1-\operatorname{Re}[\operatorname{Tr}[U_{i,j}(\theta
,t_{1}-t_{2})\rho(\theta)]]\right)  .
\end{equation}
Thus, the expectation of a random
variable $Z=\left(  -1\right)  ^{Y}$, conditioned on $t_{1}$ and $t_{2}$, is
equal to
\begin{equation}
\operatorname{Re}[\operatorname{Tr}[U_{i,j}(\theta,t_{1}-t_{2}
)\rho(\theta)]],
\end{equation}
 and including the further averaging over $t_{1}$ and $t_{2}$,
the expectation is equal to $\operatorname{Tr}[\left\{  \Phi_{\theta}
(G_{i}),\Phi_{\theta}(G_{j})\right\}  \rho(\theta)]$. As such, we can sample
$t_{1}$, $t_{2}$, and $Z$ in this way, and averaging the outcomes gives an
unbiased estimate of the first term in~\eqref{eq:Fisher-matrix-entries}. This
procedure is described in detail in Appendix~\ref{app:algo-first-term-FB} as Algorithm~\ref{algo:first-term-FB}, the result
of which is that $O(\varepsilon^{-2}\ln\delta^{-1})$ samples of $\rho(\theta)$
are required to have an accuracy of $\varepsilon>0$ with a failure probability
of $\delta\in\left(  0,1\right)  $.

We finally note that this construction can straightforwardly be generalized
beyond the case of each $G_{j}$ being a Pauli string, if they instead are Hermitian operators 
block encoded into unitary circuits~\cite{Low2019hamiltonian,Gilyen2019}.

\begin{figure}[ptb]
\begin{center}
\includegraphics[
width=3.3598in
]
{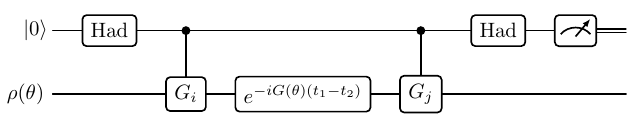}
\end{center}
\caption{Quantum circuit that plays a role in realizing an unbiased estimate
of $\frac{1}{2}\left\langle \left\{  \Phi_{\theta}(G_{i}),\Phi_{\theta
}(G_{j})\right\} \right\rangle_{\rho(\theta)} $. The quantum Boltzmann--Fisher--Bures
estimator combines this estimate with an unbiased estimate of $\left\langle
G_{i}\right\rangle _{\rho(\theta)}\left\langle G_{j}\right\rangle
_{\rho(\theta)}$ in order to realize an unbiased estimate of the Fisher--Bures information
matrix element $I_{i,j}^{\operatorname{FB}}(\theta)$ in~\eqref{eq:Fisher-matrix-entries}.}
\label{fig:QBFB-estimator}
\end{figure}

\subsection{Quantum
Boltzmann--Kubo--Mori estimator}

Let us now develop the quantum Boltzmann--Kubo--Mori estimator. The second term of~\eqref{eq:KM-info-matrix-entries} is the same as the second term of~\eqref{eq:Fisher-matrix-entries}. As such, we can use the same procedure delineated in the paragraph surrounding~\eqref{eq:simple-term-to-estimate} in order to estimate it.

The first term of~\eqref{eq:KM-info-matrix-entries} is slightly different from the first term of~\eqref{eq:Fisher-matrix-entries}, and we can use a similar procedure as outlined in the paragraph surrounding~\eqref{eq:QBFB-estimator-proof} in order to estimate it. We provide details for completeness. It follows by direct substitution of~\eqref{eq:q-channel-phi-theta}, along with further algebraic manipulations,
that
\begin{multline}
\frac{1}{2}\operatorname{Tr}[\left\{  G_{i},\Phi_{\theta}
(G_{j})\right\}  \rho(\theta)]=\\
\int_{\mathbb{R}}dt\ p(t
)\ \operatorname{Re}[\operatorname{Tr}[U_{i,j}(\theta,t
)\rho(\theta)]],
\label{eq:QBKM-estimator-proof}
\end{multline}
where $U_{i,j}(\theta,t)$ is defined in~\eqref{eq:unitary-estimator}.
We can then estimate the first term in~\eqref{eq:KM-info-matrix-entries}, as before, by a
combination of classical random sampling, Hamiltonian
simulation, and the Hadamard
test. Indeed, the basic idea is to sample
$t$ from the probability density $p(t)$. Based
on this choice, we then execute the quantum circuit in
Figure~\ref{fig:QBKM-estimator}, which outputs a Rademacher random
variable~$Y$, that has a realization $y=+1$ occurring with probability
\begin{equation}
\frac{1}{2}\left(  1+\operatorname{Re}[\operatorname{Tr}[U_{i,j}(\theta
,t)\rho(\theta)]]\right)
\end{equation}
 and $y=-1$ occurring with probability
\begin{equation}
\frac{1}{2}\left(  1-\operatorname{Re}[\operatorname{Tr}[U_{i,j}(\theta
,t)\rho(\theta)]]\right)  .
\end{equation}
Thus, the expectation of a random
variable $Z=\left(  -1\right)  ^{Y}$, conditioned on $t$, is
equal to
\begin{equation}
\operatorname{Re}[\operatorname{Tr}[U_{i,j}(\theta,t
)\rho(\theta)]]
\end{equation}
 and including the further averaging over $t$,
the expectation is equal to $\operatorname{Tr}[\left\{  
G_{i},\Phi_{\theta}(G_{j})\right\}  \rho(\theta)]$. As such, we can sample
$t$ and $Z$ in this way, and averaging the outcomes gives an
unbiased estimate of the first term in~\eqref{eq:KM-info-matrix-entries}. This
procedure is described in detail in Appendix~\ref{app:algo-first-term-KM} as Algorithm~\ref{algo:first-term-KM}, the result
of which is that $O(\varepsilon^{-2}\ln\delta^{-1})$ samples of $\rho(\theta)$
are required to have an accuracy of $\varepsilon>0$ with a failure probability
of $\delta\in\left(  0,1\right)  $.

\begin{figure}[ptb]
\begin{center}
\includegraphics[
width=3.3598in
]
{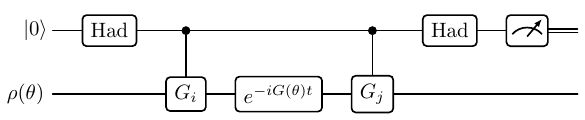}
\end{center}
\caption{Quantum circuit that plays a role in realizing an unbiased estimate
of $\frac{1}{2}\operatorname{Tr}[\left\{  G_{i},\Phi_{\theta
}(G_{j})\right\}  \rho(\theta)]$. The quantum Boltzmann--Kubo--Mori
estimator combines this estimate with an unbiased estimate of $\left\langle
G_{i}\right\rangle _{\rho(\theta)}\left\langle G_{j}\right\rangle
_{\rho(\theta)}$ in order to realize an unbiased estimate of the Kubo--Mori information
matrix element $I_{i,j}^{\operatorname{KM}}(\theta)$ in~\eqref{eq:KM-info-matrix-entries}.}
\label{fig:QBKM-estimator}
\end{figure}

\section{Natural gradient for quantum Boltzmann machine learning}

\label{sec:NG-QBM}

Our first application of Theorem~\ref{thm:main} is to
propose a natural gradient descent algorithm for optimizing quantum Boltzmann
machines. Before doing so, let us briefly motivate natural gradient descent,
following the motivation given in~\cite{Stokes2020quantumnatural}. Let
$\mathcal{L}(\theta)$ be a loss function, which is a function of the parameter
vector $\theta\in\mathbb{R}^{J}$. The goal of an optimization algorithm is to
minimize the loss $\mathcal{L}(\theta)$, and the standard gradient descent
algorithm does so by means of the following update rule:
\begin{equation}
\theta_{m+1}\coloneqq\theta_{m}-\eta\nabla_{\theta}\mathcal{L}(\theta_{m}),
\end{equation}
where $\eta>0$ is the learning rate or step size and $\nabla_{\theta
}\mathcal{L}(\theta)$ is the gradient. Given that
\begin{multline}
\theta_{m}-\eta\nabla_{\theta}\mathcal{L}(\theta_{m}
)=\label{eq:GD-euclidean-opt}\\
\arg\min_{\theta\in\mathbb{R}^{J}}\left[  \left\langle \theta-\theta
_{m},\nabla_{\theta}\mathcal{L}(\theta_{m})\right\rangle +\frac{1}{2\eta
}\left\Vert \theta-\theta_{m}\right\Vert _{2}^{2}\right]  ,
\end{multline}
the standard gradient descent algorithm moves in the direction of steepest
descent, with respect to $\ell_{2}$ (Euclidean)\ geometry.
As argued
previously, this geometry is typically not the correct geometry for the underlying
problem and can lead to difficulties when optimizing.

The intuition behind~\eqref{eq:GD-euclidean-opt} is as follows. We need to choose a point that minimizes the function $\mathcal{L}$. So, we begin with a linear approximation of this function at the point $\theta_m$:
\begin{equation}
    \mathcal{L}(\theta) \approx \mathcal{L}(\theta_m) + \left\langle \theta-\theta
_{m},\nabla_{\theta}\mathcal{L}(\theta_{m})\right\rangle .
\label{eq:linear-approx-loss}
\end{equation}
However, directly minimizing the above function leads to a value of~$-\infty$ because the above problem is an unconstrained optimization problem. Therefore, we need to add a penalty term to the right-hand side of \eqref{eq:linear-approx-loss}, which forces the optimization to look in the vicinity of the point~$\theta_m$:
\begin{equation}
    \mathcal{L}(\theta_m) + \left\langle \theta-\theta
_{m},\nabla_{\theta}\mathcal{L}(\theta_{m})\right\rangle + \frac{1}{2\eta} \left\Vert \theta - \theta_m \right\Vert_2^{2}.
    \label{eq:constrained-opt-GD}
\end{equation}
In the above, $\eta$ is the step size which one can tune depending on how much one wants to penalize the objective function. Thus, the resulting constrained optimization problem is to minimize the right-hand side of \eqref{eq:constrained-opt-GD} over all $\theta$, and an optimal choice of $\theta$ is given by~\eqref{eq:GD-euclidean-opt}.

\subsection{Natural gradient with respect to the Fisher--Bures metric}

\label{sec:NG-FB}

For thermal states of
the form in~\eqref{eq:thermal-state}, the geometry induced by the
Fisher--Bures metric is suitable. In order to accommodate it, we modify the
optimization problem in~\eqref{eq:GD-euclidean-opt} to incorporate the
Fisher--Bures metric in~\eqref{eq:Bures-Fisher-metric}:
\begin{multline}
\theta_{m+1}=\label{eq:mod-GD-nat-opt}\\
\arg\min_{\theta\in\mathbb{R}^{J}}\left[  \left\langle \theta-\theta
_{m},\nabla_{\theta}\mathcal{L}(\theta_{m})\right\rangle +\frac{1}{2\eta
}\left\Vert \theta-\theta_{m}\right\Vert _{g^{\operatorname{FB}}(\theta_{m})}^{2}\right]  ,
\end{multline}
where the squared distance measure $\left\Vert \theta-\theta_{m}\right\Vert
_{g^{\operatorname{FB}}(\theta_{m})}^{2}$ is defined as
\begin{align}
\left\Vert \theta-\theta_{m}\right\Vert _{g^{\operatorname{FB}}(\theta_{m})}^{2}  &
\coloneqq\left\langle \theta-\theta_{m},g^{\operatorname{FB}}(\theta_{m})\left(  \theta-\theta
_{m}\right)  \right\rangle ,\\
\left[  g^{\operatorname{FB}}(\theta)\right]  _{ij}  &  \coloneqq\frac{1}{4}I_{ij}
^{\operatorname{FB}}(\theta).
\end{align}
The first-order optimality condition corresponding to~\eqref{eq:mod-GD-nat-opt} is
\begin{equation}
\frac{1}{4}I^{\operatorname{FB}}(\theta_{m})\left(  \theta_{m+1}
-\theta_{m}\right)  =-\eta\nabla_{\theta}\mathcal{L}(\theta_{m}),
\end{equation}
and solving the above equation for $\theta_{m+1}$ gives the following natural
gradient descent update rule:
\begin{equation}
\theta_{m+1}=\theta_{m}-4\eta\left[  I^{\operatorname{FB}}(\theta
_{m})\right]  ^{-1}\nabla_{\theta}\mathcal{L}(\theta_{m}),
\label{eq:QNG-thermal}
\end{equation}
where $\left[  I^{\operatorname{FB}}(\theta)\right]  ^{-1}$ is the
inverse of the Fisher--Bures information matrix $I^{\operatorname{FB}}(\theta)$
if it is invertible and otherwise it is the pseudoinverse of $I
^{\operatorname{FB}}(\theta)$.

Equation~\eqref{eq:QNG-thermal} gives the basic update rule for any
optimization problem that uses the quantum Boltzmann machine ansatz. It
delineates a hybrid quantum--classical optimization loop, as is the case with
natural gradient for parameterized quantum circuits
\cite{Stokes2020quantumnatural}. Indeed, at each step, one should have a
method to evaluate the loss function gradient $\nabla_{\theta}\mathcal{L}
(\theta_{m})$, and then one can evaluate the $J^{2}$ elements of the 
Fisher--Bures information matrix $I^{\operatorname{FB}}(\theta_{m})$, by using
the quantum algorithm from Section~\ref{sec:QBFB-estimator}.

Let us note that the matrix $I^{\operatorname{FB}}(\theta
_{m})$ is positive definite (and thus invertible) for the common situation in which each $G_j$ acts non-trivially on a constant number of qubits. This follows as a consequence of the last equation of \cite[Section~7.2]{Anshu2020arXiv} and the connection made in~\eqref{eq:anshu-et-al-ineq}.

\subsection{Natural gradient with respect to the Kubo--Mori metric}

An alternative metric to use in natural gradient descent for quantum Boltzmann machines is the Kubo--Mori metric in~\eqref{eq:KM-metric-def}. The reasoning behind the algorithm is the same as given above, and so we highlight only the key differences. Defining
\begin{align}
\left\Vert \theta-\theta_{m}\right\Vert _{g^{\operatorname{KM}}(\theta_{m})}^{2}  &
\coloneqq\left\langle \theta-\theta_{m},g^{\operatorname{KM}}(\theta_{m})\left(  \theta-\theta
_{m}\right)  \right\rangle ,\\
\left[  g^{\operatorname{KM}}(\theta)\right]  _{ij}  &  \coloneqq\frac{1}{2}I_{ij}
^{\operatorname{KM}}(\theta),
\end{align}
the basic update rule for natural gradient in terms of the Kubo--Mori metric is as follows:
\begin{equation}
\theta_{m+1}=\theta_{m}-2\eta\left[  I^{\operatorname{KM}}(\theta
_{m})\right]  ^{-1}\nabla_{\theta}\mathcal{L}(\theta_{m}),
\label{eq:QNG-KM-thermal}
\end{equation}

%\mmw{need to revisit choice of $4$ as constant prefactor}

Due to \eqref{eq:FB-KM-op-ineq} and the statement at the end of Section~\ref{sec:NG-FB}, it follows that the Kubo--Mori information matrix $I^{\operatorname{KM}}(\theta
_{m})$ is positive definite (and thus invertible) for the common situation when each $G_j$ acts non-trivially on a constant number of qubits.  See also~\cite[Theorem~28 and Lemma~29]{Anshu2020arXiv} in this context.

\subsection{Example optimization tasks with natural gradient}

As a first example,  suppose the optimization problem is to estimate the ground-state energy of a
Hamiltonian $H$ that is efficiently measurable. One can combine the algorithms
from~\cite{patel2024quantumboltzmann} and Section~\ref{sec:q-alg-NG-QBM} to
evaluate the update rule in~\eqref{eq:QNG-thermal}. In this case, the loss
function is $\mathcal{L}(\theta)=\operatorname{Tr}[H\rho(\theta)]$, and, as
shown in~\cite{patel2024quantumboltzmann}, each element of the gradient
$\nabla_{\theta}\mathcal{L}(\theta)$ is given by
\begin{equation}
\partial_{j}\mathcal{L}(\theta)=-\frac{1}{2}\operatorname{Tr}\!\left[
\left\{  H,\Phi_{\theta}(G_{j})\right\}  \rho(\theta)\right] 
+\left\langle H\right\rangle _{\rho(\theta)}\left\langle G_{j}\right\rangle
_{\rho(\theta)}.
\end{equation}
One can then use the quantum Boltzmann gradient estimator of
\cite{patel2024quantumboltzmann} to evaluate each element of $\nabla_{\theta
}\mathcal{L}(\theta_{m})$, the quantum Boltzmann--Fisher--Bures estimator of
Section~\ref{sec:q-alg-NG-QBM} to evaluate each element of
$I^{\operatorname{FB}}(\theta_{m})$, and then update according to the rule in~\eqref{eq:QNG-thermal}.

Another example lies in generative modeling, using the results of
\cite{Coopmans2024}. There, the problem is that a state $\omega$ is given, and
one desires to generate a state $\rho(\theta)$ that is close to the given
state $\omega$. A natural measure of closeness is the quantum relative entropy
$D(\omega\Vert\rho(\theta))$, as defined in~\eqref{eq:rel-ent-def}. In this case, we thus set the
loss function $\mathcal{L}(\theta)=D(\omega\Vert\rho(\theta))$, and it was
shown in~\cite{Coopmans2024} that each element of the gradient $\nabla
_{\theta}\mathcal{L}(\theta)$ is given by
\begin{equation}
\partial_{j}\mathcal{L}(\theta)=\left\langle H\right\rangle _{\rho(\theta
)}-\left\langle G_{j}\right\rangle _{\rho(\theta)}.
\end{equation}
Due to this simple form of the gradient, one can use a simple algorithm to evaluate each element of
$\nabla_{\theta}\mathcal{L}(\theta_{m})$, the quantum Boltzmann--Fisher--Bures
estimator of Section~\ref{sec:q-alg-NG-QBM} to evaluate each element of
$I^{\operatorname{FB}}(\theta_{m})$, and then update according to the rule in~\eqref{eq:QNG-thermal}.

In both cases, we could alternatively employ natural gradient descent with respect to the Kubo--Mori metric, as given in~\eqref{eq:QNG-KM-thermal}.

These are but two examples of how natural gradient can be used with the
quantum Boltzmann machine ansatz, and we note again here that the update rule
in~\eqref{eq:QNG-thermal} can be incorporated into any optimization algorithm
that uses the quantum Boltzmann machine ansatz.

Although we do not provide a detailed analysis of the sample complexity of ground-state energy minimization or generative modeling in this paper, as done in the previous works~\cite{Coopmans2024,patel2024quantumboltzmann},  we should note here that, at the least, we expect the number of iterations required for natural gradient to converge to an $\varepsilon$-stationary point to be fewer than those required when performing standard gradient descent. This is due to the advantage that natural gradient has by incorporating the geometry of the landscape when making a decision about which direction navigate to next. We leave it open for now to establish a rigorous convergence analysis for natural gradient with quantum Boltzmann machines. 

The main disadvantage of natural gradient descent compared to standard gradient descent is that the update rules in~\eqref{eq:QNG-thermal} and~\eqref{eq:QNG-KM-thermal} require estimating the $J^2$ parameters in the Fisher--Bures and Kubo--Mori information matrices, respectively. Furthermore, it is necessary to compute their inverses as well. As such, there is a trade-off between this extra sampling and computation time against natural gradient descent's ability to navigate the landscape more judiciously.

\section{Estimation of Hamiltonian parameters from thermal states}

\label{sec:est-ham-param}

In this section, we briefly outline our second application, which is to estimating the parameter vector $\theta$, for a Hamiltonian of the form in~\eqref{eq:param-Ham}, from thermal states of the form in~\eqref{eq:thermal-state}. The setting here is that the tensor-power state $\rho(\theta)^{\otimes n}$, corresponding to $n$ independent copies or samples of $\rho(\theta)$, is available. One then performs a collective measurement $(M^{(n)}_m)_m$ on all $m$ copies, followed by classical processing according to the function $\hat{\theta}(m)$ in~\eqref{eq:parameter-estimate-function}, in order to make an estimate of the parameter vector $\theta$. Under the assumption that the estimator is unbiased, so that
\begin{equation}
\sum_m \operatorname{Tr}[M^{(n)}_m \rho(\theta)^{\otimes n}] \, \hat{\theta}(m) = \theta,
\end{equation}
the covariance matrix in~\eqref{eq:cov-mat-def} is a key figure of merit, determining how well a given estimator performs.

\subsection{Fundamental limitations on multiparameter estimation of thermal states}

One of the seminal results in this direction is the multiparameter Cramer--Rao bound of~\eqref{eq:Cramer--Rao-multiple}. Thus, for our problem at hand, Theorem~\ref{thm:main} provides an analytical expression for the Fisher--Bures information matrix of thermal states, and thus in turn limits the performance of any possible estimator. For sufficiently small system sizes, the Fisher--Bures information matrix elements can be calculated analytically. For larger system sizes and for suitable Hamiltonians that meet the assumptions of the quantum Boltzmann--Fisher--Bures estimator, one can estimate the matrix elements by means of this estimator, in order to understand the fundamental limitations on any possible scheme for estimating parameters of Hamiltonians.

\subsection{Asymptotically optimal strategy for single-parameter estimation of thermal states}

Our findings also apply to single-parameter estimation, that is, when one is trying to determine a single unknown parameter  $\theta_j$ from thermal-state samples of the form in~\eqref{eq:thermal-state}, and all other parameter values are known. For estimating a single parameter~$\theta_j$, an optimal measurement strategy that saturates the quantum Cramer--Rao bound (corresponding to the $j$th diagonal entry of~\eqref{eq:Cramer--Rao-multiple}) involves measuring in the eigenbasis of the  symmetric logarithmic derivative (SLD) operator $L^{(j)}(\theta)$~\cite{braunsteincaves1994}, which is given as follows \cite[Eq.~(86)]{Sidhu2020}:
\begin{equation}
    L^{(j)}(\theta) \coloneqq \sum_{k,\ell} \frac{2}{\lambda
_{k}+\lambda_{\ell}}|k\rangle \!\langle k|\partial_{j}\rho(\theta
)|\ell\rangle \!\langle \ell | .
\end{equation}
Note that the SLD satisfies the following differential equation~\cite[Eq.~(83)]{Sidhu2020}:
\begin{equation}
    \partial_{j}\rho(\theta
) = \frac{1}{2} (\rho(\theta) L^{(j)}(\theta) + L^{(j)}(\theta) \rho(\theta)).
\end{equation}
For our case of interest, we have the following result:

\begin{theorem}
\label{thm:main-sld}For the parameterized family of thermal states in
\eqref{eq:thermal-state}, the SLD operator $L^{(j)}(\theta)$
is given as follows:
\begin{equation}
L^{(j)}(\theta)=-\Phi_{\theta}(G_{j}) + \left\langle
G_{j}\right\rangle _{\rho(\theta)} I,\label{eq:SLD-op}
\end{equation}
where $\Phi_{\theta}$ is the quantum channel defined in~\eqref{eq:q-channel-phi-theta} and $\langle C \rangle_{\sigma} \equiv \operatorname{Tr}[C \sigma]$.
\end{theorem}

\begin{proof}
See Appendix~\ref{app:proof-SLD}.
\end{proof}

\medskip 

See also \cite[Appendix~C]{abiuso2024fundamentallimitsmetrologythermal} for an alternative representation of the SLD operator for thermal states. A key distinction between the form in Theorem~\ref{thm:main-sld} and that from \cite[Appendix~C]{abiuso2024fundamentallimitsmetrologythermal} is that the former leads to a quantum algorithm for measuring in the eigenbasis of the SLD operator $L^{(j)}(\theta)$, as sketched below.

Indeed, Theorem~\ref{thm:main-sld} provides an analytical form for the SLD operator, and we know from~\cite{braunsteincaves1994} that the ability to measure in its eigenbasis implies an asymptotically optimal strategy for estimating $\theta_j$, in the sense that it saturates the quantum Cramer--Rao bound. As such, to realize an optimal strategy, we need a method for measuring in the eigenbasis of $\Phi_{\theta}(G_{j})$.

With these notions in place, we now present a sketch of an algorithm that realizes the optimal measurement mentioned above, which involves measuring in the eigenbasis of the SLD operator $L^{(j)}(\theta)$. Having said that, we leave the detailed implementation and rigorous complexity analysis of this algorithm for future work. Our algorithm primarily employs the standard quantum phase estimation (QPE) algorithm to project onto the eigenbasis of the observable $\Phi_{\theta}(G_j)$. To employ standard QPE, we require controlled Hamiltonian evolutions. For this purpose, we use a Hamiltonian simulation algorithm based on quantum singular value transformation (QSVT)~\cite{Gilyen2019}. In this framework, the Hamiltonian of interest, denoted by $H$, is first block-encoded in a unitary operator~\cite{Low2019hamiltonian,Gilyen2019}. Subsequently, QSVT is employed to implement a polynomial approximation of the target evolution operator $e^{-iHt}$. In our specific case, the Hamiltonian of interest is~$\Phi_{\theta}(G_j)$.

To begin with, let us discuss how to block-encode $\Phi_{\theta}(G_j)$. We assume that we have access to a unitary~$P$ that prepares the following state:
\begin{equation}
    |\varphi(T)\rangle \coloneqq \sum_{t=0}^{T-1} \sqrt{p_D(t)} |t\rangle,
\end{equation}
where $p_D(t)$ is a discretization of the probability density in \eqref{eq-mt:prob-dens}, so that $P|0\rangle = |\varphi(T)\rangle$. Note that the discretized probability distribution $p_D(t)$ can realize the channel $\Phi_\theta$ exactly if a bound on the spectral norm of $G(\theta)$ is available.
We also assume that we have access to the inverse unitary $P^{\dagger}$. Additionally, we require controlled unitaries $W^{(\pm)}$ defined as controlled Hamiltonian evolutions:
\begin{equation}
    W^{(\pm)} \coloneqq \sum_{t=0}^{T-1} |t\rangle\!\langle t| \otimes e^{\pm iG(\theta)t}.
\end{equation}
Then the following unitary block-encodes  $\Phi_{\theta}(G_j)$:
\begin{equation}
    \left(P^{\dag} \otimes I\right) \left(W^{(-)} \left( I \otimes G_j \right) W^{(+)}\right)\left(P\otimes I\right).
\end{equation}
The quantum circuit representation of the above unitary is depicted in Figure~\ref{fig:block-encode}.

\begin{figure}
    \centering
    \includegraphics[width=\linewidth]{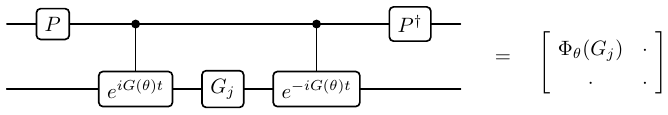}
    \caption{Block-encoding of the operator $\Phi_{\theta}(G_j)$.}
    \label{fig:block-encode}
\end{figure}

That being said, our algorithm for estimating $\theta_j$ consists of the following steps:
\begin{enumerate}

    \item Set $i \coloneqq 1$.

    \item  Block-encode $\Phi_{\theta}(G_j)$ using the aforementioned approach.
    
    \item Prepare the system register in the state $\rho(\theta)$. This state can be written in the eigenbasis $\{|\lambda_r\rangle\}_r$ of $\Phi_{\theta}(G_j)$ in the following way:
    \begin{equation}
        \rho(\theta) \coloneqq \sum_{r, r'} c_{r, r'} |\lambda_r\rangle\!\langle\lambda_{r'}|.
    \end{equation}
    
    \item Append $p$ control registers and initialize these registers to $|0\rangle^{\otimes p}$.
    
    \item Apply the standard QPE algorithm jointly to the control and system registers, which realizes the following transformation:
    \begin{equation}
        |0\rangle^{\otimes p}|\lambda_r\rangle \xrightarrow{\operatorname{QPE}} |\tilde\lambda_r\rangle |\lambda_r\rangle,
    \end{equation}
    where $\tilde\lambda_r$ is the closest $p$-bit estimate of the eigenvalue $\lambda_r$, corresponding to the eigenstate $|\lambda_r\rangle$. 
    
    \item Measure the control  registers in the computational basis and obtain a measurement outcome $\tilde\lambda_r$.
    
    \item Set $b_i \coloneqq \tilde\lambda_r$. Set $i \coloneqq i+1$.
    
    \item Repeat Steps 2-7 $n$ times.

    \item Evaluate $\hat{\theta}_j(b_1, \ldots, b_n)$, where $\hat{\theta}_j$ is a classical map that processes the measurement outcomes to produce an estimate of $\theta_j$.
\end{enumerate}

One important caveat to mention, when measuring the observable $ \Phi_\theta(G_j) $, is that knowledge of the  parameter vector $\theta$ is needed to perform the optimal measurement. This kind of problem occurs often in quantum estimation theory \cite[Section~IV-I]{Sidhu2020}, and the typical way around it is to argue that one begins an estimation scheme with a guess of the unknown parameter and adjusts as the estimation scheme proceeds. In our case, this corresponds to making an initial guess for the parameter $\theta_j$, in order to implement the quantum channel $\Phi_\theta$, and then as further information is acquired, one adjusts the guess for $\theta_j$ and employs the updated guess in future implementations of the channel $\Phi_\theta$. We leave the detailed exploration of this kind of strategy for future work, and note here that it is mentioned in \cite[Section~IV-J]{Sidhu2020} that an iterative scoring algorithm, similar to natural gradient, can be used to iteratively determine an estimate for the unknown parameter $\theta_j$.

\section{Conclusion and future directions}

\label{sec:conclusion}

The main technical results of this paper are exact analytical expressions for
the Fisher--Bures and Kubo--Mori information matrix elements of parameterized thermal
states (see Theorems~\ref{thm:main} and~\ref{thm:main-Kubo-Mori}). Based on these analytical expressions, we
then developed corresponding quantum algorithms, called the quantum
Boltzmann--Fisher--Bures and the quantum Boltzmann--Kubo--Mori estimator, which estimate each matrix element of the respective matrix (see
Section~\ref{sec:q-alg-NG-QBM}). These results have two applications, in a method
for quantum Boltzmann machine learning called natural gradient and in
establishing fundamental limitations on our ability to estimate a Hamiltonian
from a thermal state of the form in~\eqref{eq:thermal-state}. We also sketched an asymptotically optimal approach to single-parameter estimation of thermal states, leaving its detailed study for future work. The natural
gradient descent method developed here is general and broad, such that it can be used
for any optimization problem that employs the quantum Boltzmann machine
ansatz. Combined with the recent findings of~\cite{Coopmans2024} and
\cite{patel2024quantumboltzmann}, our results here imply that natural gradient descent
can be used for both quantum generative modeling and estimating ground-state
energies of Hamiltonians, respectively (as outlined at the end of
Section~\ref{sec:NG-QBM}).

In this paper, we have outlined the theory of natural gradient descent with quantum Boltzmann machines, and in future work, we plan to investigate numerically the performance of the approach. As also mentioned in~\cite{patel2024quantumboltzmann}, it remains a pressing open problem to determine whether quantum Boltzmann machine learning suffers from the barren-plateau problem~\cite{McClean_2018,marrero2021entanglement,arrasmith2022equivalence,holmes2022connecting,Fontana2024,Ragone2024} that plagues quantum machine learning using parameterized quantum circuits. There is evidence from~\cite{Coopmans2024} that quantum Boltzmann machines do not suffer from this problem in certain contexts.

Going forward from here, it is an interesting open problem to determine the
sample complexity of natural gradient with quantum Boltzmann machines for
problems of interest, such as generative modeling and estimating ground-state
energies. Analyses of the performance of natural gradient in the classical
case should be helpful for this purpose
\cite{desjardins2013metricfreenaturalgradientjointtraining}.

Although our statements about limitations on estimation focused on what the Cramer--Rao bound states (see~\eqref{eq:Cramer--Rao-multiple} and~\eqref{eq:scalar-Cramer--Rao-from-matrix-one}), one might also wonder what can be said about sample complexity, i.e., the minimum number of copies of $\sigma(\theta)$ that is needed to learn the parameter vector $\theta$ with success probability larger than $1-\delta$ and error less than $\varepsilon$. Here we note that our Fisher information calculations can be combined with a multiparameter generalization of \cite[Eq.~(53)]{meyer2025} to arrive at a sample complexity statement, but we leave a detailed analysis for future work.

Another interesting direction is to explore the connection between mirror descent and natural gradient descent for quantum Boltzmann machines. This connection was recently made in~\cite{sbahi2022provablyefficientvariationalgenerative}, for the case of quantum Hamiltonian-based models and the Kubo--Mori metric. An advantage of mirror descent is that it is a first-order method, as opposed to natural gradient descent, which is a second-order method and furthermore requires an inverse to be calculated. In~\cite{sbahi2022provablyefficientvariationalgenerative}, the two methods were compared for quantum Hamiltonian-based models, and it was found that mirror descent was more robust and stable than natural gradient descent.

\begin{acknowledgments}
DP and MMW acknowledge support from AFRL under agreement no.~FA8750-23-2-0031.
The U.S.~Government is authorized to
reproduce and distribute reprints for Governmental purposes notwithstanding
any copyright notation thereon. The views and conclusions contained herein are
those of the authors and should not be interpreted as necessarily representing
the official policies or endorsements, either expressed or implied, of Air
Force Research Laboratory or the U.S.~Government.
\end{acknowledgments}

\section*{Author contributions}

\noindent \textbf{Author Contributions}:
The following describes the different contributions of the authors of this work, using roles defined by the CRediT
(Contributor Roles Taxonomy) project~\cite{NISO}:

\medskip 
\noindent \textbf{DP}: Formal Analysis,  Investigation, Methodology, Validation, Writing - Original draft,  Writing - Review \& Editing.

\medskip 
\noindent \textbf{MMW}: Conceptualization, Formal Analysis, Funding acquisition,  Investigation, Methodology, Validation, Writing - Original draft,  Writing - Review \& Editing.

\bibliography{Ref}

\begin{thebibliography}{66}%
\makeatletter
\providecommand \@ifxundefined [1]{%
 \@ifx{#1\undefined}
}%
\providecommand \@ifnum [1]{%
 \ifnum #1\expandafter \@firstoftwo
 \else \expandafter \@secondoftwo
 \fi
}%
\providecommand \@ifx [1]{%
 \ifx #1\expandafter \@firstoftwo
 \else \expandafter \@secondoftwo
 \fi
}%
\providecommand \natexlab [1]{#1}%
\providecommand \enquote  [1]{``#1''}%
\providecommand \bibnamefont  [1]{#1}%
\providecommand \bibfnamefont [1]{#1}%
\providecommand \citenamefont [1]{#1}%
\providecommand \href@noop [0]{\@secondoftwo}%
\providecommand \href [0]{\begingroup \@sanitize@url \@href}%
\providecommand \@href[1]{\@@startlink{#1}\@@href}%
\providecommand \@@href[1]{\endgroup#1\@@endlink}%
\providecommand \@sanitize@url [0]{\catcode `\\12\catcode `\$12\catcode `\&12\catcode `\#12\catcode `\^12\catcode `\_12\catcode `\%12\relax}%
\providecommand \@@startlink[1]{}%
\providecommand \@@endlink[0]{}%
\providecommand \url  [0]{\begingroup\@sanitize@url \@url }%
\providecommand \@url [1]{\endgroup\@href {#1}{\urlprefix }}%
\providecommand \urlprefix  [0]{URL }%
\providecommand \Eprint [0]{\href }%
\providecommand \doibase [0]{https://doi.org/}%
\providecommand \selectlanguage [0]{\@gobble}%
\providecommand \bibinfo  [0]{\@secondoftwo}%
\providecommand \bibfield  [0]{\@secondoftwo}%
\providecommand \translation [1]{[#1]}%
\providecommand \BibitemOpen [0]{}%
\providecommand \bibitemStop [0]{}%
\providecommand \bibitemNoStop [0]{.\EOS\space}%
\providecommand \EOS [0]{\spacefactor3000\relax}%
\providecommand \BibitemShut  [1]{\csname bibitem#1\endcsname}%
\let\auto@bib@innerbib\@empty
%</preamble>
\bibitem [{\citenamefont {Continentino}(2021)}]{Continentino2021}%
  \BibitemOpen
  \bibfield  {author} {\bibinfo {author} {\bibfnamefont {M.~A.}\ \bibnamefont {Continentino}},\ }\href {https://doi.org/10.1088/978-0-7503-3395-5} {\emph {\bibinfo {title} {Key Methods and Concepts in Condensed Matter Physics}}},\ 2053-2563\ (\bibinfo  {publisher} {IOP Publishing},\ \bibinfo {year} {2021})\BibitemShut {NoStop}%
\bibitem [{\citenamefont {Deglmann}\ \emph {et~al.}(2015)\citenamefont {Deglmann}, \citenamefont {Sch\"afer},\ and\ \citenamefont {Lennartz}}]{DSL15}%
  \BibitemOpen
  \bibfield  {author} {\bibinfo {author} {\bibfnamefont {P.}~\bibnamefont {Deglmann}}, \bibinfo {author} {\bibfnamefont {A.}~\bibnamefont {Sch\"afer}},\ and\ \bibinfo {author} {\bibfnamefont {C.}~\bibnamefont {Lennartz}},\ }\bibfield  {title} {\bibinfo {title} {Application of quantum calculations in the chemical industry---an overview},\ }\href {https://doi.org/https://doi.org/10.1002/qua.24811} {\bibfield  {journal} {\bibinfo  {journal} {International Journal of Quantum Chemistry}\ }\textbf {\bibinfo {volume} {115}},\ \bibinfo {pages} {107} (\bibinfo {year} {2015})}\BibitemShut {NoStop}%
\bibitem [{\citenamefont {Alhambra}(2023)}]{Alhambra2023}%
  \BibitemOpen
  \bibfield  {author} {\bibinfo {author} {\bibfnamefont {A.~M.}\ \bibnamefont {Alhambra}},\ }\bibfield  {title} {\bibinfo {title} {Quantum many-body systems in thermal equilibrium},\ }\href {https://doi.org/10.1103/PRXQuantum.4.040201} {\bibfield  {journal} {\bibinfo  {journal} {PRX Quantum}\ }\textbf {\bibinfo {volume} {4}},\ \bibinfo {pages} {040201} (\bibinfo {year} {2023})}\BibitemShut {NoStop}%
\bibitem [{\citenamefont {Jaynes}(1957{\natexlab{a}})}]{Jaynes1957}%
  \BibitemOpen
  \bibfield  {author} {\bibinfo {author} {\bibfnamefont {E.~T.}\ \bibnamefont {Jaynes}},\ }\bibfield  {title} {\bibinfo {title} {Information theory and statistical mechanics},\ }\href {https://doi.org/10.1103/PhysRev.106.620} {\bibfield  {journal} {\bibinfo  {journal} {Physical Review}\ }\textbf {\bibinfo {volume} {106}},\ \bibinfo {pages} {620} (\bibinfo {year} {1957}{\natexlab{a}})}\BibitemShut {NoStop}%
\bibitem [{\citenamefont {Jaynes}(1957{\natexlab{b}})}]{Jaynes1957a}%
  \BibitemOpen
  \bibfield  {author} {\bibinfo {author} {\bibfnamefont {E.~T.}\ \bibnamefont {Jaynes}},\ }\bibfield  {title} {\bibinfo {title} {Information theory and statistical mechanics. {II}},\ }\href {https://doi.org/10.1103/PhysRev.108.171} {\bibfield  {journal} {\bibinfo  {journal} {Physical Review}\ }\textbf {\bibinfo {volume} {108}},\ \bibinfo {pages} {171} (\bibinfo {year} {1957}{\natexlab{b}})}\BibitemShut {NoStop}%
\bibitem [{\citenamefont {Giudice}\ \emph {et~al.}(2021)\citenamefont {Giudice}, \citenamefont {Cakan}, \citenamefont {Cirac},\ and\ \citenamefont {Ba\~nuls}}]{Giudice2021}%
  \BibitemOpen
  \bibfield  {author} {\bibinfo {author} {\bibfnamefont {G.}~\bibnamefont {Giudice}}, \bibinfo {author} {\bibfnamefont {A.}~\bibnamefont {Cakan}}, \bibinfo {author} {\bibfnamefont {J.~I.}\ \bibnamefont {Cirac}},\ and\ \bibinfo {author} {\bibfnamefont {M.~C.}\ \bibnamefont {Ba\~nuls}},\ }\bibfield  {title} {\bibinfo {title} {R\'enyi free energy and variational approximations to thermal states},\ }\href {https://doi.org/10.1103/PhysRevB.103.205128} {\bibfield  {journal} {\bibinfo  {journal} {Physical Review B}\ }\textbf {\bibinfo {volume} {103}},\ \bibinfo {pages} {205128} (\bibinfo {year} {2021})}\BibitemShut {NoStop}%
\bibitem [{\citenamefont {Binder}\ \emph {et~al.}(2018)\citenamefont {Binder}, \citenamefont {Correa}, \citenamefont {Gogolin}, \citenamefont {Anders},\ and\ \citenamefont {Adesso}}]{Binder2018}%
  \BibitemOpen
  \bibinfo {editor} {\bibfnamefont {F.}~\bibnamefont {Binder}}, \bibinfo {editor} {\bibfnamefont {L.~A.}\ \bibnamefont {Correa}}, \bibinfo {editor} {\bibfnamefont {C.}~\bibnamefont {Gogolin}}, \bibinfo {editor} {\bibfnamefont {J.}~\bibnamefont {Anders}},\ and\ \bibinfo {editor} {\bibfnamefont {G.}~\bibnamefont {Adesso}},\ eds.,\ \href {https://doi.org/10.1007/978-3-319-99046-0} {\emph {\bibinfo {title} {Thermodynamics in the quantum regime}}},\ Fundamental Theories of Physics\ (\bibinfo  {publisher} {Springer},\ \bibinfo {year} {2018})\BibitemShut {NoStop}%
\bibitem [{\citenamefont {Brand{\~{a}}o}\ and\ \citenamefont {Svore}(2017)}]{Brandao2017}%
  \BibitemOpen
  \bibfield  {author} {\bibinfo {author} {\bibfnamefont {F.~G. S.~L.}\ \bibnamefont {Brand{\~{a}}o}}\ and\ \bibinfo {author} {\bibfnamefont {K.~M.}\ \bibnamefont {Svore}},\ }\bibfield  {title} {\bibinfo {title} {Quantum speed-ups for solving semidefinite programs},\ }in\ \href {https://doi.org/10.1109/FOCS.2017.45} {\emph {\bibinfo {booktitle} {2017 IEEE 58th Annual Symposium on Foundations of Computer Science (FOCS)}}}\ (\bibinfo  {publisher} {IEEE Computer Society},\ \bibinfo {year} {2017})\ pp.\ \bibinfo {pages} {415--426}\BibitemShut {NoStop}%
\bibitem [{\citenamefont {Brand{\~{a}}o}\ \emph {et~al.}(2019)\citenamefont {Brand{\~{a}}o}, \citenamefont {Kalev}, \citenamefont {Li}, \citenamefont {Lin}, \citenamefont {Svore},\ and\ \citenamefont {Wu}}]{QASDP19}%
  \BibitemOpen
  \bibfield  {author} {\bibinfo {author} {\bibfnamefont {F.~G. S.~L.}\ \bibnamefont {Brand{\~{a}}o}}, \bibinfo {author} {\bibfnamefont {A.}~\bibnamefont {Kalev}}, \bibinfo {author} {\bibfnamefont {T.}~\bibnamefont {Li}}, \bibinfo {author} {\bibfnamefont {C.~Y.}\ \bibnamefont {Lin}}, \bibinfo {author} {\bibfnamefont {K.~M.}\ \bibnamefont {Svore}},\ and\ \bibinfo {author} {\bibfnamefont {X.}~\bibnamefont {Wu}},\ }\bibfield  {title} {\bibinfo {title} {Quantum {SDP} solvers: Large speed-ups, optimality, and applications to quantum learning},\ }in\ \href {https://doi.org/10.4230/LIPIcs.ICALP.2019.27} {\emph {\bibinfo {booktitle} {46th International Colloquium on Automata, Languages, and Programming, {ICALP} 2019, July 9-12, 2019, Patras, Greece}}},\ \bibinfo {series} {LIPIcs}, Vol.\ \bibinfo {volume} {132}\ (\bibinfo  {publisher} {Schloss Dagstuhl - Leibniz-Zentrum f{\"{u}}r Informatik},\ \bibinfo {year} {2019})\ pp.\ \bibinfo {pages} {27:1--27:14}\BibitemShut {NoStop}%
\bibitem [{\citenamefont {van Apeldoorn}\ \emph {et~al.}(2020)\citenamefont {van Apeldoorn}, \citenamefont {Gily{\ifmmode\acute{e}\else\'{e}\fi}n}, \citenamefont {Gribling},\ and\ \citenamefont {de~Wolf}}]{AGGW17}%
  \BibitemOpen
  \bibfield  {author} {\bibinfo {author} {\bibfnamefont {J.}~\bibnamefont {van Apeldoorn}}, \bibinfo {author} {\bibfnamefont {A.}~\bibnamefont {Gily{\ifmmode\acute{e}\else\'{e}\fi}n}}, \bibinfo {author} {\bibfnamefont {S.}~\bibnamefont {Gribling}},\ and\ \bibinfo {author} {\bibfnamefont {R.}~\bibnamefont {de~Wolf}},\ }\bibfield  {title} {\bibinfo {title} {Quantum {SDP}-solvers: Better upper and lower bounds},\ }\href {https://doi.org/10.22331/q-2020-02-14-230} {\bibfield  {journal} {\bibinfo  {journal} {Quantum}\ }\textbf {\bibinfo {volume} {4}},\ \bibinfo {pages} {230} (\bibinfo {year} {2020})}\BibitemShut {NoStop}%
\bibitem [{\citenamefont {Amin}\ \emph {et~al.}(2018)\citenamefont {Amin}, \citenamefont {Andriyash}, \citenamefont {Rolfe}, \citenamefont {Kulchytskyy},\ and\ \citenamefont {Melko}}]{Amin2018}%
  \BibitemOpen
  \bibfield  {author} {\bibinfo {author} {\bibfnamefont {M.~H.}\ \bibnamefont {Amin}}, \bibinfo {author} {\bibfnamefont {E.}~\bibnamefont {Andriyash}}, \bibinfo {author} {\bibfnamefont {J.}~\bibnamefont {Rolfe}}, \bibinfo {author} {\bibfnamefont {B.}~\bibnamefont {Kulchytskyy}},\ and\ \bibinfo {author} {\bibfnamefont {R.}~\bibnamefont {Melko}},\ }\bibfield  {title} {\bibinfo {title} {Quantum {B}oltzmann machine},\ }\href {https://doi.org/10.1103/PhysRevX.8.021050} {\bibfield  {journal} {\bibinfo  {journal} {Physical Review X}\ }\textbf {\bibinfo {volume} {8}},\ \bibinfo {pages} {021050} (\bibinfo {year} {2018})}\BibitemShut {NoStop}%
\bibitem [{\citenamefont {Benedetti}\ \emph {et~al.}(2017)\citenamefont {Benedetti}, \citenamefont {Realpe-G\'omez}, \citenamefont {Biswas},\ and\ \citenamefont {Perdomo-Ortiz}}]{Benedetti2017}%
  \BibitemOpen
  \bibfield  {author} {\bibinfo {author} {\bibfnamefont {M.}~\bibnamefont {Benedetti}}, \bibinfo {author} {\bibfnamefont {J.}~\bibnamefont {Realpe-G\'omez}}, \bibinfo {author} {\bibfnamefont {R.}~\bibnamefont {Biswas}},\ and\ \bibinfo {author} {\bibfnamefont {A.}~\bibnamefont {Perdomo-Ortiz}},\ }\bibfield  {title} {\bibinfo {title} {Quantum-assisted learning of hardware-embedded probabilistic graphical models},\ }\href {https://doi.org/10.1103/PhysRevX.7.041052} {\bibfield  {journal} {\bibinfo  {journal} {Physical Review X}\ }\textbf {\bibinfo {volume} {7}},\ \bibinfo {pages} {041052} (\bibinfo {year} {2017})}\BibitemShut {NoStop}%
\bibitem [{\citenamefont {Kieferov\'a}\ and\ \citenamefont {Wiebe}(2017)}]{Kieferova2017}%
  \BibitemOpen
  \bibfield  {author} {\bibinfo {author} {\bibfnamefont {M.}~\bibnamefont {Kieferov\'a}}\ and\ \bibinfo {author} {\bibfnamefont {N.}~\bibnamefont {Wiebe}},\ }\bibfield  {title} {\bibinfo {title} {Tomography and generative training with quantum {B}oltzmann machines},\ }\href {https://doi.org/10.1103/PhysRevA.96.062327} {\bibfield  {journal} {\bibinfo  {journal} {Physical Review A}\ }\textbf {\bibinfo {volume} {96}},\ \bibinfo {pages} {062327} (\bibinfo {year} {2017})}\BibitemShut {NoStop}%
\bibitem [{\citenamefont {Coopmans}\ and\ \citenamefont {Benedetti}(2024)}]{Coopmans2024}%
  \BibitemOpen
  \bibfield  {author} {\bibinfo {author} {\bibfnamefont {L.}~\bibnamefont {Coopmans}}\ and\ \bibinfo {author} {\bibfnamefont {M.}~\bibnamefont {Benedetti}},\ }\bibfield  {title} {\bibinfo {title} {On the sample complexity of quantum {B}oltzmann machine learning},\ }\href {https://doi.org/10.1038/s42005-024-01763-x} {\bibfield  {journal} {\bibinfo  {journal} {Communications Physics}\ }\textbf {\bibinfo {volume} {7}},\ \bibinfo {pages} {274} (\bibinfo {year} {2024})}\BibitemShut {NoStop}%
\bibitem [{\citenamefont {T\"uys\"uz}\ \emph {et~al.}(2024)\citenamefont {T\"uys\"uz}, \citenamefont {Demidik}, \citenamefont {Coopmans}, \citenamefont {Rinaldi}, \citenamefont {Croft}, \citenamefont {Haddad}, \citenamefont {Rosenkranz},\ and\ \citenamefont {Jansen}}]{tuysuz2024learninggeneratehighdimensionaldistributions}%
  \BibitemOpen
  \bibfield  {author} {\bibinfo {author} {\bibfnamefont {C.}~\bibnamefont {T\"uys\"uz}}, \bibinfo {author} {\bibfnamefont {M.}~\bibnamefont {Demidik}}, \bibinfo {author} {\bibfnamefont {L.}~\bibnamefont {Coopmans}}, \bibinfo {author} {\bibfnamefont {E.}~\bibnamefont {Rinaldi}}, \bibinfo {author} {\bibfnamefont {V.}~\bibnamefont {Croft}}, \bibinfo {author} {\bibfnamefont {Y.}~\bibnamefont {Haddad}}, \bibinfo {author} {\bibfnamefont {M.}~\bibnamefont {Rosenkranz}},\ and\ \bibinfo {author} {\bibfnamefont {K.}~\bibnamefont {Jansen}},\ }\href {https://arxiv.org/abs/2410.16363} {\bibinfo {title} {Learning to generate high-dimensional distributions with low-dimensional quantum {B}oltzmann machines}} (\bibinfo {year} {2024}),\ \Eprint {https://arxiv.org/abs/2410.16363} {arXiv:2410.16363 [quant-ph]} \BibitemShut {NoStop}%
\bibitem [{\citenamefont {Patel}\ \emph {et~al.}(2024)\citenamefont {Patel}, \citenamefont {Koch}, \citenamefont {Patel},\ and\ \citenamefont {Wilde}}]{patel2024quantumboltzmann}%
  \BibitemOpen
  \bibfield  {author} {\bibinfo {author} {\bibfnamefont {D.}~\bibnamefont {Patel}}, \bibinfo {author} {\bibfnamefont {D.}~\bibnamefont {Koch}}, \bibinfo {author} {\bibfnamefont {S.}~\bibnamefont {Patel}},\ and\ \bibinfo {author} {\bibfnamefont {M.~M.}\ \bibnamefont {Wilde}},\ }\href {https://arxiv.org/abs/2410.12935} {\bibinfo {title} {Quantum {B}oltzmann machine learning of ground-state energies}} (\bibinfo {year} {2024}),\ \Eprint {https://arxiv.org/abs/2410.12935} {arXiv:2410.12935 [quant-ph]} \BibitemShut {NoStop}%
\bibitem [{\citenamefont {Anshu}\ \emph {et~al.}(2021)\citenamefont {Anshu}, \citenamefont {Arunachalam}, \citenamefont {Kuwahara},\ and\ \citenamefont {Soleimanifar}}]{Anshu2021}%
  \BibitemOpen
  \bibfield  {author} {\bibinfo {author} {\bibfnamefont {A.}~\bibnamefont {Anshu}}, \bibinfo {author} {\bibfnamefont {S.}~\bibnamefont {Arunachalam}}, \bibinfo {author} {\bibfnamefont {T.}~\bibnamefont {Kuwahara}},\ and\ \bibinfo {author} {\bibfnamefont {M.}~\bibnamefont {Soleimanifar}},\ }\bibfield  {title} {\bibinfo {title} {Sample-efficient learning of interacting quantum systems},\ }\href {https://doi.org/10.1038/s41567-021-01232-0} {\bibfield  {journal} {\bibinfo  {journal} {Nature Physics}\ }\textbf {\bibinfo {volume} {17}},\ \bibinfo {pages} {931} (\bibinfo {year} {2021})}\BibitemShut {NoStop}%
\bibitem [{\citenamefont {Bakshi}\ \emph {et~al.}(2023)\citenamefont {Bakshi}, \citenamefont {Liu}, \citenamefont {Moitra},\ and\ \citenamefont {Tang}}]{bakshi2023learning}%
  \BibitemOpen
  \bibfield  {author} {\bibinfo {author} {\bibfnamefont {A.}~\bibnamefont {Bakshi}}, \bibinfo {author} {\bibfnamefont {A.}~\bibnamefont {Liu}}, \bibinfo {author} {\bibfnamefont {A.}~\bibnamefont {Moitra}},\ and\ \bibinfo {author} {\bibfnamefont {E.}~\bibnamefont {Tang}},\ }\href {https://arxiv.org/abs/2310.02243} {\bibinfo {title} {Learning quantum {H}amiltonians at any temperature in polynomial time}} (\bibinfo {year} {2023}),\ \Eprint {https://arxiv.org/abs/2310.02243} {arXiv:2310.02243 [quant-ph]} \BibitemShut {NoStop}%
\bibitem [{\citenamefont {Rouz\'e}\ \emph {et~al.}(2024)\citenamefont {Rouz\'e}, \citenamefont {Stilck~Franca}, \citenamefont {Onorati},\ and\ \citenamefont {Watson}}]{Rouze2024}%
  \BibitemOpen
  \bibfield  {author} {\bibinfo {author} {\bibfnamefont {C.}~\bibnamefont {Rouz\'e}}, \bibinfo {author} {\bibfnamefont {D.}~\bibnamefont {Stilck~Franca}}, \bibinfo {author} {\bibfnamefont {E.}~\bibnamefont {Onorati}},\ and\ \bibinfo {author} {\bibfnamefont {J.~D.}\ \bibnamefont {Watson}},\ }\bibfield  {title} {\bibinfo {title} {Efficient learning of ground and thermal states within phases of matter},\ }\href {https://doi.org/10.1038/s41467-024-51439-x} {\bibfield  {journal} {\bibinfo  {journal} {Nature Communications}\ }\textbf {\bibinfo {volume} {15}},\ \bibinfo {pages} {7755} (\bibinfo {year} {2024})}\BibitemShut {NoStop}%
\bibitem [{\citenamefont {Garc\'{\i}a-Pintos}\ \emph {et~al.}(2024)\citenamefont {Garc\'{\i}a-Pintos}, \citenamefont {Bharti}, \citenamefont {Bringewatt}, \citenamefont {Dehghani}, \citenamefont {Ehrenberg}, \citenamefont {Yunger~Halpern},\ and\ \citenamefont {Gorshkov}}]{GarciaPintos2024}%
  \BibitemOpen
  \bibfield  {author} {\bibinfo {author} {\bibfnamefont {L.~P.}\ \bibnamefont {Garc\'{\i}a-Pintos}}, \bibinfo {author} {\bibfnamefont {K.}~\bibnamefont {Bharti}}, \bibinfo {author} {\bibfnamefont {J.}~\bibnamefont {Bringewatt}}, \bibinfo {author} {\bibfnamefont {H.}~\bibnamefont {Dehghani}}, \bibinfo {author} {\bibfnamefont {A.}~\bibnamefont {Ehrenberg}}, \bibinfo {author} {\bibfnamefont {N.}~\bibnamefont {Yunger~Halpern}},\ and\ \bibinfo {author} {\bibfnamefont {A.~V.}\ \bibnamefont {Gorshkov}},\ }\bibfield  {title} {\bibinfo {title} {Estimation of {H}amiltonian parameters from thermal states},\ }\href {https://doi.org/10.1103/PhysRevLett.133.040802} {\bibfield  {journal} {\bibinfo  {journal} {Physical Review Letters}\ }\textbf {\bibinfo {volume} {133}},\ \bibinfo {pages} {040802} (\bibinfo {year} {2024})}\BibitemShut {NoStop}%
\bibitem [{\citenamefont {Abiuso}\ \emph {et~al.}(2024)\citenamefont {Abiuso}, \citenamefont {Sekatski}, \citenamefont {Calsamiglia},\ and\ \citenamefont {Perarnau-Llobet}}]{abiuso2024fundamentallimitsmetrologythermal}%
  \BibitemOpen
  \bibfield  {author} {\bibinfo {author} {\bibfnamefont {P.}~\bibnamefont {Abiuso}}, \bibinfo {author} {\bibfnamefont {P.}~\bibnamefont {Sekatski}}, \bibinfo {author} {\bibfnamefont {J.}~\bibnamefont {Calsamiglia}},\ and\ \bibinfo {author} {\bibfnamefont {M.}~\bibnamefont {Perarnau-Llobet}},\ }\href {https://arxiv.org/abs/2402.06582} {\bibinfo {title} {Fundamental limits of metrology at thermal equilibrium}} (\bibinfo {year} {2024}),\ \Eprint {https://arxiv.org/abs/2402.06582} {arXiv:2402.06582 [quant-ph]} \BibitemShut {NoStop}%
\bibitem [{\citenamefont {Chen}\ \emph {et~al.}(2023{\natexlab{a}})\citenamefont {Chen}, \citenamefont {Kastoryano}, \citenamefont {Brand{\~a}o},\ and\ \citenamefont {Gily{\'e}n}}]{chen2023quantumthermalstatepreparation}%
  \BibitemOpen
  \bibfield  {author} {\bibinfo {author} {\bibfnamefont {C.-F.}\ \bibnamefont {Chen}}, \bibinfo {author} {\bibfnamefont {M.~J.}\ \bibnamefont {Kastoryano}}, \bibinfo {author} {\bibfnamefont {F.~G. S.~L.}\ \bibnamefont {Brand{\~a}o}},\ and\ \bibinfo {author} {\bibfnamefont {A.}~\bibnamefont {Gily{\'e}n}},\ }\href {https://arxiv.org/abs/2303.18224} {\bibinfo {title} {Quantum thermal state preparation}} (\bibinfo {year} {2023}{\natexlab{a}}),\ \Eprint {https://arxiv.org/abs/2303.18224} {arXiv:2303.18224 [quant-ph]} \BibitemShut {NoStop}%
\bibitem [{\citenamefont {Chen}\ \emph {et~al.}(2023{\natexlab{b}})\citenamefont {Chen}, \citenamefont {Kastoryano},\ and\ \citenamefont {Gily{\'e}n}}]{chen2023efficient}%
  \BibitemOpen
  \bibfield  {author} {\bibinfo {author} {\bibfnamefont {C.-F.}\ \bibnamefont {Chen}}, \bibinfo {author} {\bibfnamefont {M.~J.}\ \bibnamefont {Kastoryano}},\ and\ \bibinfo {author} {\bibfnamefont {A.}~\bibnamefont {Gily{\'e}n}},\ }\href {https://arxiv.org/abs/2311.09207} {\bibinfo {title} {An efficient and exact noncommutative quantum {G}ibbs sampler}} (\bibinfo {year} {2023}{\natexlab{b}}),\ \Eprint {https://arxiv.org/abs/2311.09207} {arXiv:2311.09207 [quant-ph]} \BibitemShut {NoStop}%
\bibitem [{\citenamefont {Bergamaschi}\ \emph {et~al.}(2024)\citenamefont {Bergamaschi}, \citenamefont {Chen},\ and\ \citenamefont {Liu}}]{bergamaschi2024}%
  \BibitemOpen
  \bibfield  {author} {\bibinfo {author} {\bibfnamefont {T.}~\bibnamefont {Bergamaschi}}, \bibinfo {author} {\bibfnamefont {C.-F.}\ \bibnamefont {Chen}},\ and\ \bibinfo {author} {\bibfnamefont {Y.}~\bibnamefont {Liu}},\ }\href {https://arxiv.org/abs/2404.14639} {\bibinfo {title} {Quantum computational advantage with constant-temperature {G}ibbs sampling}} (\bibinfo {year} {2024}),\ \Eprint {https://arxiv.org/abs/2404.14639} {arXiv:2404.14639 [quant-ph]} \BibitemShut {NoStop}%
\bibitem [{\citenamefont {Chen}\ \emph {et~al.}(2024)\citenamefont {Chen}, \citenamefont {Li}, \citenamefont {Lu},\ and\ \citenamefont {Ying}}]{chen2024randomizedmethodsimulatinglindblad}%
  \BibitemOpen
  \bibfield  {author} {\bibinfo {author} {\bibfnamefont {H.}~\bibnamefont {Chen}}, \bibinfo {author} {\bibfnamefont {B.}~\bibnamefont {Li}}, \bibinfo {author} {\bibfnamefont {J.}~\bibnamefont {Lu}},\ and\ \bibinfo {author} {\bibfnamefont {L.}~\bibnamefont {Ying}},\ }\href {https://arxiv.org/abs/2407.06594} {\bibinfo {title} {A randomized method for simulating {L}indblad equations and thermal state preparation}} (\bibinfo {year} {2024}),\ \Eprint {https://arxiv.org/abs/2407.06594} {arXiv:2407.06594 [quant-ph]} \BibitemShut {NoStop}%
\bibitem [{\citenamefont {Rajakumar}\ and\ \citenamefont {Watson}(2024)}]{rajakumar2024gibbssamplinggivesquantum}%
  \BibitemOpen
  \bibfield  {author} {\bibinfo {author} {\bibfnamefont {J.}~\bibnamefont {Rajakumar}}\ and\ \bibinfo {author} {\bibfnamefont {J.~D.}\ \bibnamefont {Watson}},\ }\href {https://arxiv.org/abs/2408.01516} {\bibinfo {title} {Gibbs sampling gives quantum advantage at constant temperatures with {$O(1)$}-local {H}amiltonians}} (\bibinfo {year} {2024}),\ \Eprint {https://arxiv.org/abs/2408.01516} {arXiv:2408.01516 [quant-ph]} \BibitemShut {NoStop}%
\bibitem [{\citenamefont {Rouz{\'e}}\ \emph {et~al.}(2024)\citenamefont {Rouz{\'e}}, \citenamefont {Franca},\ and\ \citenamefont {Alhambra}}]{rouze2024efficient}%
  \BibitemOpen
  \bibfield  {author} {\bibinfo {author} {\bibfnamefont {C.}~\bibnamefont {Rouz{\'e}}}, \bibinfo {author} {\bibfnamefont {D.~S.}\ \bibnamefont {Franca}},\ and\ \bibinfo {author} {\bibfnamefont {{\'{A}}.~M.}\ \bibnamefont {Alhambra}},\ }\href {https://arxiv.org/abs/2403.12691} {\bibinfo {title} {Efficient thermalization and universal quantum computing with quantum {G}ibbs samplers}} (\bibinfo {year} {2024}),\ \Eprint {https://arxiv.org/abs/2403.12691} {arXiv:2403.12691 [quant-ph]} \BibitemShut {NoStop}%
\bibitem [{\citenamefont {Bakshi}\ \emph {et~al.}(2024)\citenamefont {Bakshi}, \citenamefont {Liu}, \citenamefont {Moitra},\ and\ \citenamefont {Tang}}]{bakshi2024hightemperature}%
  \BibitemOpen
  \bibfield  {author} {\bibinfo {author} {\bibfnamefont {A.}~\bibnamefont {Bakshi}}, \bibinfo {author} {\bibfnamefont {A.}~\bibnamefont {Liu}}, \bibinfo {author} {\bibfnamefont {A.}~\bibnamefont {Moitra}},\ and\ \bibinfo {author} {\bibfnamefont {E.}~\bibnamefont {Tang}},\ }\href {https://arxiv.org/abs/2403.16850} {\bibinfo {title} {High-temperature {G}ibbs states are unentangled and efficiently preparable}} (\bibinfo {year} {2024}),\ \Eprint {https://arxiv.org/abs/2403.16850} {arXiv:2403.16850 [quant-ph]} \BibitemShut {NoStop}%
\bibitem [{\citenamefont {Ding}\ \emph {et~al.}(2024)\citenamefont {Ding}, \citenamefont {Li}, \citenamefont {Lin},\ and\ \citenamefont {Zhang}}]{ding2024polynomial}%
  \BibitemOpen
  \bibfield  {author} {\bibinfo {author} {\bibfnamefont {Z.}~\bibnamefont {Ding}}, \bibinfo {author} {\bibfnamefont {B.}~\bibnamefont {Li}}, \bibinfo {author} {\bibfnamefont {L.}~\bibnamefont {Lin}},\ and\ \bibinfo {author} {\bibfnamefont {R.}~\bibnamefont {Zhang}},\ }\href {https://arxiv.org/abs/2410.01206} {\bibinfo {title} {Polynomial-time preparation of low-temperature {G}ibbs states for 2{D} toric code}} (\bibinfo {year} {2024}),\ \Eprint {https://arxiv.org/abs/2410.01206} {arXiv:2410.01206 [quant-ph]} \BibitemShut {NoStop}%
\bibitem [{\citenamefont {Jarzyna}\ and\ \citenamefont {Kolodynski}(2020)}]{Jarzyna2020}%
  \BibitemOpen
  \bibfield  {author} {\bibinfo {author} {\bibfnamefont {M.}~\bibnamefont {Jarzyna}}\ and\ \bibinfo {author} {\bibfnamefont {J.}~\bibnamefont {Kolodynski}},\ }\bibfield  {title} {\bibinfo {title} {Geometric approach to quantum statistical inference},\ }\href {https://doi.org/10.1109/JSAIT.2020.3017469} {\bibfield  {journal} {\bibinfo  {journal} {IEEE Journal on Selected Areas in Information Theory}\ }\textbf {\bibinfo {volume} {1}},\ \bibinfo {pages} {367} (\bibinfo {year} {2020})}\BibitemShut {NoStop}%
\bibitem [{\citenamefont {Amari}(1998)}]{amari1998natural}%
  \BibitemOpen
  \bibfield  {author} {\bibinfo {author} {\bibfnamefont {S.-I.}\ \bibnamefont {Amari}},\ }\bibfield  {title} {\bibinfo {title} {Natural gradient works efficiently in learning},\ }\href {https://doi.org/10.1162/089976698300017746} {\bibfield  {journal} {\bibinfo  {journal} {Neural Computation}\ }\textbf {\bibinfo {volume} {10}},\ \bibinfo {pages} {251} (\bibinfo {year} {1998})}\BibitemShut {NoStop}%
\bibitem [{\citenamefont {Neyshabur}\ \emph {et~al.}(2015)\citenamefont {Neyshabur}, \citenamefont {Salakhutdinov},\ and\ \citenamefont {Srebro}}]{neyshabur2015path}%
  \BibitemOpen
  \bibfield  {author} {\bibinfo {author} {\bibfnamefont {B.}~\bibnamefont {Neyshabur}}, \bibinfo {author} {\bibfnamefont {R.~R.}\ \bibnamefont {Salakhutdinov}},\ and\ \bibinfo {author} {\bibfnamefont {N.}~\bibnamefont {Srebro}},\ }\bibfield  {title} {\bibinfo {title} {Path-{SGD}: Path-normalized optimization in deep neural networks},\ }\href@noop {} {\bibfield  {journal} {\bibinfo  {journal} {Advances in Neural Information Processing Systems}\ }\textbf {\bibinfo {volume} {28}} (\bibinfo {year} {2015})}\BibitemShut {NoStop}%
\bibitem [{\citenamefont {Stokes}\ \emph {et~al.}(2020)\citenamefont {Stokes}, \citenamefont {Izaac}, \citenamefont {Killoran},\ and\ \citenamefont {Carleo}}]{Stokes2020quantumnatural}%
  \BibitemOpen
  \bibfield  {author} {\bibinfo {author} {\bibfnamefont {J.}~\bibnamefont {Stokes}}, \bibinfo {author} {\bibfnamefont {J.}~\bibnamefont {Izaac}}, \bibinfo {author} {\bibfnamefont {N.}~\bibnamefont {Killoran}},\ and\ \bibinfo {author} {\bibfnamefont {G.}~\bibnamefont {Carleo}},\ }\bibfield  {title} {\bibinfo {title} {Quantum natural gradient},\ }\href {https://doi.org/10.22331/q-2020-05-25-269} {\bibfield  {journal} {\bibinfo  {journal} {Quantum}\ }\textbf {\bibinfo {volume} {4}},\ \bibinfo {pages} {269} (\bibinfo {year} {2020})}\BibitemShut {NoStop}%
\bibitem [{\citenamefont {McClean}\ \emph {et~al.}(2018)\citenamefont {McClean}, \citenamefont {Boixo}, \citenamefont {Smelyanskiy}, \citenamefont {Babbush},\ and\ \citenamefont {Neven}}]{McClean_2018}%
  \BibitemOpen
  \bibfield  {author} {\bibinfo {author} {\bibfnamefont {J.~R.}\ \bibnamefont {McClean}}, \bibinfo {author} {\bibfnamefont {S.}~\bibnamefont {Boixo}}, \bibinfo {author} {\bibfnamefont {V.~N.}\ \bibnamefont {Smelyanskiy}}, \bibinfo {author} {\bibfnamefont {R.}~\bibnamefont {Babbush}},\ and\ \bibinfo {author} {\bibfnamefont {H.}~\bibnamefont {Neven}},\ }\bibfield  {title} {\bibinfo {title} {Barren plateaus in quantum neural network training landscapes},\ }\href {https://doi.org/10.1038/s41467-018-07090-4} {\bibfield  {journal} {\bibinfo  {journal} {Nature Communications}\ }\textbf {\bibinfo {volume} {9}},\ \bibinfo {pages} {4812} (\bibinfo {year} {2018})}\BibitemShut {NoStop}%
\bibitem [{\citenamefont {Marrero}\ \emph {et~al.}(2021)\citenamefont {Marrero}, \citenamefont {Kieferov{\'a}},\ and\ \citenamefont {Wiebe}}]{marrero2021entanglement}%
  \BibitemOpen
  \bibfield  {author} {\bibinfo {author} {\bibfnamefont {C.~O.}\ \bibnamefont {Marrero}}, \bibinfo {author} {\bibfnamefont {M.}~\bibnamefont {Kieferov{\'a}}},\ and\ \bibinfo {author} {\bibfnamefont {N.}~\bibnamefont {Wiebe}},\ }\bibfield  {title} {\bibinfo {title} {Entanglement-induced barren plateaus},\ }\href {https://doi.org/10.1103/PRXQuantum.2.040316} {\bibfield  {journal} {\bibinfo  {journal} {PRX Quantum}\ }\textbf {\bibinfo {volume} {2}},\ \bibinfo {pages} {040316} (\bibinfo {year} {2021})}\BibitemShut {NoStop}%
\bibitem [{\citenamefont {Arrasmith}\ \emph {et~al.}(2022)\citenamefont {Arrasmith}, \citenamefont {Holmes}, \citenamefont {Cerezo},\ and\ \citenamefont {Coles}}]{arrasmith2022equivalence}%
  \BibitemOpen
  \bibfield  {author} {\bibinfo {author} {\bibfnamefont {A.}~\bibnamefont {Arrasmith}}, \bibinfo {author} {\bibfnamefont {Z.}~\bibnamefont {Holmes}}, \bibinfo {author} {\bibfnamefont {M.}~\bibnamefont {Cerezo}},\ and\ \bibinfo {author} {\bibfnamefont {P.~J.}\ \bibnamefont {Coles}},\ }\bibfield  {title} {\bibinfo {title} {Equivalence of quantum barren plateaus to cost concentration and narrow gorges},\ }\href {https://doi.org/10.1088/2058-9565/ac7d06} {\bibfield  {journal} {\bibinfo  {journal} {Quantum Science and Technology}\ }\textbf {\bibinfo {volume} {7}},\ \bibinfo {pages} {045015} (\bibinfo {year} {2022})}\BibitemShut {NoStop}%
\bibitem [{\citenamefont {Holmes}\ \emph {et~al.}(2022)\citenamefont {Holmes}, \citenamefont {Sharma}, \citenamefont {Cerezo},\ and\ \citenamefont {Coles}}]{holmes2022connecting}%
  \BibitemOpen
  \bibfield  {author} {\bibinfo {author} {\bibfnamefont {Z.}~\bibnamefont {Holmes}}, \bibinfo {author} {\bibfnamefont {K.}~\bibnamefont {Sharma}}, \bibinfo {author} {\bibfnamefont {M.}~\bibnamefont {Cerezo}},\ and\ \bibinfo {author} {\bibfnamefont {P.~J.}\ \bibnamefont {Coles}},\ }\bibfield  {title} {\bibinfo {title} {Connecting ansatz expressibility to gradient magnitudes and barren plateaus},\ }\href {https://doi.org/10.1103/PRXQuantum.3.010313} {\bibfield  {journal} {\bibinfo  {journal} {PRX Quantum}\ }\textbf {\bibinfo {volume} {3}},\ \bibinfo {pages} {010313} (\bibinfo {year} {2022})}\BibitemShut {NoStop}%
\bibitem [{\citenamefont {Fontana}\ \emph {et~al.}(2024)\citenamefont {Fontana}, \citenamefont {Herman}, \citenamefont {Chakrabarti}, \citenamefont {Kumar}, \citenamefont {Yalovetzky}, \citenamefont {Heredge}, \citenamefont {Sureshbabu},\ and\ \citenamefont {Pistoia}}]{Fontana2024}%
  \BibitemOpen
  \bibfield  {author} {\bibinfo {author} {\bibfnamefont {E.}~\bibnamefont {Fontana}}, \bibinfo {author} {\bibfnamefont {D.}~\bibnamefont {Herman}}, \bibinfo {author} {\bibfnamefont {S.}~\bibnamefont {Chakrabarti}}, \bibinfo {author} {\bibfnamefont {N.}~\bibnamefont {Kumar}}, \bibinfo {author} {\bibfnamefont {R.}~\bibnamefont {Yalovetzky}}, \bibinfo {author} {\bibfnamefont {J.}~\bibnamefont {Heredge}}, \bibinfo {author} {\bibfnamefont {S.~H.}\ \bibnamefont {Sureshbabu}},\ and\ \bibinfo {author} {\bibfnamefont {M.}~\bibnamefont {Pistoia}},\ }\bibfield  {title} {\bibinfo {title} {Characterizing barren plateaus in quantum ans\"atze with the adjoint representation},\ }\href {https://doi.org/10.1038/s41467-024-49910-w} {\bibfield  {journal} {\bibinfo  {journal} {Nature Communications}\ }\textbf {\bibinfo {volume} {15}},\ \bibinfo {pages} {7171} (\bibinfo {year} {2024})}\BibitemShut {NoStop}%
\bibitem [{\citenamefont {Ragone}\ \emph {et~al.}(2024)\citenamefont {Ragone}, \citenamefont {Bakalov}, \citenamefont {Sauvage}, \citenamefont {Kemper}, \citenamefont {Ortiz~Marrero}, \citenamefont {Larocca},\ and\ \citenamefont {Cerezo}}]{Ragone2024}%
  \BibitemOpen
  \bibfield  {author} {\bibinfo {author} {\bibfnamefont {M.}~\bibnamefont {Ragone}}, \bibinfo {author} {\bibfnamefont {B.~N.}\ \bibnamefont {Bakalov}}, \bibinfo {author} {\bibfnamefont {F.}~\bibnamefont {Sauvage}}, \bibinfo {author} {\bibfnamefont {A.~F.}\ \bibnamefont {Kemper}}, \bibinfo {author} {\bibfnamefont {C.}~\bibnamefont {Ortiz~Marrero}}, \bibinfo {author} {\bibfnamefont {M.}~\bibnamefont {Larocca}},\ and\ \bibinfo {author} {\bibfnamefont {M.}~\bibnamefont {Cerezo}},\ }\bibfield  {title} {\bibinfo {title} {A {L}ie algebraic theory of barren plateaus for deep parameterized quantum circuits},\ }\href {https://doi.org/10.1038/s41467-024-49909-3} {\bibfield  {journal} {\bibinfo  {journal} {Nature Communications}\ }\textbf {\bibinfo {volume} {15}},\ \bibinfo {pages} {7172} (\bibinfo {year} {2024})}\BibitemShut {NoStop}%
\bibitem [{\citenamefont {Haug}\ \emph {et~al.}(2021)\citenamefont {Haug}, \citenamefont {Bharti},\ and\ \citenamefont {Kim}}]{Haug2021}%
  \BibitemOpen
  \bibfield  {author} {\bibinfo {author} {\bibfnamefont {T.}~\bibnamefont {Haug}}, \bibinfo {author} {\bibfnamefont {K.}~\bibnamefont {Bharti}},\ and\ \bibinfo {author} {\bibfnamefont {M.}~\bibnamefont {Kim}},\ }\bibfield  {title} {\bibinfo {title} {Capacity and quantum geometry of parametrized quantum circuits},\ }\href {https://doi.org/10.1103/PRXQuantum.2.040309} {\bibfield  {journal} {\bibinfo  {journal} {PRX Quantum}\ }\textbf {\bibinfo {volume} {2}},\ \bibinfo {pages} {040309} (\bibinfo {year} {2021})}\BibitemShut {NoStop}%
\bibitem [{\citenamefont {Sbahi}\ \emph {et~al.}(2022)\citenamefont {Sbahi}, \citenamefont {Martinez}, \citenamefont {Patel}, \citenamefont {Saberi}, \citenamefont {Yoo}, \citenamefont {Roeder},\ and\ \citenamefont {Verdon}}]{sbahi2022provablyefficientvariationalgenerative}%
  \BibitemOpen
  \bibfield  {author} {\bibinfo {author} {\bibfnamefont {F.~M.}\ \bibnamefont {Sbahi}}, \bibinfo {author} {\bibfnamefont {A.~J.}\ \bibnamefont {Martinez}}, \bibinfo {author} {\bibfnamefont {S.}~\bibnamefont {Patel}}, \bibinfo {author} {\bibfnamefont {D.}~\bibnamefont {Saberi}}, \bibinfo {author} {\bibfnamefont {J.~H.}\ \bibnamefont {Yoo}}, \bibinfo {author} {\bibfnamefont {G.}~\bibnamefont {Roeder}},\ and\ \bibinfo {author} {\bibfnamefont {G.}~\bibnamefont {Verdon}},\ }\href {https://arxiv.org/abs/2206.04663v1} {\bibinfo {title} {Provably efficient variational generative modeling of quantum many-body systems via quantum-probabilistic information geometry}} (\bibinfo {year} {2022}),\ \Eprint {https://arxiv.org/abs/2206.04663v1} {arXiv:2206.04663v1 [quant-ph]} \BibitemShut {NoStop}%
\bibitem [{\citenamefont {Bengtsson}\ and\ \citenamefont {Zyczkowski}(2006)}]{Bengtsson2006}%
  \BibitemOpen
  \bibfield  {author} {\bibinfo {author} {\bibfnamefont {I.}~\bibnamefont {Bengtsson}}\ and\ \bibinfo {author} {\bibfnamefont {K.}~\bibnamefont {Zyczkowski}},\ }\href {https://doi.org/10.1017/CBO9780511535048} {\emph {\bibinfo {title} {Geometry of Quantum States: An Introduction to Quantum Entanglement}}}\ (\bibinfo  {publisher} {Cambridge University Press},\ \bibinfo {year} {2006})\BibitemShut {NoStop}%
\bibitem [{\citenamefont {Liu}\ \emph {et~al.}(2020)\citenamefont {Liu}, \citenamefont {Yuan}, \citenamefont {Lu},\ and\ \citenamefont {Wang}}]{Liu2019}%
  \BibitemOpen
  \bibfield  {author} {\bibinfo {author} {\bibfnamefont {J.}~\bibnamefont {Liu}}, \bibinfo {author} {\bibfnamefont {H.}~\bibnamefont {Yuan}}, \bibinfo {author} {\bibfnamefont {X.-M.}\ \bibnamefont {Lu}},\ and\ \bibinfo {author} {\bibfnamefont {X.}~\bibnamefont {Wang}},\ }\bibfield  {title} {\bibinfo {title} {Quantum {F}isher information matrix and multiparameter estimation},\ }\href {https://doi.org/10.1088/1751-8121/ab5d4d} {\bibfield  {journal} {\bibinfo  {journal} {Journal of Physics A: Mathematical and Theoretical}\ }\textbf {\bibinfo {volume} {53}},\ \bibinfo {pages} {023001} (\bibinfo {year} {2020})}\BibitemShut {NoStop}%
\bibitem [{\citenamefont {Sidhu}\ and\ \citenamefont {Kok}(2020)}]{Sidhu2020}%
  \BibitemOpen
  \bibfield  {author} {\bibinfo {author} {\bibfnamefont {J.~S.}\ \bibnamefont {Sidhu}}\ and\ \bibinfo {author} {\bibfnamefont {P.}~\bibnamefont {Kok}},\ }\bibfield  {title} {\bibinfo {title} {Geometric perspective on quantum parameter estimation},\ }\href {https://doi.org/10.1116/1.5119961} {\bibfield  {journal} {\bibinfo  {journal} {AVS Quantum Science}\ }\textbf {\bibinfo {volume} {2}},\ \bibinfo {pages} {014701} (\bibinfo {year} {2020})}\BibitemShut {NoStop}%
\bibitem [{\citenamefont {Meyer}(2021)}]{Meyer2021fisherinformationin}%
  \BibitemOpen
  \bibfield  {author} {\bibinfo {author} {\bibfnamefont {J.~J.}\ \bibnamefont {Meyer}},\ }\bibfield  {title} {\bibinfo {title} {Fisher information in noisy intermediate-scale quantum applications},\ }\href {https://doi.org/10.22331/q-2021-09-09-539} {\bibfield  {journal} {\bibinfo  {journal} {{Quantum}}\ }\textbf {\bibinfo {volume} {5}},\ \bibinfo {pages} {539} (\bibinfo {year} {2021})}\BibitemShut {NoStop}%
\bibitem [{\citenamefont {Scandi}\ \emph {et~al.}(2024)\citenamefont {Scandi}, \citenamefont {Abiuso}, \citenamefont {Surace},\ and\ \citenamefont {Santis}}]{scandi2024quantumfisherinformationdynamical}%
  \BibitemOpen
  \bibfield  {author} {\bibinfo {author} {\bibfnamefont {M.}~\bibnamefont {Scandi}}, \bibinfo {author} {\bibfnamefont {P.}~\bibnamefont {Abiuso}}, \bibinfo {author} {\bibfnamefont {J.}~\bibnamefont {Surace}},\ and\ \bibinfo {author} {\bibfnamefont {D.~D.}\ \bibnamefont {Santis}},\ }\href {https://arxiv.org/abs/2304.14984} {\bibinfo {title} {Quantum {F}isher information and its dynamical nature}} (\bibinfo {year} {2024}),\ \Eprint {https://arxiv.org/abs/2304.14984} {arXiv:2304.14984 [quant-ph]} \BibitemShut {NoStop}%
\bibitem [{\citenamefont {Helstrom}(1968)}]{Helstrom1968}%
  \BibitemOpen
  \bibfield  {author} {\bibinfo {author} {\bibfnamefont {C.}~\bibnamefont {Helstrom}},\ }\bibfield  {title} {\bibinfo {title} {The minimum variance of estimates in quantum signal detection},\ }\href {https://doi.org/10.1109/TIT.1968.1054108} {\bibfield  {journal} {\bibinfo  {journal} {IEEE Transactions on Information Theory}\ }\textbf {\bibinfo {volume} {14}},\ \bibinfo {pages} {234} (\bibinfo {year} {1968})}\BibitemShut {NoStop}%
\bibitem [{\citenamefont {Katariya}\ and\ \citenamefont {Wilde}(2021)}]{Katariya2021}%
  \BibitemOpen
  \bibfield  {author} {\bibinfo {author} {\bibfnamefont {V.}~\bibnamefont {Katariya}}\ and\ \bibinfo {author} {\bibfnamefont {M.~M.}\ \bibnamefont {Wilde}},\ }\bibfield  {title} {\bibinfo {title} {Geometric distinguishability measures limit quantum channel estimation and discrimination},\ }\href {https://doi.org/10.1007/s11128-021-02992-7} {\bibfield  {journal} {\bibinfo  {journal} {Quantum Information Processing}\ }\textbf {\bibinfo {volume} {20}},\ \bibinfo {pages} {78} (\bibinfo {year} {2021})}\BibitemShut {NoStop}%
\bibitem [{\citenamefont {Bures}(1969)}]{bures1969extension}%
  \BibitemOpen
  \bibfield  {author} {\bibinfo {author} {\bibfnamefont {D.}~\bibnamefont {Bures}},\ }\bibfield  {title} {\bibinfo {title} {An extension of {K}akutani's theorem on infinite product measures to the tensor product of semifinite {$W$}*-algebras},\ }\href {https://doi.org/10.1090/S0002-9947-1969-0236719-2} {\bibfield  {journal} {\bibinfo  {journal} {Transactions of the American Mathematical Society}\ }\textbf {\bibinfo {volume} {135}},\ \bibinfo {pages} {199} (\bibinfo {year} {1969})}\BibitemShut {NoStop}%
\bibitem [{\citenamefont {Uhlmann}(1976)}]{Uhlmann1976}%
  \BibitemOpen
  \bibfield  {author} {\bibinfo {author} {\bibfnamefont {A.}~\bibnamefont {Uhlmann}},\ }\bibfield  {title} {\bibinfo {title} {The ``transition probability'' in the state space of a *-algebra},\ }\href {https://doi.org/10.1016/0034-4877(76)90060-4} {\bibfield  {journal} {\bibinfo  {journal} {Reports on Mathematical Physics}\ }\textbf {\bibinfo {volume} {9}},\ \bibinfo {pages} {273} (\bibinfo {year} {1976})}\BibitemShut {NoStop}%
\bibitem [{\citenamefont {Umegaki}(1962)}]{umegaki1962ConditionalExpectationOperator}%
  \BibitemOpen
  \bibfield  {author} {\bibinfo {author} {\bibfnamefont {H.}~\bibnamefont {Umegaki}},\ }\bibfield  {title} {\bibinfo {title} {Conditional expectation in an operator algebra, {{IV}} (entropy and information)},\ }\href {https://doi.org/10.2996/kmj/1138844604} {\bibfield  {journal} {\bibinfo  {journal} {Kodai Mathematical Journal}\ }\textbf {\bibinfo {volume} {14}},\ \bibinfo {pages} {59} (\bibinfo {year} {1962})}\BibitemShut {NoStop}%
\bibitem [{\citenamefont {Lindblad}(1975)}]{Lindblad1975}%
  \BibitemOpen
  \bibfield  {author} {\bibinfo {author} {\bibfnamefont {G.}~\bibnamefont {Lindblad}},\ }\bibfield  {title} {\bibinfo {title} {Completely positive maps and entropy inequalities},\ }\href {https://doi.org/10.1007/BF01609396} {\bibfield  {journal} {\bibinfo  {journal} {Communications in Mathematical Physics}\ }\textbf {\bibinfo {volume} {40}},\ \bibinfo {pages} {147} (\bibinfo {year} {1975})}\BibitemShut {NoStop}%
\bibitem [{\citenamefont {Hiai}\ and\ \citenamefont {Petz}(1991)}]{hiai1991ProperFormulaRelative}%
  \BibitemOpen
  \bibfield  {author} {\bibinfo {author} {\bibfnamefont {F.}~\bibnamefont {Hiai}}\ and\ \bibinfo {author} {\bibfnamefont {D.}~\bibnamefont {Petz}},\ }\bibfield  {title} {\bibinfo {title} {The proper formula for relative entropy and its asymptotics in quantum probability},\ }\href {https://doi.org/10.1007/BF02100287} {\bibfield  {journal} {\bibinfo  {journal} {Communications in Mathematical Physics}\ }\textbf {\bibinfo {volume} {143}},\ \bibinfo {pages} {99} (\bibinfo {year} {1991})}\BibitemShut {NoStop}%
\bibitem [{\citenamefont {Nagaoka}\ and\ \citenamefont {Ogawa}(2000)}]{nagaoka2000StrongConverseSteins}%
  \BibitemOpen
  \bibfield  {author} {\bibinfo {author} {\bibfnamefont {H.}~\bibnamefont {Nagaoka}}\ and\ \bibinfo {author} {\bibfnamefont {T.}~\bibnamefont {Ogawa}},\ }\bibfield  {title} {\bibinfo {title} {Strong converse and {{Stein}}'s lemma in quantum hypothesis testing},\ }\href {https://doi.org/10.1109/18.887855} {\bibfield  {journal} {\bibinfo  {journal} {IEEE Transactions on Information Theory}\ }\textbf {\bibinfo {volume} {46}},\ \bibinfo {pages} {2428} (\bibinfo {year} {2000})}\BibitemShut {NoStop}%
\bibitem [{\citenamefont {Hastings}(2007)}]{Hastings2007}%
  \BibitemOpen
  \bibfield  {author} {\bibinfo {author} {\bibfnamefont {M.~B.}\ \bibnamefont {Hastings}},\ }\bibfield  {title} {\bibinfo {title} {Quantum belief propagation: An algorithm for thermal quantum systems},\ }\href {https://doi.org/10.1103/PhysRevB.76.201102} {\bibfield  {journal} {\bibinfo  {journal} {Physical Review B}\ }\textbf {\bibinfo {volume} {76}},\ \bibinfo {pages} {201102} (\bibinfo {year} {2007})}\BibitemShut {NoStop}%
\bibitem [{\citenamefont {Kim}(2012)}]{Kim2012}%
  \BibitemOpen
  \bibfield  {author} {\bibinfo {author} {\bibfnamefont {I.~H.}\ \bibnamefont {Kim}},\ }\bibfield  {title} {\bibinfo {title} {Perturbative analysis of topological entanglement entropy from conditional independence},\ }\href {https://doi.org/10.1103/PhysRevB.86.245116} {\bibfield  {journal} {\bibinfo  {journal} {Physical Review B}\ }\textbf {\bibinfo {volume} {86}},\ \bibinfo {pages} {245116} (\bibinfo {year} {2012})}\BibitemShut {NoStop}%
\bibitem [{\citenamefont {Anshu}\ \emph {et~al.}(2020)\citenamefont {Anshu}, \citenamefont {Arunachalam}, \citenamefont {Kuwahara},\ and\ \citenamefont {Soleimanifar}}]{Anshu2020arXiv}%
  \BibitemOpen
  \bibfield  {author} {\bibinfo {author} {\bibfnamefont {A.}~\bibnamefont {Anshu}}, \bibinfo {author} {\bibfnamefont {S.}~\bibnamefont {Arunachalam}}, \bibinfo {author} {\bibfnamefont {T.}~\bibnamefont {Kuwahara}},\ and\ \bibinfo {author} {\bibfnamefont {M.}~\bibnamefont {Soleimanifar}},\ }\href {https://arxiv.org/abs/2004.07266v1} {\bibinfo {title} {Sample-efficient learning of quantum many-body systems}} (\bibinfo {year} {2020}),\ \Eprint {https://arxiv.org/abs/2004.07266v1} {arXiv:2004.07266v1 [quant-ph]} \BibitemShut {NoStop}%
\bibitem [{\citenamefont {Lloyd}(1996)}]{lloyd1996universal}%
  \BibitemOpen
  \bibfield  {author} {\bibinfo {author} {\bibfnamefont {S.}~\bibnamefont {Lloyd}},\ }\bibfield  {title} {\bibinfo {title} {Universal quantum simulators},\ }\href {https://doi.org/10.1126/science.273.5278.1073} {\bibfield  {journal} {\bibinfo  {journal} {Science}\ }\textbf {\bibinfo {volume} {273}},\ \bibinfo {pages} {1073} (\bibinfo {year} {1996})}\BibitemShut {NoStop}%
\bibitem [{\citenamefont {Childs}\ \emph {et~al.}(2018)\citenamefont {Childs}, \citenamefont {Maslov}, \citenamefont {Nam}, \citenamefont {Ross},\ and\ \citenamefont {Su}}]{childs2018toward}%
  \BibitemOpen
  \bibfield  {author} {\bibinfo {author} {\bibfnamefont {A.~M.}\ \bibnamefont {Childs}}, \bibinfo {author} {\bibfnamefont {D.}~\bibnamefont {Maslov}}, \bibinfo {author} {\bibfnamefont {Y.}~\bibnamefont {Nam}}, \bibinfo {author} {\bibfnamefont {N.~J.}\ \bibnamefont {Ross}},\ and\ \bibinfo {author} {\bibfnamefont {Y.}~\bibnamefont {Su}},\ }\bibfield  {title} {\bibinfo {title} {Toward the first quantum simulation with quantum speedup},\ }\href {https://doi.org/10.1073/pnas.1801723115} {\bibfield  {journal} {\bibinfo  {journal} {Proceedings of the National Academy of Sciences}\ }\textbf {\bibinfo {volume} {115}},\ \bibinfo {pages} {9456} (\bibinfo {year} {2018})}\BibitemShut {NoStop}%
\bibitem [{\citenamefont {Cleve}\ \emph {et~al.}(1998)\citenamefont {Cleve}, \citenamefont {Ekert}, \citenamefont {Macchiavello},\ and\ \citenamefont {Mosca}}]{Cleve1998}%
  \BibitemOpen
  \bibfield  {author} {\bibinfo {author} {\bibfnamefont {R.}~\bibnamefont {Cleve}}, \bibinfo {author} {\bibfnamefont {A.}~\bibnamefont {Ekert}}, \bibinfo {author} {\bibfnamefont {C.}~\bibnamefont {Macchiavello}},\ and\ \bibinfo {author} {\bibfnamefont {M.}~\bibnamefont {Mosca}},\ }\bibfield  {title} {\bibinfo {title} {Quantum algorithms revisited},\ }\href {https://doi.org/10.1098/rspa.1998.0164} {\bibfield  {journal} {\bibinfo  {journal} {Proceedings of the Royal Society A}\ }\textbf {\bibinfo {volume} {454}},\ \bibinfo {pages} {339} (\bibinfo {year} {1998})}\BibitemShut {NoStop}%
\bibitem [{\citenamefont {Low}\ and\ \citenamefont {Chuang}(2019)}]{Low2019hamiltonian}%
  \BibitemOpen
  \bibfield  {author} {\bibinfo {author} {\bibfnamefont {G.~H.}\ \bibnamefont {Low}}\ and\ \bibinfo {author} {\bibfnamefont {I.~L.}\ \bibnamefont {Chuang}},\ }\bibfield  {title} {\bibinfo {title} {Hamiltonian simulation by qubitization},\ }\href {https://doi.org/10.22331/q-2019-07-12-163} {\bibfield  {journal} {\bibinfo  {journal} {{Quantum}}\ }\textbf {\bibinfo {volume} {3}},\ \bibinfo {pages} {163} (\bibinfo {year} {2019})}\BibitemShut {NoStop}%
\bibitem [{\citenamefont {Gily\'{e}n}\ \emph {et~al.}(2019)\citenamefont {Gily\'{e}n}, \citenamefont {Su}, \citenamefont {Low},\ and\ \citenamefont {Wiebe}}]{Gilyen2019}%
  \BibitemOpen
  \bibfield  {author} {\bibinfo {author} {\bibfnamefont {A.}~\bibnamefont {Gily\'{e}n}}, \bibinfo {author} {\bibfnamefont {Y.}~\bibnamefont {Su}}, \bibinfo {author} {\bibfnamefont {G.~H.}\ \bibnamefont {Low}},\ and\ \bibinfo {author} {\bibfnamefont {N.}~\bibnamefont {Wiebe}},\ }\bibfield  {title} {\bibinfo {title} {Quantum singular value transformation and beyond: exponential improvements for quantum matrix arithmetics},\ }in\ \href {https://doi.org/10.1145/3313276.3316366} {\emph {\bibinfo {booktitle} {Proceedings of the 51st Annual ACM SIGACT Symposium on Theory of Computing}}},\ \bibinfo {series and number} {STOC 2019}\ (\bibinfo  {publisher} {Association for Computing Machinery},\ \bibinfo {address} {New York, NY, USA},\ \bibinfo {year} {2019})\ pp.\ \bibinfo {pages} {193--204}\BibitemShut {NoStop}%
\bibitem [{\citenamefont {Braunstein}\ and\ \citenamefont {Caves}(1994)}]{braunsteincaves1994}%
  \BibitemOpen
  \bibfield  {author} {\bibinfo {author} {\bibfnamefont {S.~L.}\ \bibnamefont {Braunstein}}\ and\ \bibinfo {author} {\bibfnamefont {C.~M.}\ \bibnamefont {Caves}},\ }\bibfield  {title} {\bibinfo {title} {Statistical distance and the geometry of quantum states},\ }\href {https://doi.org/10.1103/PhysRevLett.72.3439} {\bibfield  {journal} {\bibinfo  {journal} {Physical Review Letters}\ }\textbf {\bibinfo {volume} {72}},\ \bibinfo {pages} {3439} (\bibinfo {year} {1994})}\BibitemShut {NoStop}%
\bibitem [{\citenamefont {Desjardins}\ \emph {et~al.}(2013)\citenamefont {Desjardins}, \citenamefont {Pascanu}, \citenamefont {Courville},\ and\ \citenamefont {Bengio}}]{desjardins2013metricfreenaturalgradientjointtraining}%
  \BibitemOpen
  \bibfield  {author} {\bibinfo {author} {\bibfnamefont {G.}~\bibnamefont {Desjardins}}, \bibinfo {author} {\bibfnamefont {R.}~\bibnamefont {Pascanu}}, \bibinfo {author} {\bibfnamefont {A.}~\bibnamefont {Courville}},\ and\ \bibinfo {author} {\bibfnamefont {Y.}~\bibnamefont {Bengio}},\ }\href {https://arxiv.org/abs/1301.3545} {\bibinfo {title} {Metric-free natural gradient for joint-training of {B}oltzmann machines}} (\bibinfo {year} {2013}),\ \Eprint {https://arxiv.org/abs/1301.3545} {arXiv:1301.3545 [cs.LG]} \BibitemShut {NoStop}%
\bibitem [{\citenamefont {Meyer}\ \emph {et~al.}(2025)\citenamefont {Meyer}, \citenamefont {Khatri}, \citenamefont {Stilck Fran\ifmmode~\mbox{\c{c}}\else \c{c}\fi{}a}, \citenamefont {Eisert},\ and\ \citenamefont {Faist}}]{meyer2025}%
  \BibitemOpen
  \bibfield  {author} {\bibinfo {author} {\bibfnamefont {J.~J.}\ \bibnamefont {Meyer}}, \bibinfo {author} {\bibfnamefont {S.}~\bibnamefont {Khatri}}, \bibinfo {author} {\bibfnamefont {D.}~\bibnamefont {Stilck Fran\ifmmode~\mbox{\c{c}}\else \c{c}\fi{}a}}, \bibinfo {author} {\bibfnamefont {J.}~\bibnamefont {Eisert}},\ and\ \bibinfo {author} {\bibfnamefont {P.}~\bibnamefont {Faist}},\ }\bibfield  {title} {\bibinfo {title} {Quantum metrology in the finite-sample regime},\ }\href {https://doi.org/10.1103/qbn1-p6bq} {\bibfield  {journal} {\bibinfo  {journal} {PRX Quantum}\ }\textbf {\bibinfo {volume} {6}},\ \bibinfo {pages} {030336} (\bibinfo {year} {2025})}\BibitemShut {NoStop}%
\bibitem [{\citenamefont {NISO}()}]{NISO}%
  \BibitemOpen
  \bibfield  {author} {\bibinfo {author} {\bibnamefont {NISO}},\ }\href@noop {} {\bibinfo {title} {Credit – contributor roles taxonomy}},\ \bibinfo {note} {\url{https://credit.niso.org/}, Accessed 2024-10-28}\BibitemShut {NoStop}%
\end{thebibliography}%

\onecolumngrid

\large

\appendix

\section{Additivity of Fisher--Bures information matrix for tensor-product parameterized families}

\label{app:add-FB-info-mat}

In this appendix, we prove that the Fisher--Bures information matrix is
additive with respect to tensor-products. This additivity relation is
essential to establishing~\eqref{eq:Cramer--Rao-multiple}, but it is frequently omitted in arguments used to establish~\eqref{eq:Cramer--Rao-multiple} (see, e.g., the end of \cite[Appendix~H]{Liu2019}, where it is simply written ``Consider the repetition of experiments (denoted as $n$), above bound needs to add a factor of $1/n$.'', thus omitting the essential additivity relation). Specifically, we prove the following:

\begin{theorem}
Let $\left(  \sigma(\theta)\right)  _{\theta\in\mathbb{R}^{J}}$ and $\left(
\tau(\theta)\right)  _{\theta\in\mathbb{R}^{J}}$ be parameterized families of
positive-definite states. Then
\begin{equation}
I^{\operatorname{FB}}(\theta;\left(  \sigma(\theta)\otimes\tau(\theta)\right)
_{\theta\in\mathbb{R}^{J}})=I^{\operatorname{FB}}(\theta;\left(  \sigma
(\theta)\right)  _{\theta\in\mathbb{R}^{J}})+I^{\operatorname{FB}}
(\theta;\left(  \tau(\theta)\right)  _{\theta\in\mathbb{R}^{J}}
).\label{eq:additivity-QFI}
\end{equation}

\end{theorem}

\begin{proof}
Let spectral decompositions of $\sigma(\theta)$ and $\tau(\theta)$ be as
follows:
\begin{align}
\sigma(\theta)  &  =\sum_{k}\lambda_{k}(\theta)|\psi_{k}(\theta)\rangle
\langle\psi_{k}(\theta)|,\\
\tau(\theta)  &  =\sum_{m}\mu_{m}(\theta)|\varphi_{m}(\theta)\rangle
\langle\varphi_{m}(\theta)|.
\end{align}
This implies that
\begin{equation}
\sigma(\theta)\otimes\tau(\theta)=\sum_{k,m}\lambda_{k}(\theta)\mu_{m}
(\theta)|\psi_{k}(\theta)\rangle\!\langle\psi_{k}(\theta)|\otimes|\varphi
_{m}(\theta)\rangle\!\langle\varphi_{m}(\theta)|.
\end{equation}
Plugging into~\eqref{eq:qfim-explicit}, we find that
\begin{multline}
I_{ij}^{\operatorname{FB}}(\theta;\left(  \sigma(\theta)\otimes\tau
(\theta)\right)  _{\theta\in\mathbb{R}^{J}})\label{eq:QFI-tensor-product}\\
=2\sum_{k,\ell,m,n}\frac{\langle\psi_{k}(\theta)|\langle\varphi_{m}
(\theta)|\partial_{i}\left[  \sigma(\theta)\otimes\tau(\theta)\right]
|\psi_{\ell}(\theta)\rangle|\varphi_{n}(\theta)\rangle\!\langle\psi_{\ell
}(\theta)|\langle\varphi_{n}(\theta)|\partial_{j}\left[  \sigma(\theta
)\otimes\tau(\theta)\right]  |\psi_{k}(\theta)\rangle|\varphi_{m}
(\theta)\rangle}{\lambda_{k}(\theta)\mu_{m}(\theta)+\lambda_{\ell}(\theta
)\mu_{n}(\theta)}.
\end{multline}
Consider that
\begin{equation}
\partial_{i}\left(  \sigma(\theta)\otimes\tau(\theta)\right)  =\left[
\partial_{i}\sigma(\theta)\right]  \otimes\tau(\theta)+\sigma(\theta
)\otimes\left[  \partial_{i}\tau(\theta)\right]  .
\end{equation}
This implies that
\begin{align}
&  \langle\psi_{k}(\theta)|\langle\varphi_{m}(\theta)|\partial_{i}\left[
\sigma(\theta)\otimes\tau(\theta)\right]  |\psi_{\ell}(\theta)\rangle
|\varphi_{n}(\theta)\rangle\nonumber\\
&  =\langle\psi_{k}(\theta)|\langle\varphi_{m}(\theta)|\left[  \left[
\partial_{i}\sigma(\theta)\right]  \otimes\tau(\theta)+\sigma(\theta
)\otimes\left[  \partial_{i}\tau(\theta)\right]  \right]  |\psi_{\ell}
(\theta)\rangle|\varphi_{n}(\theta)\rangle\\
&  =\langle\psi_{k}(\theta)|\langle\varphi_{m}(\theta)|\left[  \left[
\partial_{i}\sigma(\theta)\right]  \otimes\tau(\theta)\right]  |\psi_{\ell
}(\theta)\rangle|\varphi_{n}(\theta)\rangle\nonumber\\
&  \qquad+\langle\psi_{k}(\theta)|\langle\varphi_{m}(\theta)|\left[
\sigma(\theta)\otimes\left[  \partial_{i}\tau(\theta)\right]  \right]
|\psi_{\ell}(\theta)\rangle|\varphi_{n}(\theta)\rangle\\
&  =\langle\psi_{k}(\theta)|\left[  \partial_{i}\sigma(\theta)\right]
|\psi_{\ell}(\theta)\rangle\!\langle\varphi_{m}(\theta)|\tau(\theta)|\varphi
_{n}(\theta)\rangle\nonumber\\
&  \qquad+\langle\psi_{k}(\theta)|\sigma(\theta)|\psi_{\ell}(\theta
)\rangle\!\langle\varphi_{m}(\theta)|\left[  \partial_{i}\tau(\theta)\right]
|\varphi_{n}(\theta)\rangle\\
&  =\langle\psi_{k}(\theta)|\left[  \partial_{i}\sigma(\theta)\right]
|\psi_{\ell}(\theta)\rangle\mu_{m}(\theta)\delta_{mn}+\lambda_{k}
(\theta)\delta_{k\ell}\langle\varphi_{m}(\theta)|\left[  \partial_{i}
\tau(\theta)\right]  |\varphi_{n}(\theta)\rangle.
\end{align}
Plugging into~\eqref{eq:QFI-tensor-product}, we find that
\begin{align}
&  I_{ij}^{\operatorname{FB}}(\theta;\left(  \sigma(\theta)\otimes\tau
(\theta)\right)  _{\theta\in\mathbb{R}^{J}})\\
&  =2\sum_{k,\ell,m,n}\frac{\left(  \langle\psi_{k}(\theta)|\left[
\partial_{i}\sigma(\theta)\right]  |\psi_{\ell}(\theta)\rangle\mu_{m}
(\theta)\delta_{mn}+\lambda_{k}(\theta)\delta_{k\ell}\langle\varphi_{m}
(\theta)|\left[  \partial_{i}\tau(\theta)\right]  |\varphi_{n}(\theta
)\rangle\right)  }{\lambda_{k}(\theta)\mu_{m}(\theta)+\lambda_{\ell}
(\theta)\mu_{n}(\theta)}\times\nonumber\\
&  \qquad\left(  \langle\psi_{\ell}(\theta)|\left[  \partial_{j}\sigma
(\theta)\right]  |\psi_{k}(\theta)\rangle\mu_{m}(\theta)\delta_{mn}
+\lambda_{k}(\theta)\delta_{k\ell}\langle\varphi_{n}(\theta)|\left[
\partial_{j}\tau(\theta)\right]  |\varphi_{m}(\theta)\rangle\right) \\
&  =2\sum_{k,\ell,m,n}\frac{\langle\psi_{k}(\theta)|\left[  \partial_{i}
\sigma(\theta)\right]  |\psi_{\ell}(\theta)\rangle\mu_{m}(\theta)\delta
_{mn}\langle\psi_{\ell}(\theta)|\left[  \partial_{j}\sigma(\theta)\right]
|\psi_{k}(\theta)\rangle\mu_{m}(\theta)\delta_{mn}}{\lambda_{k}(\theta)\mu
_{m}(\theta)+\lambda_{\ell}(\theta)\mu_{n}(\theta)}\nonumber\\
&  \qquad+2\sum_{k,\ell,m,n}\frac{\langle\psi_{k}(\theta)|\left[  \partial
_{i}\sigma(\theta)\right]  |\psi_{\ell}(\theta)\rangle\mu_{m}(\theta
)\delta_{mn}\lambda_{k}(\theta)\delta_{k\ell}\langle\varphi_{n}(\theta
)|\left[  \partial_{j}\tau(\theta)\right]  |\varphi_{m}(\theta)\rangle
}{\lambda_{k}(\theta)\mu_{m}(\theta)+\lambda_{\ell}(\theta)\mu_{n}(\theta
)}\nonumber\\
&  \qquad+2\sum_{k,\ell,m,n}\frac{\lambda_{k}(\theta)\delta_{k\ell}
\langle\varphi_{m}(\theta)|\left[  \partial_{i}\tau(\theta)\right]
|\varphi_{n}(\theta)\rangle\!\langle\psi_{\ell}(\theta)|\left[  \partial
_{j}\sigma(\theta)\right]  |\psi_{k}(\theta)\rangle\mu_{m}(\theta)\delta_{mn}
}{\lambda_{k}(\theta)\mu_{m}(\theta)+\lambda_{\ell}(\theta)\mu_{n}(\theta
)}\nonumber\\
&  \qquad+2\sum_{k,\ell,m,n}\frac{\lambda_{k}(\theta)\delta_{k\ell}
\langle\varphi_{m}(\theta)|\left[  \partial_{i}\tau(\theta)\right]
|\varphi_{n}(\theta)\rangle\lambda_{k}(\theta)\delta_{k\ell}\langle\varphi
_{n}(\theta)|\left[  \partial_{j}\tau(\theta)\right]  |\varphi_{m}
(\theta)\rangle}{\lambda_{k}(\theta)\mu_{m}(\theta)+\lambda_{\ell}(\theta
)\mu_{n}(\theta)}\\
&  =2\sum_{k,\ell,m}\frac{\langle\psi_{k}(\theta)|\left[  \partial_{i}
\sigma(\theta)\right]  |\psi_{\ell}(\theta)\rangle\left[  \mu_{m}
(\theta)\right]  ^{2}\langle\psi_{\ell}(\theta)|\left[  \partial_{j}
\sigma(\theta)\right]  |\psi_{k}(\theta)\rangle}{\left[  \lambda_{k}
(\theta)+\lambda_{\ell}(\theta)\right]  \mu_{m}(\theta)}\nonumber\\
&  \qquad+2\sum_{k,m}\frac{\langle\psi_{k}(\theta)|\left[  \partial_{i}
\sigma(\theta)\right]  |\psi_{k}(\theta)\rangle\mu_{m}(\theta)\lambda
_{k}(\theta)\langle\varphi_{m}(\theta)|\left[  \partial_{j}\tau(\theta
)\right]  |\varphi_{m}(\theta)\rangle}{2\lambda_{k}(\theta)\mu_{m}(\theta
)}\nonumber\\
&  \qquad+2\sum_{k,m}\frac{\lambda_{k}(\theta)\langle\varphi_{m}
(\theta)|\left[  \partial_{i}\tau(\theta)\right]  |\varphi_{m}(\theta
)\rangle\!\langle\psi_{k}(\theta)|\left[  \partial_{j}\sigma(\theta)\right]
|\psi_{k}(\theta)\rangle\mu_{m}(\theta)}{2\lambda_{k}(\theta)\mu_{m}(\theta
)}\nonumber\\
&  \qquad+2\sum_{k,m,n}\frac{\left[  \lambda_{k}(\theta)\right]  ^{2}
\langle\varphi_{m}(\theta)|\left[  \partial_{i}\tau(\theta)\right]
|\varphi_{n}(\theta)\rangle\!\langle\varphi_{n}(\theta)|\left[  \partial_{j}
\tau(\theta)\right]  |\varphi_{m}(\theta)\rangle}{\lambda_{k}(\theta)\left[
\mu_{m}(\theta)+\mu_{n}(\theta)\right]  }\\
&  = 2\sum_{k,\ell,m}\frac{\langle\psi_{k}(\theta)|\left[  \partial_{i}
\sigma(\theta)\right]  |\psi_{\ell}(\theta)\rangle\!\langle\psi_{\ell}
(\theta)|\left[  \partial_{j}\sigma(\theta)\right]  |\psi_{k}(\theta)\rangle
}{\lambda_{k}(\theta)+\lambda_{\ell}(\theta)}\mu_{m}(\theta)\nonumber\\
&  \qquad+\sum_{k,m}\langle\psi_{k}(\theta)|\left[  \partial_{i}\sigma
(\theta)\right]  |\psi_{k}(\theta)\rangle\!\langle\varphi_{m}(\theta)|\left[
\partial_{j}\tau(\theta)\right]  |\varphi_{m}(\theta)\rangle\nonumber\\
&  \qquad+\sum_{k,m}\langle\varphi_{m}(\theta)|\left[  \partial_{i}\tau
(\theta)\right]  |\varphi_{m}(\theta)\rangle\!\langle\psi_{k}(\theta)|\left[
\partial_{j}\sigma(\theta)\right]  |\psi_{k}(\theta)\rangle\nonumber\\
&  \qquad+2\sum_{k,m,n}\lambda_{k}(\theta)\frac{\langle\varphi_{m}
(\theta)|\left[  \partial_{i}\tau(\theta)\right]  |\varphi_{n}(\theta
)\rangle\!\langle\varphi_{n}(\theta)|\left[  \partial_{j}\tau(\theta)\right]
|\varphi_{m}(\theta)\rangle}{\mu_{m}(\theta)+\mu_{n}(\theta)}\\
&  =2\sum_{k,\ell}\frac{\langle\psi_{k}(\theta)|\left[  \partial_{i}
\sigma(\theta)\right]  |\psi_{\ell}(\theta)\rangle\!\langle\psi_{\ell}
(\theta)|\left[  \partial_{j}\sigma(\theta)\right]  |\psi_{k}(\theta)\rangle
}{\lambda_{k}(\theta)+\lambda_{\ell}(\theta)}\nonumber\\
&  \qquad+\operatorname{Tr}[\partial_{i}\sigma(\theta)]\operatorname{Tr}
[\partial_{j}\tau(\theta)]+\operatorname{Tr}[\partial_{i}\tau(\theta
)]\operatorname{Tr}[\partial_{j}\sigma(\theta)]\nonumber\\
&  \qquad+2\sum_{m,n}\frac{\langle\varphi_{m}(\theta)|\left[  \partial_{i}
\tau(\theta)\right]  |\varphi_{n}(\theta)\rangle\!\langle\varphi_{n}
(\theta)|\left[  \partial_{j}\tau(\theta)\right]  |\varphi_{m}(\theta)\rangle
}{\mu_{m}(\theta)+\mu_{n}(\theta)}\\
&  =I_{ij}^{\operatorname{FB}}(\theta;\left(  \sigma(\theta)\right)
_{\theta\in\mathbb{R}^{J}})+I_{ij}^{\operatorname{FB}}(\theta;\left(
\tau(\theta)\right)  _{\theta\in\mathbb{R}^{J}}).
\end{align}
This concludes the proof.
\end{proof}

\bigskip 

Finally, to establish the additivity relation
\begin{equation}
I^{\operatorname{FB}}(\theta;\left(  \sigma^{\otimes n}(\theta)\right)
_{\theta\in\mathbb{R}^{J}}) = n\cdot I^{\operatorname{FB}}(\theta;\left(  \sigma(\theta)\right)
_{\theta\in\mathbb{R}^{J}}),
\end{equation}
which we stress again is needed to establish~\eqref{eq:Cramer--Rao-multiple}, we repeatedly apply
\eqref{eq:additivity-QFI} as follows:
\begin{align}
I^{\operatorname{FB}}(\theta;\left(  \sigma^{\otimes n}(\theta)\right)
_{\theta\in\mathbb{R}^{J}})  & =I^{\operatorname{FB}}(\theta;\left(
\sigma(\theta)\otimes\sigma^{\otimes n-1}(\theta)\right)  _{\theta
\in\mathbb{R}^{J}})\\
& =I^{\operatorname{FB}}(\theta;\left(  \sigma(\theta)\right)  _{\theta
\in\mathbb{R}^{J}})+I^{\operatorname{FB}}(\theta;\left(  \sigma^{\otimes
n-1}(\theta)\right)  _{\theta\in\mathbb{R}^{J}})\\
& =\cdots\\
& =n\cdot I^{\operatorname{FB}}(\theta;\left(  \sigma(\theta)\right)
_{\theta\in\mathbb{R}^{J}}).
\end{align}

\section{Proof of Theorem~\ref{thm:main}}

\label{app:FB-info-mat}

Plugging~\eqref{eq:derivative-thermal-state} into~\eqref{eq:qfim-explicit}, we find
that we should consider the following term:
\begin{align}
\langle k|\partial_{i}\rho(\theta)|\ell\rangle &  =\langle k|\left[  -\frac
{1}{2}\left\{  \Phi_{\theta}(G_{i}),\rho(\theta)\right\}  +\rho(\theta
)\left\langle G_{i}\right\rangle _{\rho(\theta)}\right]  |\ell\rangle\\
&  =-\frac{1}{2}\langle k|\left\{  \Phi_{\theta}(G_{i}),\rho(\theta)\right\}
|\ell\rangle+\langle k|\rho(\theta)|\ell\rangle\left\langle G_{i}\right\rangle
_{\rho(\theta)}\\
&  =-\frac{1}{2}\left(  \langle k|\Phi_{\theta}(G_{i})|\ell\rangle
\lambda_{\ell}+\langle k|\Phi_{\theta}(G_{i})|\ell\rangle\lambda_{k}\right)
+\delta_{k\ell}\lambda_{\ell}\left\langle G_{i}\right\rangle _{\rho(\theta)}\\
&  =-\frac{1}{2}\langle k|\Phi_{\theta}(G_{i})|\ell\rangle\left(  \lambda
_{k}+\lambda_{\ell}\right)  +\delta_{k\ell}\lambda_{\ell}\left\langle
G_{i}\right\rangle _{\rho(\theta)}.\label{eq:basic-calc-deriv}
\end{align}
This implies that
\begin{equation}
\langle\ell|\partial_{j}\rho(\theta)|k\rangle=-\frac{1}{2}\langle\ell
|\Phi_{\theta}(G_{j})|k\rangle\left(  \lambda_{k}+\lambda_{\ell}\right)
+\delta_{k\ell}\lambda_{\ell}\left\langle G_{j}\right\rangle _{\rho(\theta
)}\label{eq:basic-calc-deriv-HC}
\end{equation}
Plugging~\eqref{eq:basic-calc-deriv} and~\eqref{eq:basic-calc-deriv-HC}
into~\eqref{eq:qfim-explicit}, we find that
\begin{align}
I_{ij}^{\operatorname{FB}} &  =2\sum_{k,\ell}\frac{\langle k|\partial_{i}\rho(\theta
)|\ell\rangle\!\langle\ell|\partial_{j}\rho(\theta)|k\rangle}{\lambda
_{k}+\lambda_{\ell}}\nonumber\\
&  =2\sum_{k,\ell}\frac{\left[
\begin{array}
[c]{c}
\left(  -\frac{1}{2}\langle k|\Phi_{\theta}(G_{i})|\ell\rangle\left(
\lambda_{k}+\lambda_{\ell}\right)  +\delta_{k\ell}\lambda_{\ell}\left\langle
G_{i}\right\rangle _{\rho(\theta)}\right)  \times\\
\left(  -\frac{1}{2}\langle\ell|\Phi_{\theta}(G_{j})|k\rangle\left(
\lambda_{k}+\lambda_{\ell}\right)  +\delta_{k\ell}\lambda_{\ell}\left\langle
G_{j}\right\rangle _{\rho(\theta)}\right)
\end{array}
\right]  }{\lambda_{k}+\lambda_{\ell}}\\
&  =2\sum_{k,\ell}\frac{1}{4}\frac{\langle k|\Phi_{\theta}(G_{i})|\ell
\rangle\!\langle\ell|\Phi_{\theta}(G_{j})|k\rangle\left(  \lambda_{k}
+\lambda_{\ell}\right)  ^{2}}{\lambda_{k}+\lambda_{\ell}}\nonumber\\
&  \qquad+2\sum_{k,\ell}\left(  -\frac{1}{2}\right)  \frac{\langle
k|\Phi_{\theta}(G_{i})|\ell\rangle\left(  \lambda_{k}+\lambda_{\ell}\right)
\delta_{k\ell}\lambda_{\ell}\left\langle G_{j}\right\rangle _{\rho(\theta)}
}{\lambda_{k}+\lambda_{\ell}}\nonumber\\
&  \qquad+2\sum_{k,\ell}\left(  -\frac{1}{2}\right)  \frac{\langle\ell
|\Phi_{\theta}(G_{j})|k\rangle\left(  \lambda_{k}+\lambda_{\ell}\right)
\delta_{k\ell}\lambda_{\ell}\left\langle G_{i}\right\rangle _{\rho(\theta)}
}{\lambda_{k}+\lambda_{\ell}}\nonumber\\
&  \qquad+2\sum_{k,\ell}\frac{\delta_{k\ell}\lambda_{\ell}\left\langle
G_{i}\right\rangle _{\rho(\theta)}\delta_{k\ell}\lambda_{\ell}\left\langle
G_{j}\right\rangle _{\rho(\theta)}}{\lambda_{k}+\lambda_{\ell}}\\
&  =\frac{1}{2}\sum_{k,\ell}\langle k|\Phi_{\theta}(G_{i})|\ell\rangle
\langle\ell|\Phi_{\theta}(G_{j})|k\rangle\left(  \lambda_{k}+\lambda_{\ell
}\right)  \nonumber\\
&  \qquad-\sum_{k,\ell}\langle k|\Phi_{\theta}(G_{i})|\ell\rangle\delta
_{k\ell}\lambda_{\ell}\left\langle G_{j}\right\rangle _{\rho(\theta)}
-\sum_{k,\ell}\langle\ell|\Phi_{\theta}(G_{j})|k\rangle\delta_{k\ell}
\lambda_{\ell}\left\langle G_{i}\right\rangle _{\rho(\theta)}\nonumber\\
&  \qquad+2\sum_{k,\ell}\frac{\delta_{k\ell}\lambda_{\ell}\left\langle
G_{i}\right\rangle _{\rho(\theta)}\delta_{k\ell}\lambda_{\ell}\left\langle
G_{j}\right\rangle _{\rho(\theta)}}{\lambda_{k}+\lambda_{\ell}}\\
&  =\frac{1}{2}\sum_{k,\ell}\langle k|\Phi_{\theta}(G_{i})|\ell\rangle
\langle\ell|\Phi_{\theta}(G_{j})|k\rangle\lambda_{k}+\frac{1}{2}\sum_{k,\ell
}\langle k|\Phi_{\theta}(G_{i})|\ell\rangle\!\langle\ell|\Phi_{\theta}
(G_{j})|k\rangle\lambda_{\ell}\nonumber\\
&  \qquad-\sum_{k}\langle k|\Phi_{\theta}(G_{i})|k\rangle\lambda
_{k}\left\langle G_{j}\right\rangle _{\rho(\theta)}-\sum_{k}\langle
k|\Phi_{\theta}(G_{j})|k\rangle\lambda_{k}\left\langle G_{i}\right\rangle
_{\rho(\theta)}\nonumber\\
&  \qquad+2\sum_{k}\frac{\lambda_{k}^{2}\left\langle G_{i}\right\rangle
_{\rho(\theta)}\left\langle G_{j}\right\rangle _{\rho(\theta)}}{2\lambda_{k}
}\\
&  =\frac{1}{2}\operatorname{Tr}[\Phi_{\theta}(G_{j})\rho(\theta)\Phi_{\theta
}(G_{i})]+\frac{1}{2}\operatorname{Tr}[\Phi_{\theta}(G_{i})\rho(\theta
)\Phi_{\theta}(G_{j})]\nonumber\\
&  \qquad-\operatorname{Tr}[\rho(\theta)\Phi_{\theta}(G_{i})]\left\langle
G_{j}\right\rangle _{\rho(\theta)}-\operatorname{Tr}[\rho(\theta)\Phi_{\theta
}(G_{j})]\left\langle G_{i}\right\rangle _{\rho(\theta)}\nonumber\\
&  \qquad+\left\langle G_{i}\right\rangle _{\rho(\theta)}\left\langle
G_{j}\right\rangle _{\rho(\theta)}\\
&  =\frac{1}{2}\operatorname{Tr}[\Phi_{\theta}(G_{j})\rho(\theta)\Phi_{\theta
}(G_{i})]+\frac{1}{2}\operatorname{Tr}[\Phi_{\theta}(G_{i})\rho(\theta
)\Phi_{\theta}(G_{j})]\nonumber\\
&  \qquad-\left\langle G_{i}\right\rangle _{\rho(\theta)}\left\langle
G_{j}\right\rangle _{\rho(\theta)}\\
&  =\operatorname{Re}\left[  \operatorname{Tr}[\Phi_{\theta}(G_{j})\rho
(\theta)\Phi_{\theta}(G_{i})\right]  ]-\left\langle G_{i}\right\rangle
_{\rho(\theta)}\left\langle G_{j}\right\rangle _{\rho(\theta)}.
\end{align}

\section{Proof of Theorem~\ref{thm:main-Kubo-Mori}}

\label{app:KM-info-mat-proof}

Consider that
\begin{align}  I_{ij}^{\operatorname{KM}}& =\sum_{k,\ell}\frac{\ln \lambda_k - \ln \lambda_{\ell}}{\lambda_{k}-\lambda_{\ell}}  \langle k|\partial
_{i}\rho(\theta)|\ell\rangle\!\langle\ell|\partial_{j}\rho(\theta)|k\rangle
  \\
& = \sum_{k,\ell}\frac{\ln \lambda_k - \ln \lambda_{\ell} }{\lambda_{k}-\lambda_{\ell}}\left[
\begin{array}
[c]{c}
\left(  -\frac{1}{2}\langle k|\Phi_{\theta}(G_{i})|\ell\rangle\left(
\lambda_{k}+\lambda_{\ell}\right)  +\delta_{k\ell}\lambda_{\ell}\left\langle
G_{i}\right\rangle _{\rho(\theta)}\right)  \times\\
\left(  -\frac{1}{2}\langle\ell|\Phi_{\theta}(G_{j})|k\rangle\left(
\lambda_{k}+\lambda_{\ell}\right)  +\delta_{k\ell}\lambda_{\ell}\left\langle
G_{j}\right\rangle _{\rho(\theta)}\right)
\end{array}
\right] \\
&  =\sum_{k,\ell}\frac{1}{4}\frac{\ln \lambda_k - \ln \lambda_{\ell}}
{\lambda_{k}-\lambda_{\ell}}  \langle
k|\Phi_{\theta}(G_{i})|\ell\rangle\!\langle\ell|\Phi_{\theta}(G_{j}
)|k\rangle  \left(  \lambda_{k}+\lambda_{\ell}\right)  ^{2}\nonumber\\
&  \qquad+\sum_{k,\ell}-\frac{1}{2}\frac{\ln \lambda_k - \ln \lambda_{\ell}}{\lambda_{k}-\lambda_{\ell}}  \langle
k|\Phi_{\theta}(G_{i})|\ell\rangle\left(  \lambda_{k}+\lambda_{\ell}\right)
\delta_{k\ell}\lambda_{\ell}\left\langle G_{j}\right\rangle _{\rho(\theta
)}  \nonumber\\
&  \qquad+\sum_{k,\ell}-\frac{1}{2}\frac{\ln \lambda_k - \ln \lambda_{\ell}}{\lambda_{k}-\lambda_{\ell}}  \langle
\ell|\Phi_{\theta}(G_{j})|k\rangle\left(  \lambda_{k}+\lambda_{\ell}\right)
\delta_{k\ell}\lambda_{\ell}\left\langle G_{i}\right\rangle _{\rho(\theta
)}  \nonumber\\
&  \qquad+\sum_{k,\ell}\frac{\ln \lambda_k - \ln \lambda_{\ell}
}{\lambda_{k}-\lambda_{\ell}}  \delta_{k\ell}
\lambda_{\ell}\left\langle G_{i}\right\rangle _{\rho(\theta)}\delta_{k\ell
}\lambda_{\ell}\left\langle G_{j}\right\rangle _{\rho(\theta)}\\
&  =\sum_{k,\ell}\frac{1}{4}\frac{(\ln \lambda_k - \ln \lambda_{\ell})\left(  \lambda_{k}+\lambda_{\ell}\right)  ^{2}}
{\lambda_{k}-\lambda_{\ell}}  \langle
k|\Phi_{\theta}(G_{i})|\ell\rangle\!\langle\ell|\Phi_{\theta}(G_{j}
)|k\rangle  \nonumber\\
&  \qquad+\sum_{k}\left(-\frac{1}{2}\right) \frac{1}{\lambda_{k}}  \langle
k|\Phi_{\theta}(G_{i})|k\rangle\left(  2\lambda_{k}\right)
\lambda_{k}\left\langle G_{j}\right\rangle _{\rho(\theta
)}  +\sum_{k} \left(-\frac{1}{2}\right)\frac{1}{\lambda_{k}} \langle
k|\Phi_{\theta}(G_{j})|k\rangle\left(  2\lambda_{k}\right)
\lambda_{k}\left\langle G_{i}\right\rangle _{\rho(\theta
)}  \nonumber\\
&  \qquad+\sum_{k}\frac{1
}{\lambda_{k}} 
\lambda_{k}\left\langle G_{i}\right\rangle _{\rho(\theta)}\lambda_{k}\left\langle G_{j}\right\rangle _{\rho(\theta)}\\
&  =\sum_{k,\ell}\frac{1}{4}\frac{(\ln \lambda_k - \ln \lambda_{\ell})\left(  \lambda_{k}+\lambda_{\ell}\right)  ^{2}}
{\lambda_{k}-\lambda_{\ell}}  \langle
k|\Phi_{\theta}(G_{i})|\ell\rangle\!\langle\ell|\Phi_{\theta}(G_{j}
)|k\rangle  \nonumber\\
&  \qquad-\sum_{k}  \langle
k|\Phi_{\theta}(G_{i})|k\rangle
\lambda_{k}\left\langle G_{j}\right\rangle _{\rho(\theta
)}  -\sum_{k} \langle
k|\Phi_{\theta}(G_{j})|k\rangle
\lambda_{k}\left\langle G_{i}\right\rangle _{\rho(\theta
)}  +\sum_{k} 
\left\langle G_{i}\right\rangle _{\rho(\theta)}\lambda_{k}\left\langle G_{j}\right\rangle _{\rho(\theta)}\\
&  =\sum_{k,\ell}\frac{1}{4}\frac{(\ln \lambda_k - \ln \lambda_{\ell})\left(  \lambda_{k}+\lambda_{\ell}\right)  ^{2}}
{\lambda_{k}-\lambda_{\ell}}  \langle
k|\Phi_{\theta}(G_{i})|\ell\rangle\!\langle\ell|\Phi_{\theta}(G_{j}
)|k\rangle  \nonumber\\
&  \qquad-\operatorname{Tr}\left[\rho(\theta) \Phi_{\theta}(G_i) \right]\left\langle G_{j}\right\rangle_{\rho(\theta
)} -\operatorname{Tr}\left[\rho(\theta) \Phi_{\theta}(G_j) \right]\left\langle G_{i}\right\rangle _{\rho(\theta
)}+\left\langle G_{i}\right\rangle _{\rho(\theta
)}\left\langle G_{j}\right\rangle _{\rho(\theta
)}\\
& =\sum_{k,\ell}\frac{1}{4}\frac{(\ln \lambda_k - \ln \lambda_{\ell})\left(  \lambda_{k}+\lambda_{\ell}\right)  ^{2}}
{\lambda_{k}-\lambda_{\ell}}  \langle
k|\Phi_{\theta}(G_{i})|\ell\rangle\!\langle\ell|\Phi_{\theta}(G_{j}
)|k\rangle  -\left\langle G_{i}\right\rangle _{\rho(\theta
)}\left\langle G_{j}\right\rangle _{\rho(\theta
)}.\label{eq:KM-expansion}
\end{align}
The fourth equality is a consequence of the following fact:
\begin{equation}
    \lim_{x\rightarrow y} \frac{\ln x - \ln y}{ x-y} = \frac{1}{y}.
\end{equation}
Now consider that $G(\theta) = \sum_k \mu_k |k\rangle\!\langle k|$. This implies that for all $k$, we have $\lambda_k =\frac{e^{-\mu_k}}{Z}$, where $Z$ is the partition function. Plugging this into~\eqref{eq:KM-expansion}, we obtain:
\begin{align}
    & \sum_{k,\ell}\frac{1}{4}\frac{(\ln \lambda_k - \ln \lambda_{\ell})\left(  \lambda_{k}+\lambda_{\ell}\right)  ^{2}}
{\lambda_{k}-\lambda_{\ell}}  \langle
k|\Phi_{\theta}(G_{i})|\ell\rangle\!\langle\ell|\Phi_{\theta}(G_{j}
)|k\rangle\nonumber\\
    & =\sum_{k,\ell}\frac{1}{4}\frac{\left(\ln \frac{e^{-\mu_k}}{Z} - \ln \frac{e^{-\mu_\ell}}{Z}\right)\left(  \frac{e^{-\mu_k}}{Z}+\frac{e^{-\mu_\ell}}{Z}\right)  ^{2}}
{\frac{e^{-\mu_k}}{Z}-\frac{e^{-\mu_\ell}}{Z}}\langle
k|\Phi_{\theta}(G_{i})|\ell\rangle\!\langle\ell|\Phi_{\theta}(G_{j}
)|k\rangle  \\
& =\sum_{k,\ell}\frac{1}{4Z}\frac{\left(-\mu_k + \mu_\ell\right)\left(  e^{-\mu_k}+e^{-\mu_\ell}\right)  ^{2}}
{e^{-\mu_k}-e^{-\mu_\ell}}\langle
k|\Phi_{\theta}(G_{i})|\ell\rangle\!\langle\ell|\Phi_{\theta}(G_{j}
)|k\rangle \\
& =\sum_{k,\ell}\frac{1}{4Z}\frac{\left(-\mu_k + \mu_\ell\right)\left(  e^{-\mu_k}+e^{-\mu_\ell}\right)  }
{\frac{e^{-\mu_k}-e^{-\mu_\ell}}{e^{-\mu_k}+e^{-\mu_\ell}}}  \langle
k|\Phi_{\theta}(G_{i})|\ell\rangle\!\langle\ell|\Phi_{\theta}(G_{j}
)|k\rangle \\
& =\sum_{k,\ell}\frac{1}{4Z}\frac{\left(-\mu_k + \mu_\ell\right)\left(  e^{-\mu_k}+e^{-\mu_\ell}\right)  }
{\frac{e^{-\mu_k + \mu_\ell} - 1}{e^{-\mu_k + \mu_\ell}+1}} \langle
k|\Phi_{\theta}(G_{i})|\ell\rangle\!\langle\ell|\Phi_{\theta}(G_{j}
)|k\rangle \\
& =\sum_{k,\ell}\frac{1}{4Z}\frac{\left(-\mu_k + \mu_\ell\right)\left(  e^{-\mu_k}+e^{-\mu_\ell}\right)  }
{\tanh\!\left(\frac{-\mu_k + \mu_\ell}{2}\right)} \langle
k|\Phi_{\theta}(G_{i})|\ell\rangle\!\langle\ell|\Phi_{\theta}(G_{j}
)|k\rangle \\
& =\sum_{k,\ell}\frac{1}{4Z}\frac{\left(-\mu_k + \mu_\ell\right)\left(  e^{-\mu_k}+e^{-\mu_\ell}\right)  }
{\tanh\!\left(\frac{-\mu_k + \mu_\ell}{2}\right)} \langle
k| \int_{\mathbb{R}} dt\ p(t)\ e^{iG(\theta)t} G_i e^{-iG(\theta)t}|\ell\rangle\!\langle\ell|\Phi_{\theta}(G_{j}
)|k\rangle \\
& =\sum_{k,\ell}\frac{1}{4Z}\frac{\left(-\mu_k + \mu_\ell\right)\left(  e^{-\mu_k}+e^{-\mu_\ell}\right)  }
{\tanh\!\left(\frac{-\mu_k + \mu_\ell}{2}\right)}\times \nonumber \\
& \qquad   \langle
k| \int_{\mathbb{R}} dt\ p(t)\ \left(\sum_m |m\rangle\!\langle m| e^{i\mu_m t} \right)G_i \left(\sum_n |n\rangle\!\langle n| e^{-i\mu_n t} \right)|\ell\rangle\!\langle\ell|\Phi_{\theta}(G_{j}
)|k\rangle \\
& =\sum_{k,\ell}\frac{1}{4Z}\frac{\left(-\mu_k + \mu_\ell\right)\left(  e^{-\mu_k}+e^{-\mu_\ell}\right)  }
{\tanh\!\left(\frac{-\mu_k + \mu_\ell}{2}\right)} \int_{\mathbb{R}} dt\ p(t)\ e^{i\mu_k t} \langle k | G_i|\ell\rangle  e^{-i\mu_\ell t} \langle\ell|\Phi_{\theta}(G_{j}
)|k\rangle \\
& =\sum_{k,\ell}\frac{1}{4Z}\frac{\left(-\mu_k + \mu_\ell\right)\left(  e^{-\mu_k}+e^{-\mu_\ell}\right)  }
{\tanh\!\left(\frac{-\mu_k + \mu_\ell}{2}\right)} \int_{\mathbb{R}} dt\ p(t)\ e^{-i(\mu_\ell- \mu_k) t} \langle k | G_i|\ell\rangle   \langle\ell|\Phi_{\theta}(G_{j}
)|k\rangle \\
& =\sum_{k,\ell}\frac{1}{4Z}\frac{\left(-\mu_k + \mu_\ell\right)\left(  e^{-\mu_k}+e^{-\mu_\ell}\right)  }
{\tanh\!\left(\frac{-\mu_k + \mu_\ell}{2}\right)} \frac{\tanh\!\left(\frac{-\mu_k + \mu_\ell}{2}\right)}{\frac{\left(-\mu_k + \mu_\ell\right)}{2}}\langle k | G_i|\ell\rangle   \langle\ell|\Phi_{\theta}(G_{j}
)|k\rangle \\
& =\sum_{k,\ell}\frac{1}{2Z}\left(  e^{-\mu_k}+e^{-\mu_\ell}\right)  
\langle k | G_i|\ell\rangle   \langle\ell|\Phi_{\theta}(G_{j}
)|k\rangle \\
& =\sum_{k,\ell}\frac{1}{2}\left(  \frac{e^{-\mu_k}}{Z}+\frac{e^{-\mu_\ell}}{Z}\right)  
\langle k | G_i|\ell\rangle   \langle\ell|\Phi_{\theta}(G_{j}
)|k\rangle \\
& =\sum_{k,\ell}\frac{1}{2}\left(  \lambda_k +\lambda_\ell\right)  
\langle k | G_i|\ell\rangle   \langle\ell|\Phi_{\theta}(G_{j}
)|k\rangle \\
& =\frac{1}{2}\operatorname{Tr}\left[\rho(\theta)  G_i \Phi_{\theta}(G_{j}
)  \right]+\frac{1}{2}\operatorname{Tr}\left[ G_i \rho(\theta)\Phi_{\theta}(G_{j}
)  \right]\\
& =\frac{1}{2}\operatorname{Tr}\left[ \{G_i,\Phi_{\theta}(G_{j}
)\} \rho(\theta) \right] .
\label{eq:last-line-KM}
\end{align}
When combining~\eqref{eq:last-line-KM} with~\eqref{eq:KM-expansion}, the proof is concluded.

\section{Review of Hadamard test}

\label{app:quant_primitive}

In this appendix, we recall a fundamental primitive that estimates the following quantity:
\begin{equation}
    \frac{1}{2}\operatorname{Tr}\!\left[\left(U_1^\dag U_0  +  U_0^{\dagger}U_1\right)\rho\right]\label{eq:U0_U1},
\end{equation}
where $U_0$ and $U_1$ are unitaries and $\rho$ is a quantum state. 
This primitive was reviewed  in~\cite[Section~6 of Appendix B]{patel2024quantumboltzmann}, and we review it here for convenience.
For this primitive, we assume that we have access to multiple copies of $\rho$. The quantum circuit for this primitive is depicted in Figure~\ref{fig:qc-primitive}, where the controlled gate is given by $|0\rangle\!\langle 0 | \otimes U_0 + |1\rangle\!\langle 1 | \otimes U_1$. This circuit consists of the following two quantum registers:
\begin{itemize}
    \item a control register, initialized in the state $|0\rangle \!\langle 0|$,
    \item a system register, which is in the state $\rho$.
\end{itemize}
After executing this circuit and obtaining a measurement outcome~$b$ in the control register, the final state $\sigma_{\operatorname{sub}}^{(b)}$ (sub-normalized) of the system register is as follows, where $b\in\{0, 1\}$:
\begin{align}
    \sigma_{\operatorname{sub}}^{(b)} & = \left(\frac{U_0 + (-1)^b U_1}{2}\right) \rho \left(\frac{U_0^\dag + (-1)^b U_1^{\dagger}}{2}\right)\\
    & = \frac{1}{4}\left( U_0\rho U_0^\dag +(-1)^b U_0\rho U_1^\dag + (-1)^b U_1\rho U_0^{\dagger} + U_1\rho U_1^{\dagger} \right).
\end{align}
The probability $p_{b}$ of obtaining the measurement outcome $b$ is given as:
\begin{align}
    p_{b} = \operatorname{Tr}\!\left[\sigma_{\operatorname{sub}}^{(b)}\right] & = \frac{2 + (-1)^b\operatorname{Tr}\!\left[U_0\rho U_1^\dag \right] + (-1)^b\operatorname{Tr}\!\left[U_1\rho U_0^{\dagger} \right]}{4}\\
    & = \frac{2 + (-1)^b\operatorname{Tr}\!\left[\left(U_1^\dag U_0  +  U_0^{\dagger}U_1\right) \rho \right]}{4}\label{eq:fun_pri_prob}.
\end{align}

\begin{figure}
    \centering
    \includegraphics[width=0.5\linewidth]{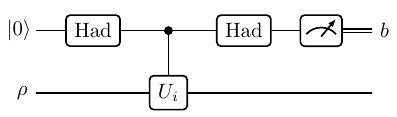}
    \caption{Quantum primitive for estimating  $\frac{1}{2}\operatorname{Tr}\!\left[\left(U_1^\dag U_0  +  U_0^{\dagger}U_1\right)\rho\right]$.}
    \label{fig:qc-primitive}
\end{figure}

Now to estimate the quantity, given by~\eqref{eq:U0_U1}, we employ the following approach. Let $b_1, \ldots, b_N$ represent the measurement results obtained from $N$ independent executions of the quantum circuit mentioned above. We define a random variable $X$ as $X \coloneqq (-1)^b$. The sample mean
\begin{equation}
    \overline{X} \coloneqq \frac{1}{N}\sum_{i=1}^{N} X_{i} 
\end{equation}
serves as an unbiased estimator for the quantity of interest because
\begin{align}
    \mathbb{E}\!\left[\,\overline{X}\, \right] & = \mathbb{E}\!\left[\frac{1}{N}\sum_{i=1}^{N} X_{i}\right] = \frac{1}{N}\sum_{i=1}^{N} \mathbb{E}\!\left[X_{i}\right]= \frac{1}{N}\sum_{i=1}^{N} \sum_{b_{i} \in \{0, 1\}}p_{b_{i}} (-1)^{b_{i}} \\
    & = \frac{1}{N}\sum_{i=1}^{N}\left ( \frac{2 + \operatorname{Tr}\!\left[\left(U_1^\dag U_0  +  U_0^{\dagger}U_1\right)\rho \right]}{4} - \frac{2 - \operatorname{Tr}\!\left[\left(U_1^\dag U_0  +  U_0^{\dagger}U_1\right)\rho \right]}{4}\right)\\
    & = \frac{1}{2} \operatorname{Tr}\!\left[\left(U_1^\dag U_0  +  U_0^{\dagger}U_1\right)\rho \right] \label{eq:prim-ub-final}.
\end{align}

\section{Estimating the first term of the Fisher--Bures information matrix elements}

\label{app:algo-first-term-FB}

Let us first recall, from the statement of  Theorem~\ref{thm:main}, the expression for the $(i, j)$-th element of the Fisher--Bures information matrix $I^{\operatorname{FB}}(\theta)$:
\begin{equation}
I_{ij}^{\operatorname{FB}}(\theta)=\frac{1}{2}\operatorname{Tr}[\left\{
\Phi_{\theta}(G_{i}),\Phi_{\theta}(G_{j})\right\}  \rho(\theta)]-\left\langle
G_{i}\right\rangle _{\rho(\theta)}\left\langle G_{j}\right\rangle
_{\rho(\theta)}.\label{eq:Fisher-matrix-entries-app}
\end{equation}
As previously stated in the main text (paragraph surrounding~\eqref{eq:simple-term-to-estimate}), estimating the second term of the right-hand side of the above equation is  straightforward. So, in what follows, we present an algorithm for estimating the first term of the above equation in greater detail.

Consider the following:
\begin{align}
    & \frac{1}{2}\operatorname{Tr}[\left\{
\Phi_{\theta}(G_{i}),\Phi_{\theta}(G_{j})\right\}  \rho(\theta)]\nonumber\\
& = \frac{1}{2}\left(\operatorname{Tr}[
\Phi_{\theta}(G_{i})\Phi_{\theta}(G_{j})  \rho(\theta)] + \operatorname{Tr}[
\Phi_{\theta}(G_{j})\Phi_{\theta}(G_{i})  \rho(\theta)] \right)\\
& = \frac{1}{2}\Bigg(\operatorname{Tr}\left[
\int_{\mathbb{R}}\int_{\mathbb{R}} dt_1\ dt_2\ p(t_1) p(t_2) e^{iG(\theta)t_1} G_i e^{-iG(\theta)t_1} e^{iG(\theta)t_2} G_j e^{-iG(\theta)t_2}\rho(\theta)\right]\notag \\
& \qquad+ \operatorname{Tr}\left[
\int_{\mathbb{R}}\int_{\mathbb{R}} dt_1\ dt_2\ p(t_1) p(t_2) e^{iG(\theta)t_2} G_j e^{-iG(\theta)t_2} e^{iG(\theta)t_1} G_i e^{-iG(\theta)t_1}\rho(\theta)\right]\Bigg)\\
& = \frac{1}{2}\int_{\mathbb{R}}\int_{\mathbb{R}} dt_1\ dt_2\ p(t_1) p(t_2) \Bigg(\operatorname{Tr}\left[
e^{iG(\theta)t_1} G_i e^{-iG(\theta)t_1} e^{iG(\theta)t_2} G_j e^{-iG(\theta)t_2}\rho(\theta)\right]\notag \\
& \qquad+ \operatorname{Tr}\left[
e^{iG(\theta)t_2} G_j e^{-iG(\theta)t_2} e^{iG(\theta)t_1} G_i e^{-iG(\theta)t_1}\rho(\theta)\right]\Bigg)\\
& = \frac{1}{2}\int_{\mathbb{R}}\int_{\mathbb{R}} dt_1\ dt_2\ p(t_1) p(t_2) \Bigg(\operatorname{Tr}\left[
G_i e^{-iG(\theta)t_1} e^{iG(\theta)t_2} G_j e^{-iG(\theta)t_2}\rho(\theta)e^{iG(\theta)t_1} \right]\notag \\
& \qquad+ \operatorname{Tr}\left[e^{iG(\theta)t_2}
G_j e^{-iG(\theta)t_2} e^{iG(\theta)t_1} G_i \rho(\theta) e^{-iG(\theta)t_1}\right]\Bigg)\\
& = \frac{1}{2}\int_{\mathbb{R}}\int_{\mathbb{R}} dt_1\ dt_2\ p(t_1) p(t_2) \Bigg(\operatorname{Tr}\left[
G_i e^{-iG(\theta)t_1} e^{iG(\theta)t_2} G_j e^{-iG(\theta)t_2}e^{iG(\theta)t_1} \rho(\theta)\right]\notag \\
& \qquad+ \operatorname{Tr}\left[e^{-iG(\theta)t_1}e^{iG(\theta)t_2}
G_j e^{-iG(\theta)t_2} e^{iG(\theta)t_1} G_i \rho(\theta)\right]\Bigg)\\
& = \frac{1}{2}\int_{\mathbb{R}}\int_{\mathbb{R}} dt_1\ dt_2\ p(t_1) p(t_2) \Bigg(\operatorname{Tr}\left[
G_i e^{-iG(\theta)(t_1-t_2)} G_j e^{iG(\theta)(t_1-t_2)}\rho(\theta)\right]\notag \\
& \qquad+ \operatorname{Tr}\left[e^{-iG(\theta)(t_1-t_2)}
G_j e^{iG(\theta)(t_1-t_2)} G_i \rho(\theta)\right]\Bigg)\\
& = \int_{\mathbb{R}}\int_{\mathbb{R}} dt_1\ dt_2\ p(t_1) p(t_2) \Bigg(\frac{1}{2}\operatorname{Tr}\Big[\Big(
G_i e^{-iG(\theta)(t_1-t_2)} G_j e^{iG(\theta)(t_1-t_2)}\notag \\
& \qquad+ e^{-iG(\theta)(t_1-t_2)}
G_j e^{iG(\theta)(t_1-t_2)} G_i\Big) \rho(\theta)\Big]\Bigg)\label{eq:first_fb_simplify}.
\end{align}

We are now in a position to present an algorithm (Algorithm~\ref{algo:first-term-FB}) to estimate the first term of~\eqref{eq:Fisher-matrix-entries-app} using its equivalent form mentioned above. At the core of our algorithm lies a quantum primitive (refer to Appendix~\ref{app:quant_primitive} for more details) that estimates a quantity of the form given by
\begin{equation}
    \frac{1}{2}\operatorname{Tr}\!\left[\left(U_1^\dag U_0  +  U_0^{\dagger}U_1\right)\rho\right].
\end{equation}
Note that by substituting $U_1  =  G_i e^{-iG(\theta)(t_1-t_2)} G_j, U_0  = e^{-iG(\theta)(t_1-t_2)}$, and $\rho = \rho(\theta)$ into the equation above, we recover the expression inside the integral of~\eqref{eq:first_fb_simplify}. To this end, the quantum circuit of the primitive with these values of $U_0$ and $U_1$ is depicted in Figure~\ref{fig:QBFB-estimator}.

\begin{algorithm}[H]
\caption{\texorpdfstring{$\mathtt{estimate\_first\_term\_FB}(i, j, \theta, \{G_\ell\}_{\ell=1}^{J}, \varepsilon, \delta)$}{estimate first term}}
\label{algo:first-term-FB}
\begin{algorithmic}[1]
\STATE \textbf{Input:} Indices $i, j \in [J]$, parameter vector $\theta = \left( \theta_{1}, \ldots,  \theta_{J}\right)^{\mathsf{T}} \in \mathbb{R}^{J}$, Gibbs local Hamiltonians $\{G_\ell\}_{\ell=1}^{J}$, precision $\varepsilon > 0$, error probability $\delta \in (0,1)$
\STATE $N \leftarrow \lceil\sfrac{2 \ln(\sfrac{2}{\delta})}{\varepsilon^2}\rceil$
\FOR{$n = 0$ to $N-1$}
\STATE Initialize the control register to $|0\rangle\!\langle 0 |$
\STATE Prepare the system register in the state $\rho(\theta)$
\STATE Sample $t_1$ and $t_2$ independently at random with probability $p(t)$ (defined in~\eqref{eq-mt:prob-dens})
\STATE Apply the Hadamard gate to the control register
\STATE Apply the following unitaries to the control and system registers:
\STATE \hspace{0.6cm} \textbullet~Controlled-$G_i$: $G_i$ is a local unitary with the control on the control register
\STATE \hspace{0.6cm} \textbullet~$e^{-iG(\theta)(t_1 - t_2)}$: Hamiltonian simulation for time $t_1 - t_2$ on the system register
\STATE \hspace{0.6cm} \textbullet~Controlled-$G_j$: $G_j$ is a local unitary with the control on the control register
\STATE Apply the Hadamard gate to the control register
\STATE Measure the control register in the computational basis and store the measurement outcome~$b_n$
\STATE $Y_{n}^{(\operatorname{FB})} \leftarrow (-1)^{b_n}$
\ENDFOR
\RETURN $\overline{Y}^{(\operatorname{FB})} \leftarrow \frac{1}{N}\sum_{n=0}^{N-1}Y_{n}^{(\operatorname{FB})}$
\end{algorithmic}
\end{algorithm}

We now show that the output of Algorithm~\ref{algo:first-term-FB}, i.e., $\overline{Y}^{(\operatorname{FB})}$, is an unbiased estimator of the first term of~\eqref{eq:Fisher-matrix-entries-app}:
\begin{align}
    \mathbb{E}\!\left [\overline{Y}^{(\operatorname{FB})}\right]&  = \mathbb{E}\!\left [\frac{1}{N}\sum_{n=0}^{N-1}Y_n^{(\operatorname{FB})}\right] = \mathbb{E}\!\left [\frac{1}{N}\sum_{n=0}^{N-1} (-1)^{b_n} \right]\\
    & = \frac{1}{N}\sum_{n=0}^{N-1}\mathbb{E}\!\left [  (-1)^{b_n} \right] = \frac{1}{N}\sum_{n=0}^{N-1}\sum_{b_{n} \in \{0, 1\}}p_{b_{n}}\left[ (-1)^{b_{n}} \right]\label{eq:exp-sim-FB}.
\end{align}
Observe that 
\begin{align}
    p_{b_n} \coloneqq  \int_{\mathbb{R}} \int_{\mathbb{R}} dt_1 \, dt_2\  p(t_1)  p(t_2)\left ( \frac{2 + (-1)^{b_n}\operatorname{Tr}\!\left[\left(U_1^\dag U_0  +  U_0^{\dagger}U_1\right) \rho(\theta) \right]}{4} \right).
\end{align}
This is because in Algorithm~\ref{algo:first-term-FB}, we first sample $t_1$ and $t_2$ independently at random with probability $p(t)$ and then apply the primitive introduced in Appendix~\ref{app:quant_primitive} with $U_1  =  G_i e^{-iG(\theta)(t_1-t_2)} G_j, U_0  = e^{-iG(\theta)(t_1-t_2)}$, whose probability of outputting a bit $b$ is given by~\eqref{eq:fun_pri_prob}. Substituting the above expression for $p_{b_n}$ into~\eqref{eq:exp-sim-FB}, we get
\begin{align}
    & \mathbb{E}\!\left [\overline{Y}^{(\operatorname{FB})}\right]\nonumber\\
    &  = \frac{1}{N}\sum_{n=0}^{N-1}\sum_{b_{n} \in \{0, 1\}}\left(\int_{\mathbb{R}} \int_{\mathbb{R}} dt_1 \, dt_2\  p(t_1)  p(t_2)\left ( \frac{2 + (-1)^{b_n}\operatorname{Tr}\!\left[\left(U_1^\dag U_0  +  U_0^{\dagger}U_1\right) \rho(\theta) \right]}{4} \right)\right)\left[ (-1)^{b_{n}} \right]\\
    & = \int_{\mathbb{R}} \int_{\mathbb{R}} dt_1 \, dt_2\  p(t_1)  p(t_2) \left(\frac{1}{N}\sum_{n=0}^{N-1}\sum_{b_{n} \in \{0, 1\}}\left ( \frac{2 + (-1)^{b_n}\operatorname{Tr}\!\left[\left(U_1^\dag U_0  +  U_0^{\dagger}U_1\right) \rho(\theta) \right]}{4} \right)\left[ (-1)^{b_{n}} \right]\right)\\
    & = \int_{\mathbb{R}} \int_{\mathbb{R}} dt_1 \, dt_2\  p(t_1)  p(t_2) \left( \frac{1}{2}\operatorname{Tr}\!\left[\left(U_1^\dag U_0  +  U_0^{\dagger}U_1\right)\rho(\theta)\right] \right)\\
    & = \int_{\mathbb{R}}\int_{\mathbb{R}} dt_1\ dt_2\ p(t_1) p(t_2) \Bigg(\frac{1}{2}\operatorname{Tr}\Big[\Big(
G_i e^{-iG(\theta)(t_1-t_2)} G_j e^{iG(\theta)(t_1-t_2)}\notag \\
& \qquad+ e^{-iG(\theta)(t_1-t_2)}
G_j e^{iG(\theta)(t_1-t_2)} G_i\Big) \rho(\theta)\Big]\Bigg),
\end{align}
where the third equality follows directly from~\eqref{eq:prim-ub-final}.

\section{Estimating the first term of the Kubo-Mori information matrix elements}

\label{app:algo-first-term-KM}

To begin, we recall from the statement of Theorem~\ref{thm:main-Kubo-Mori} the expression for the $(i, j)$-th element of the Kubo-Mori information matrix:
\begin{equation}
I_{ij}^{\operatorname{KM}}(\theta)=\frac{1}{2}\left \langle\left\{
G_{i},\Phi_{\theta}(G_{j})\right\}\right \rangle_{\rho(\theta)}  -\left\langle G_{i}
\right\rangle _{\rho(\theta)}\left\langle G_{j}\right\rangle _{\rho(\theta
)},\label{eq:KM-info-matrix-entries-app}
\end{equation}
Here, as well, the estimation of the second term of the above equation is relatively straightforward. Therefore, in what follows, we focus on estimating the first term. Consider the following:
\begin{align}
    \frac{1}{2}\left \langle\left\{
G_{i},\Phi_{\theta}(G_{j})\right\}\right \rangle_{\rho(\theta)} & = \frac{1}{2}\int_{\mathbb{R}} dt \ p(t) \operatorname{Tr}\!\left[\left\{G_i,e^{-i G(\theta)t}G_j  e^{i G(\theta)t} \right\}\rho(\theta) \right]\\
 & = \frac{1}{2}\int_{\mathbb{R}} dt \ p(t)\operatorname{Tr}\!\left[\left(G_i e^{-i G(\theta)t}G_j  e^{i G(\theta)t} + e^{-i G(\theta)t}G_j  e^{i G(\theta)t} G_i\right) \rho(\theta) \right]
\end{align}
In a similar vein to our previous approach for estimating the first term of the Fisher-Bures information matrix elements, we employ the primitive presented in Appendix~\ref{app:quant_primitive} here as well. This time, we set $U_1 = G_i e^{-i G(\theta)t}G_j$ and $U_0 = e^{-i G(\theta)t}$, and the quantum circuit of the primitive with these values of $U_0$ and $U_1$ is depicted in Figure~\ref{fig:QBKM-estimator}. To this end, we present an algorithm (Algorithm~\ref{algo:first-term-KM}) for estimating $\frac{1}{2}\left \langle\left\{
G_{i},\Phi_{\theta}(G_{j})\right\}\right \rangle_{\rho(\theta)}$. This algorithm outputs an unbiased estimator $\overline{Y}^{(\operatorname{KM})}$ of the quantity of interest, and the proof of this unbiasedness follows similarly to the analysis of Algorithm~\ref{algo:first-term-FB} shown in Appendix~\ref{app:algo-first-term-FB}.

\begin{algorithm}[H]
\caption{\texorpdfstring{$\mathtt{estimate\_first\_term\_KM}(i, j, \theta, \{G_\ell\}_{\ell=1}^{J}, \varepsilon, \delta)$}{estimate first term}}
\label{algo:first-term-KM}
\begin{algorithmic}[1]
\STATE \textbf{Input:} Indices $i, j \in [J]$, parameter vector $\theta = \left( \theta_{1}, \ldots,  \theta_{J}\right)^{\mathsf{T}} \in \mathbb{R}^{J}$, Gibbs local Hamiltonians $\{G_\ell\}_{\ell=1}^{J}$, precision $\varepsilon > 0$, error probability $\delta \in (0,1)$
\STATE $N \leftarrow \lceil\sfrac{2 \ln(\sfrac{2}{\delta})}{\varepsilon^2}\rceil$
\FOR{$n = 0$ to $N-1$}
\STATE Initialize the control register to $|0\rangle\!\langle 0|$
\STATE Prepare the system register in the state $\rho(\theta)$
\STATE Sample $t$ with probability $p(t)$ (defined in~\eqref{eq-mt:prob-dens})
\STATE Apply the Hadamard gate to the control register
\STATE Apply the following unitaries to the control and system registers:
\STATE \hspace{0.6cm} \textbullet~Controlled-$G_i$: $G_i$ is a local unitary with the control on the control register
\STATE \hspace{0.6cm} \textbullet~$e^{-iG(\theta)t}$: Hamiltonian simulation for time $t$ on the system register
\STATE \hspace{0.6cm} \textbullet~Controlled-$G_j$: $G_j$ is a local unitary with the control on the control register
\STATE Apply the Hadamard gate to the control register
\STATE Measure the control register in the computational basis and store the measurement outcome~$b_n$
\STATE $Y_{n}^{(\operatorname{KM})} \leftarrow (-1)^{b_n}$
\ENDFOR
\RETURN $\overline{Y}^{(\operatorname{KM})} \leftarrow \frac{1}{N}\sum_{n=0}^{N-1}Y_{n}^{(\operatorname{KM})}$
\end{algorithmic}
\end{algorithm}

\section{Proof of Theorem~\ref{thm:main-sld}}

\label{app:proof-SLD}

Recall that the $(k, \ell)$-th element of the SLD operator $L^{(j)}(\theta)$ is given as follows:
\begin{align}
    L_{k\ell}^{(j)}(\theta) = \frac{2\langle k|\partial_{j}\rho(\theta
)|\ell\rangle}{\lambda
_{k}+\lambda_{\ell}}.
\end{align}
Plugging~\eqref{eq:basic-calc-deriv-HC} into the above equation, we find that
\begin{align}
    L_{k\ell}^{(j)}(\theta) & = \frac{2\left(  -\frac{1}{2}\langle k|\Phi_{\theta}(G_{j})|\ell\rangle\left(
\lambda_{k}+\lambda_{\ell}\right)  +\delta_{k,\ell}\lambda_{\ell}\left\langle
G_{j}\right\rangle _{\rho(\theta)}\right) }{\lambda
_{k}+\lambda_{\ell}}\\
& = \frac{ -\langle k|\Phi_{\theta}(G_{j})|\ell\rangle\left(
\lambda_{k}+\lambda_{\ell}\right)}{\lambda
_{k}+\lambda_{\ell}} + \frac{2\delta_{k,\ell}\lambda_{\ell}\left\langle
G_{j}\right\rangle _{\rho(\theta)}}{\lambda
_{k}+\lambda_{\ell}}\\
& = -\langle k|\Phi_{\theta}(G_{j})|\ell\rangle + \frac{2\lambda_{k}\left\langle
G_{j}\right\rangle _{\rho(\theta)}}{2\lambda
_{k}}\delta_{k,\ell}\\
& = -\langle k|\Phi_{\theta}(G_{j})|\ell\rangle + \left\langle
G_{j}\right\rangle _{\rho(\theta)}\delta_{k,\ell}.
\end{align}
This implies that
\begin{align}
    L^{(j)}(\theta) & = \sum_{k, \ell}L_{k\ell}^{(j)}(\theta) |k\rangle\!\langle\ell|\\
    &= \sum_{k, \ell}\left(-\langle k|\Phi_{\theta}(G_{j})|\ell\rangle + \left\langle
G_{j}\right\rangle _{\rho(\theta)}\delta_{k,\ell}\right) |k\rangle\!\langle\ell|\\
& = -\Phi_{\theta}(G_{j}) + \left\langle
G_{j}\right\rangle _{\rho(\theta)} I.
\end{align}

\end{document}